\documentclass[11pt, twoside, sort, a4paper]{article}

\usepackage{enumitem}
\usepackage{geometry}
 \usepackage{nccmath}
 \usepackage{wrapfig}
\usepackage{empheq}
\usepackage{graphicx}
\usepackage{bm}
\usepackage{mathtools}
\newcommand{\sifin}[1]{\color{black}{#1}}

\usepackage{geometry}
\usepackage{empheq}
 \usepackage{mathrsfs}
 \usepackage{tikz}

\newcommand{\e}{\mathrm{e}}

\AtBeginDocument{
\addtolength{\abovedisplayskip}{-0.75ex}
\addtolength{\abovedisplayshortskip}{-0.75ex}
\addtolength{\belowdisplayskip}{-0.75ex}
\addtolength{\belowdisplayshortskip}{-0.75ex}
}
\makeatletter
\def\thm@space@setup{\thm@preskip=3pt
\thm@postskip=3pt}
\makeatother
\usepackage{titlesec}
\titlespacing*{\section}
  {0pt}{0.6\baselineskip}{0.2\baselineskip}
\titlespacing*{\subsection}
  {0pt}{0.6\baselineskip}{0.2\baselineskip}
\titlespacing*{\subsubsection}
  {0pt}{0.6\baselineskip}{0.2\baselineskip}

\newcommand{\subalign}[1]{%
  \vcenter{%
    \Let@ \restore@math@cr \default@tag
    \baselineskip\fontdimen10 \scriptfont\tw@
    \advance\baselineskip\fontdimen12 \scriptfont\tw@
    \lineskip\thr@@\fontdimen8 \scriptfont\thr@@
    \lineskiplimit\lineskip
    \ialign{\hfil$\m@th\scriptstyle##$&$\m@th\scriptstyle{}##$\crcr
      #1\crcr
    }%
  }
}

\usepackage{natbib}
\usepackage{amsmath, amsthm, amssymb}
\usepackage{caption}
\usepackage{subfigure}
\usepackage{tikz}
\usetikzlibrary{matrix, positioning, calc}
\usepackage{multirow, nccmath}
\usepackage{algorithmic}
\usepackage{algorithm}

\usepackage{comment}
\usepackage{color}
\usepackage[titletoc,title]{appendix}

\newcommand{\purple}[1]{{\color{black}{#1}}}

\usepackage[running, mathlines]{lineno}

\newcommand{\EQ}{\begin{equation}}
\newcommand{\EN}{\end{equation}}
\newcommand{\EQS}{\begin{equation*}}
\newcommand{\ENS}{\end{equation*}}
\newcommand{\EQA}{\begin{eqnarray}}
\newcommand{\ENA}{\end{eqnarray}}
\newcommand{\EQAS}{\begin{eqnarray*}}
\newcommand{\ENAS}{\end{eqnarray*}}
\newcommand{\ds}{\displaystyle}

\usepackage{accents}
\newlength{\dhatheight}

\newcommand{\myin}{\scalebox{0.7}{\text{in}}}

\newcommand{\Ebb}{\mathbb{E}}

\newcommand{\Pbb}{\mathbb{P}}
\numberwithin{equation}{section}
\numberwithin{table}{section}
\numberwithin{figure}{section}

\newtheorem{defn}{Definition}
\newtheorem{theorem}{Theorem}
\newtheorem{lemma}{Lemma}
\newtheorem{remark}{Remark}[section]

\newtheorem{proposition}{Proposition}

\numberwithin{defn}{section}
\numberwithin{theorem}{section}
\numberwithin{lemma}{section}
\numberwithin{remark}{section}
\numberwithin{assumption}{section}
\numberwithin{condition}{section}
\numberwithin{property}{section}
\numberwithin{proposition}{section}
\numberwithin{corollary}{section}
\numberwithin{algorithm}{section}

\def\argmax{\mathop{\rm arg\,max}}
\def\argmin{\mathop{\rm arg\,min}}
\def\argsup{\mathop{\rm arg\sup}}
\def\arginf{\mathop{\rm arg\inf}}
\def\r{\right}
\def\l{\left}

\makeatletter
\renewcommand*\env@matrix[1][c]{\hskip -\arraycolsep
  \let\@ifnextchar\new@ifnextchar
  \array{*\c@MaxMatrixCols #1}}
\makeatother
\usepackage{algorithmic}
\usepackage{algorithm}
\usepackage{colortbl}
\usepackage{multirow}
\usepackage{setspace}
\usepackage{amsfonts}
\usepackage{mathrsfs}
\AtBeginDocument{
\addtolength{\abovedisplayskip}{-0.4ex}
\addtolength{\abovedisplayshortskip}{-0.4ex}
\addtolength{\belowdisplayskip}{-0.4ex}
\addtolength{\belowdisplayshortskip}{-0.4ex}
}

\definecolor{lg}{gray}{0.8}

\usepackage[breaklinks]{hyperref}

\definecolor{darkgreen}{rgb}{0.0, 0.5, 0.0}

\makeatletter
\def\thm@space@setup{\thm@preskip=3pt
\thm@postskip=3pt}
\makeatother
\usepackage{titlesec}
\titlespacing*{\section}
  {0pt}{0.5\baselineskip}{0.1\baselineskip}
\titlespacing*{\subsection}
  {0pt}{0.5\baselineskip}{0.1\baselineskip}
\titlespacing*{\subsubsection}
  {0pt}{0.5\baselineskip}{0.1\baselineskip}

\usepackage[normalem]{ulem}
\usepackage{cancel}

\begin{document}

\title{Multi-period Mean-Buffered Probability of Exceedance in \\Defined Contribution Portfolio Optimization}

\author{
Duy-Minh Dang\thanks{School of Mathematics and Physics, The University of Queensland, St Lucia, Brisbane 4072, Australia,
email: \texttt{duyminh.dang@uq.edu.au}
}
\and
Chang Chen \thanks{School of Mathematics and Physics, The University of Queensland, St Lucia, Brisbane 4072, Australia,
email: \texttt{chang.chen1@student.uq.edu.au} 
}
}
\date{\today}
\maketitle
\begin{abstract}
We investigate multi-period mean–risk portfolio optimization for long-horizon Defined Contribution plans, focusing on buffered Probability of Exceedance (bPoE), a more intuitive, dollar-based alternative to Conditional Value-at-Risk (CVaR).
We formulate both pre-commitment and time-consistent Mean–bPoE and Mean–CVaR portfolio optimization problems under realistic investment constraints (e.g.\ no leverage, no short selling) and jump-diffusion dynamics. These formulations are naturally framed as bilevel optimization problems, with an outer search over the shortfall threshold and an inner optimization over rebalancing decisions. We establish an equivalence between the pre-commitment formulations through a one-to-one correspondence of their scalarization optimal sets, while showing that no such equivalence holds in the time-consistent setting.
We develop provably convergent numerical schemes for the value functions associated
with both pre-commitment and time-consistent formulations of these mean-risk control problems.

Using nearly a century of market data, we find that time-consistent Mean–bPoE strategies closely resemble their pre-commitment counterparts. In particular, they maintain alignment with investors' preferences for a minimum acceptable terminal wealth level---unlike time-consistent Mean–CVaR, which often leads to counterintuitive control behaviour.
We further show that bPoE, as a strictly tail-oriented measure, prioritizes guarding against catastrophic shortfalls while allowing meaningful upside exposure—making it especially appealing for long-horizon wealth security.
These findings highlight bPoE's practical advantages for Defined Contribution investment planning.

\vspace{.1in}

\vspace{.1in}

\end{abstract}

\section{Introduction}
\label{sc:intro}
The main objective of this paper is to investigate buffered Probability of Exceedance as a viable alternative risk measure to Conditional Value-at-Risk for multi-period
Defined Contribution portfolio optimization.
\subsection{Motivation}
Saving for and managing wealth throughout retirement remains one of the most critical and challenging financial tasks facing individuals.\footnote{Nobel Laureate William Sharpe famously described this challenge as ``the nastiest,
hardest problem in finance'' \cite{ritholz2017tackling}.}
This challenge is magnified by the global shift from Defined Benefit (DB) to Defined Contribution (DC) pension systems, which places the full burden of investment and longevity risk---the possibility that individuals outlive their retirement savings---on plan participants. In many countries, including Australia, the United States, Canada, and parts of Europe, this shift has led to the emergence of ``full-life cycle'' DC plans that span several decades, covering both the accumulation (working years) and decumulation (retirement) phases.

Compounding these challenges is the rapid aging of the world's population~\cite{un2024datasources}, together with the current era of rising inflation and economic turbulence, all of which underscore the urgency of prudent lifecycle financial planning. Moreover, low levels of financial literacy further exacerbate long-horizon financial decision-making.\footnote{Numerous studies consistently show that, even in advanced economies with well-developed financial markets, many adults---including younger individuals---lack foundational financial knowledge, such as an understanding of investment risk, portfolio diversification, and the implications of long-term market fluctuations \cite{oecd2017g20infe, productivity2018super, klapper2015finlit, lusardi2017ordinary, agnew2013financial}.
In Australia, for example, approximately 30\% of adults exhibit low financial literacy, and about one-quarter lack an understanding of basic financial concepts~\cite[page~21]{productivity2018super}.}
Consequently, there is a clear need for risk measures that combine mathematical rigor with intuitive, readily interpretable insights, enabling retail investors---particularly retirees---and DC plan providers to understand, communicate, and manage downside risk within long-horizon portfolio decision-making.

Over the past few decades, a number of risk measures have been proposed for portfolio optimization, with increasing emphasis on left-tail metrics that capture downside risk.
In this context, investors typically aim to maximize some notion of reward--such as the expectation of terminal wealth or total withdrawals---while minimizing a measure of risk, leading to so-called ``reward–risk'' portfolio optimization.

Among the various risk measures proposed, Conditional Value-at-Risk (CVaR), also known as Expected Shortfall, stands out  as a natural means of capturing left-tail risk in
a portfolio's wealth outcomes. Specifically, it measures the average of the worst $\alpha$-fraction of wealth outcomes, where $\alpha\in(0,1)$ is a confidence level~\cite{RT2000}. Mathematically, CVaR at confidence level $\alpha$ can be expressed through a threshold-based optimization formulation, in which a candidate shortfall threshold variable $W_a$ partitions the wealth distribution into its worst $\alpha$-fraction and the remainder.
This threshold-based formulation is widely employed in portfolio optimization due to its significant computational advantages \cite{RT2000, Mafusalov2018}.
Although CVaR is often defined on losses (making lower values preferable), this paper defines CVaR on terminal wealth so that \mbox{higher values are desirable}.

CVaR has gained considerable popularity in risk management due to its coherent risk properties and its ability to be optimized efficiently using scenario-based linear programs or non-smooth convex optimization methods \cite{RT2000, krokhmal2002portfolio}.
As a result, CVaR-based formulations have become standard in institutional portfolio construction and widely adopted in academic research (see, for example, \cite{ban2018machine, cui2019, miller2017optimal, alexander2006, lwin2017mean, PA2020a} among many other publications). Crucially, CVaR's focus on left-tail risk is especially relevant for individuals in the accumulation phase---who face the possibility of not achieving sufficient retirement savings---as well as for retirees in the decumulation phase---who risk outliving their wealth---making CVaR a suitable risk measure across the entire retirement lifecycle \cite{forsyth2020optimal}.\footnote{Beyond retirement planning, CVaR has also found widespread application in fields such as supply chain management, scheduling and networks, energy systems, and medicine~\cite{filippi2020conditional}.}

Despite CVaR's widespread adoption in institutional settings and academic research, in our experience, practitioners frequently observe that its probability-based formulation can be
opaque to many retail investors, who tend to prefer tangible minimum acceptable terminal wealth levels over an abstract confidence level such as~$\alpha$.
This observation aligns with broader research documenting widespread gaps in financial literacy, as mentioned earlier, with \cite{lusardi2019financial} noting that
``\emph{across countries, financial literacy is at a crisis level,\ldots
individuals have the lowest level of knowledge around the concept of risk, and risk diversification.}''

In light of these considerations, the buffered Probability of Exceedance (bPoE) \cite{bpoe2018, Mafusalov2018} emerges as a risk measure that preserves many of CVaR's key advantages while offering a more intuitive, dollar-based perspective on minimum acceptable terminal wealth. More specifically, bPoE directly specifies a ``disaster level'' $D$ in dollar terms and identifies the minimal probability mass in the left tail whose conditional average equals $D$.
As with CVaR, we define bPoE on terminal wealth, ensuring that a lower bPoE value corresponds to a more favorable outcome from the investor's perspective.
Conceptually, bPoE can be viewed as the inverse of CVaR: whereas CVaR fixes $\alpha$ and calculates the average wealth in the left-$\alpha$ tail of the distribution, bPoE instead fixes the disaster level $D$ and determines the corresponding probability level~\cite{bpoe2018}.
Similar to CVaR, bPoE at a fixed disaster level $D$ can also be expressed through a threshold-based optimization formulation, using a candidate shortfall threshold $W_o$ that partitions the terminal wealth distribution into a left-tail
with conditional average equal to $D$, and the remainder~\cite{bpoe2018}.

To illustrate, consider an investor concerned about their DC account balance falling below
\$100{,}000---a level they regard as the minimum acceptable terminal wealth. In a CVaR framework, the investor would first select a probability level $\alpha$ (e.g.\ 5\%) and then compute the average wealth in the worst $\alpha$-fraction of scenarios---an approach that
appears detached from the practical concern of falling below a specific dollar-based terminal wealth level. By contrast, bPoE begins by fixing \$100{,}000 as the disaster level $D$ and determines the likelihood that terminal wealth falls into a region whose conditional average equals \$100{,}000. This dollar-based framing naturally addresses the core question: ``What is the chance that my retirement wealth falls into a disaster region where the average outcome is only \$100{,}000?''. As a result, bPoE is more intuitive for retail investors and supports clearer communication between DC plan holders and providers.

While CVaR-based frameworks have become standard in portfolio optimization \cite{miller2017optimal, gao2017dynamic, PA2020a}, bPoE remains largely unexplored within this domain---despite its promising theoretical properties \cite{bpoe2018, cvar_bpoe_2019}.\footnote{bPoE has seen applications in logistics, natural-disaster analysis, statistics, stochastic programming, and machine learning.} A notable exception is the single-period analysis in \cite{cvar_bpoe_2019}, which derives closed-form expressions for bPoE under certain probability distributions and demonstrates its use in portfolio selection and density estimation.
However, to the best of our knowledge, bPoE has not been explored in either discrete- or continuous-time multi-period portfolio optimization settings, where notions of time-consistency and pre-commitment (the latter being inherently time-inconsistent)
play a central role in investment decisions. As we shall see, the fixed dollar disaster level~$D$ in bPoE not only aligns more naturally with investors' absolute wealth perspectives, but also offers insights into significant advantages over CVaR in developing time-consistent solutions for multi-period portfolio optimization.

Furthermore, while tail-based reward–risk formulations, such as Mean–CVaR, are well established in theory and increasingly explored in practice, there remains a notable lack of provably convergent numerical methods for computing the value functions of these problems in multi-period settings under realistic market dynamics and investment constraints.
This limitation highlights the need for reliable numerical tools that enable the practical adoption of these strategies.

\subsection{Background}
Broadly speaking, two main approaches have been developed for multi-period mean–risk portfolio optimization: the pre-commitment approach and the time-consistent approach.
Pre-commitment strategies are known to exhibit time-inconsistent behavior \cite{bjork2014, bjork2010general, bjork2016theory, bjork2017}. Concretely, time-inconsistency means that there exists $0 \le t_k < t_l < t_n \le T$, where $T$ is the finite investment horizon, such that the optimal decision for time $t_n$, computed at $t_k$, differs from the optimal decision for the same time $t_n$, but computed at a later time $t_l$.

A classic illustration of this phenomenon is Mean-Variance optimization, where
the variance of the terminal wealth serves as the risk-measure. Because the variance term in the Mean–Variance objective is not separable in the sense of dynamic programming, the Bellman principle cannot be applied, thereby resulting in time-inconsistency \cite{duan2000, DangForsyth2014, zhou2000, miller2017optimal}.
Likewise, CVaR and bPoE are nonlinear risk measures, and hence their associated mean–risk objectives are also non-separable in the sense of dynamic programming.

In contrast, the time-consistent approach enforces an explicit time-consistency constraint that, for any \mbox{$0 \le t_k < t_l < t_n \le T$},  the optimal control for time~$t_n$, computed at~$t_l$, must match the optimal control for the same time $t_n$, but computed at earlier time~$t_k$. Imposing this constraint restores dynamic-programming tractability by making the mean-risk objective satisfy the Bellman principle \cite{basak2010, PMVS2019, PMVS2018timeconsistent, PA2020a}.
However, since time-consistent  strategies can be viewed as pre-commitment strategies with an additional constraint, they are generally not globally optimal when evaluated from time zero. The reader is referred to \cite{vigna2022tail} for a broader discussion of the merits and demerits of time-consistent and pre-commitment policies.

Without enforcing time consistency, multi-period portfolio optimization under the pre-commitment framework generally yields strategies that are globally optimal  when viewed from time zero, but may not remain optimal as time evolves. This has led many researchers
to label pre-commitment strategies as non-implementable.
However, in certain specialized settings, the pre-commitment strategy can be reformulated
as a time-consistent solution under an alternative objective.

For example, in Mean–Variance optimization, the pre-commitment strategy can be recast as a time-consistent solution under an alternative objective that minimizes quadratic shortfall relative to a fixed wealth target, often referred to as the Mean–Variance–induced utility maximization problem~\cite{strub2019}. Likewise, in both continuous-time \cite{miller2017optimal, gao2017dynamic} and multi-period \cite{PA2020a} Mean–CVaR portfolio optimization, although pre-commitment strategies are formally time-inconsistent, they can be equivalently viewed at time zero as linear shortfall policies with a fixed wealth threshold---policies that are time-consistent~\cite{PA2020a}. Therefore, in both cases, the investor has no incentive to deviate from the strategy computed at inception,
rendering the pre-commitment solution effectively implementable.

Interestingly, the literature has observed that imposing time-consistency constraints in Mean-CVaR portfolio optimization can lead to undesirable properties in the resulting optimal controls, even under otherwise realistic conditions. In discrete rebalancing settings with no leverage and no short selling, time-consistent Mean-CVaR strategies have been found to be
wealth--independent in lump-sum investment scenarios---and only weakly dependent otherwise---thus behaving similarly to deterministic strategies and offering little or no improvement over simple constant-weight strategies~\cite{PA2020a, ForsythVetzal2019}.

This phenomenon arises because time-consistent Mean-CVaR policies re-optimize the shortfall threshold $W_a$ at each rebalancing time based on current wealth, which can vary substantially from earlier wealth levels. Thus, the shortfall threshold $W_a$ shifts over time in response to these wide fluctuations. In contrast, as noted earlier, most investors naturally prefer anchoring their investment strategy to a fixed minimum terminal wealth level rather than redefining the shortfall threshold whenever their wealth changes. In fact, \cite{PA2020a} states
\begin{center}
\begin{minipage}{0.7\textwidth}
``\emph{At time zero, we have some idea of what we desire as a minimum final wealth.
Fixing this shortfall target for all $t > 0$ makes intuitive sense.}''
\end{minipage}
\end{center}
Therefore, re-optimizing the shortfall threshold $W_a$ across a wide range of possible wealth levels in time-consistent Mean-CVaR solutions can conflict with these absolute goals,
leading to strategies that are often viewed as counterintuitive from a practical standpoint.
We note that related paradoxes concerning time-consistent Mean–Variance formulations under constraints have also been discussed; see \cite{bensoussan2019paradox} for details.

By contrast, pre-commitment Mean-CVaR formulations, while formally time-inconsistent, adopt a fixed shortfall threshold from inception---thus aligning more closely with the
minimum terminal wealth levels  that investors typically prefer and delivering more practically appealing outcomes.

\subsection{Contributions}
While the prevailing narrative in the literature emphasizes that enforcing time-consistency can lead to counterintuitive control behaviour, we argue that the choice of risk measure is equally critical in shaping time-consistent portfolio strategies. In particular, such behaviour can arise not only from the imposition of time-consistency, but also from the way sub-problems are re-solved at each state and time step, an aspect fundamentally governed by the structure of the risk measure.

Unlike its CVaR counterpart, time-consistent Mean-bPoE---anchored by a fixed disaster level $D$---enforces a constrained re-optimization of the bPoE shortfall threshold $W_o$ at each time step. Specifically,  $W_o$ must define a left-tail region whose conditional average equals $D$, thereby limiting how much the threshold can shift in response to wealth fluctuations. As a result, the policy may remain better aligned with an investor's minimum terminal wealth preferences, thus avoiding the counterintuitive behavior of investment controls  often observed in time-consistent formulations.

In response to these observations, this paper sets out to achieve three primary objectives.
First, we formulate  multi-period Mean–bPoE and Mean–CVaR optimization frameworks under both pre-commitment and time-consistent settings, incorporating realistic investment constraints and modeling assumptions appropriate for long-horizon DC plans.
We then develop and analyze a provably convergent numerical scheme to compute the resulting value functions and optimal controls under these frameworks.
For illustration, we focus specifically on the accumulation phase.
Second, we analyze the mathematical connection between Mean–bPoE and Mean–CVaR strategies,
motivated by the fact that these risk measures are inverse to one another,
and explore how this duality influences their behaviour under both optimization paradigms.
Third, using nearly a century of actual market data, we illustrate the properties of both Mean–bPoE and Mean–CVaR strategies, with particular focus on whether bPoE's minimum terminal wealth perspective can avoid the counterintuitive control behaviour often observed in time-consistent mean–risk formulations.

The main contributions of this paper are as follows.
\begin{itemize}

\item We present both the pre-commitment and time-consistent formulations of multi-period Mean–bPoE and Mean–CVaR portfolio optimization problems using a scalarization approach. For the pre-commitment case, we establish the existence of finite optimal thresholds. Building on this result, we then show that pre-commitment Mean–bPoE and Mean–CVaR formulations can be reformulated as time-consistent target-based portfolio problems, thereby ensuring their implementability.\footnote{The existence of a finite optimal threshold is essential not only to this result, but also to the mathematical analysis and the development of numerical methods, yet a formal proof is often overlooked in the existing literature.}

\item
    We then rigorously establish an equivalence between the pre-commitment Mean–bPoE and Mean–CVaR formulations, in the sense of a one-to-one correspondence between points on their scalarization optimal sets (i.e.\ efficient frontiers), under appropriately calibrated scalarization parameters. 
    {\sifin{For each such pair of corresponding points, there exists a threshold/control pair, not necessarily unique, that is simultaneously optimal in both formulations and attains these points. As a result, the induced distributions of terminal wealth coincide.}}


   In contrast, we demonstrate that this equivalence does not hold in the time-consistent setting: the time-consistent Mean–bPoE and Mean–CVaR optimal controls exhibit fundamentally different behavior with respect to wealth dependence.

\item We develop a unified numerical framework for both pre-commitment and time-consistent multi-period Mean–bPoE and Mean–CVaR formulations, based on monotone numerical integration. To the best of our knowledge, this is the first work to establish convergence of a numerical scheme to the value function for these problems under realistic market dynamics and investment constraints.

\item We conduct a comprehensive numerical comparison of Mean–bPoE and Mean–CVaR optimization results, examining the behavior of optimal investment strategies, terminal wealth distributions, and several key performance metrics---including the mean of terminal wealth, CVaR, bPoE, and the 5th, 50th, and 95th percentiles. All numerical experiments use model parameters calibrated to 88 years of inflation-adjusted long-term U.S.\ market data, enabling realistic conclusions to be drawn.

    {\purple{We find that pre-commitment Mean–bPoE and its CVaR counterpart yield
    virtually identical investment outcomes across all comparison metrics,
    consistent with the theoretical equivalence established earlier.}}
    Consequently, pre-commitment Mean–bPoE can be integrated seamlessly into existing Mean–CVaR workflows, allowing institutions to adopt a dollar-based risk measure without altering broader frameworks or results.

    \item {\purple{In contrast, time-consistent Mean–bPoE and Mean–CVaR strategies exhibit fundamentally different control behaviour and investment outcomes,
        consistent with theoretical findings.}}

        For the same mean of terminal wealth, time-consistent Mean–bPoE outperforms its Mean–CVaR counterpart on nearly all comparison metrics, including CVaR, bPoE, and the 5th and 95th percentiles,  while yielding a noticeably lower median.
        This reflects a key distinction: Mean–CVaR tends to compress the wealth distribution toward the median, whereas Mean--bPoE emphasizes tail performance by prioritizing protection against catastrophic shortfalls while allowing meaningful upside exposure. As a result, bPoE may be especially appealing for investors focused on long-term wealth security rather than distributional tightness.

\end{itemize}
The remainder of the paper is organized as follows. Section~\ref{sec:modelsetup} describes the investment setting and the underlying asset dynamics. Section~\ref{sc:risk_measures} introduces the two risk measures---CVaR and bPoE---and discusses their inverse relationship. In Section~\ref{sc:pareto}, we define the scalarization optimal sets for the pre-commitment formulations and examine their key properties.
Section~\ref{sc:precommitment} establishes the equivalence between pre-commitment Mean–bPoE and Mean–CVaR formulations and discusses their respective implementability under a target-based interpretation. In Section~\ref{sc:time_consistent}, we present the formulations for time-consistent Mean–bPoE and Mean–CVaR, and examine key differences in the behavior of their optimal controls. A unified numerical framework applicable to both settings is developed in Section~\ref{sc:num_methods}. Section~\ref{sc:num} presents and discusses the numerical results. Section~\ref{sc:conclude} concludes the paper and outlines directions for future research.

\section{Modelling}
\label{sec:modelsetup}
Except where noted, our modelling framework---including the discrete rebalancing schedule, jump-diffusion dynamics, and investment constraints---closely follows the discrete-time setting of \cite{PA2020a}.
\subsection{Rebalancing discussion}
\label{sec:rebalancing}
We consider a portfolio consisting of a risky asset and a risk-free asset. In practice, the risky asset may represent a broad market index, while the risk-free asset could be a bank account. This setup is justified by the observation that a diversified portfolio of various risky assets, such as stocks, can often be approximated by a single broad index. Long-term investors typically focus on strategic asset allocation—determining how much to allocate to different asset classes—rather than selecting individual stocks.

We consider a complete filtered probability space $(\mathcal{S}, \mathcal{F}, \{\mathcal{F}_t\}_{0 \le t \le T}, \mathbb{P})$, where $\mathcal{S}$ is the sample space, $\mathcal{F}$ is a $\sigma$-algebra, $\{\mathcal{F}_t\}_{0 \le t \le T}$ is the filtration over a finite investment horizon $T>0$, and $\mathbb{P}$ is the real-world probability measure. Let $S(t)$ and $B(t)$ respectively denote the amounts invested in the risky and risk-free assets at time $t \in [0, T]$. The total wealth at time $t$ is $W(t) = S(t) + B(t)$.
For simplicity, we often write $S_t = S(t)$, $B_t = B(t)$, and $W_t = W(t)$. We also denote by $\{X_t\}$, $t \in [0,T]$, where  $X_t= (S_t, B_t)$, the multi-dimensional controlled underlying process, and let $x = (s,b)$ represent a generic state of the system.

We also let $\mathcal{T}$ be the set of $M$ pre-determined, equally spaced,
rebalancing times in $[0, T]$:\footnote{The assumption of equal spacing is made for simplicity. In practice, industry conventions typically follow a fixed rebalancing schedule, such as semi-annually or yearly adjustments, rather than a mix of different intervals.}
\EQ
\label{eq: T_M}
\mathcal{T} = \{ t_m \mid t_m = m \Delta t, \ m = 0, \dots, M-1 \},
\quad
\Delta t = T/M,
\EN
where $t_0 = 0$ is the inception time of the investment.
{\sifin{At each rebalancing time $t_m \in \mathcal{T}$, the investor first adjusts
the portfolio by incorporating a cashflow $q_m$ and then rebalances the
portfolio.
We focus on the accumulation phase and assume
$q_m \ge 0$ for all $m$, so that all cashflows are non--negative
contributions (``cash injections'').}}
At time $t_M = T$, the portfolio is liquidated (no rebalancing), yielding the terminal wealth $W_T$.

For subsequent use, we define the shorthand notation for instants just before and after
time $t \in [0, T]$ respectively as:
\[
t^- = t - \epsilon,
\quad\text{ and } \quad
t^+ = t + \epsilon,
\quad
\text{ where } \epsilon \to 0^+.
\]
For a generic time-dependent function $f(t)$, we write:
\[
f_{m}^- = \lim_{\epsilon\to 0^+} f(t_m - \epsilon),
\quad
f_{m}^+ = \lim_{\epsilon\to 0^+} f(t_m + \epsilon),
\quad \text{where } t_m \in \mathcal{T}.
\]
As noted in \cite{PA2020a}, DC plan savings are typically held in tax-advantaged accounts, where portfolio rebalancing does not trigger immediate tax liabilities.
Given this, we also assume no taxes in our analysis. In addition, rebalancing occurs infrequently on a fixed schedule, such as annually, reducing trading activity and significantly lowers costs associated with bid-ask spreads, brokerage fees, and market impact; hence, we assume no transaction costs. Given these assumptions, the total wealth at time $t_m^+$ after incorporating the {\sifin{cash injection}} $q_m$ is given by
\EQ
\label{eq:cash}
W_{m}^+ = W_{m}^- + q_m, \quad t_m \in \mathcal{T}.
\EN
For a rebalancing time $t_m \in \mathcal{T}$,
we use $u_{m} (\cdot) \equiv  u(\cdot)$ to denote the rebalancing control which is
the proportion of total wealth allocated to the risky asset.
This proportion depends on both the total wealth $W_{m}^+$
(including the {\sifin{cash injection $q_m$}}) and time $t_m$, i.e.\
\[
u_{m} (\cdot) = u_{m}(W_{m}^+) = u(W_{m}^+, t_m), \quad
\text {where $W_{m}^+$ is given by \eqref{eq:cash}}.
\]
We denote by $\mathcal{Z}$ the set of all admissible rebalancing controls, i.e.\ $u_{m} \in \mathcal{Z}$ for every $t_m \in \mathcal{T}$. The set $\mathcal{Z}$ is typically determined by investment constraints. To enforce no leverage and no shorting, we impose $\mathcal{Z} = [0,1]$.

Suppose that the investor applies the rebalancing control $u_{m} \in \mathcal{Z}$ at $t_m \in \mathcal{T}$, and the system is in state $x = (s, b)$ immediately before rebalancing, i.e.\
at time $t_m^-$, hence, $W_{m}^- = s+b$.  We denote by $X_{m}^+ = (S_{m}^+, B_{m}^+) \equiv (S^+(s, b, u_{m}), B^+(s, b, u_{m}))$ the state of the system immediately after applying $u_{m}$. We then have
\EQ
\label{eq: W=W+q}
\begin{aligned}
        S_{m}^+ &\equiv {\purple{s^+(s, b, u_{m})}} = u_{m} \, W_{m}^+, \quad \text{and} \quad
        B_{m}^+ \equiv {\purple{b^+(s, b, u_{m})}} = \left(1-u_{m})\right)\, W_{m}^+,
        \\
        & \qquad \text{ where by \eqref{eq:cash} } W_{m}^+ = W_{m}^- + q_m = s+b + q_m.
\end{aligned}
\EN
Let $\mathcal{A}$ be the set of admissible controls, defined as follows
\begin{eqnarray}
   \mathcal{A} = \left\{\, \mathcal{U} =   \left\{ u_{m}\right\}_{m = 0, \ldots, M-1} \big|
   u_{m} \in \mathcal{Z}, \text{ for } m = 0, \ldots, M-1 \right\}.\label{eq: P}
\end{eqnarray}
For any rebalancing time $t_m$, we define the subset of controls applicable from $t_m$ onward as
\begin{eqnarray}
    \mathcal{U}_{m} & = & \left\{ \left. u_{m'} \ \right| \ u_{m'} \in \mathcal{Z}, \ m' = m, \ldots, M-1 \right\} \subseteq \mathcal{U}_{0} \equiv \mathcal{U}.\label{eq: P_m}
\end{eqnarray}

\subsection{Underlying dynamics}
\label{sec:dynamics}
In practice, real (inflation-adjusted) returns are more relevant to investors than nominal returns \cite{ForsythVetzal2017, PMVS2021c}. Consequently, we model both the risky and risk-free assets in real terms. All parameters, including the risk-free interest rate, are thus taken to be inflation-adjusted. Given that we focus on relatively long investment horizons (often 20 or 30 years) and that real interest rates tend to be mean-reverting, we assume a constant, continuously compounded real risk-free  interest rate $r$.


Between consecutive rebalancing times, the process $\{B_t\}$ follows
\begin{eqnarray}
    dB_t &=& r\,B_t\,dt,
    \quad
    t \in [\,t_m^+,\,t_{m+1}^-], \quad t_m \in \mathcal{T},
    \label{eq: dBt}
\end{eqnarray}
where $r$ is the constant real risk-free rate.
{\purple{In a discrete setting, the amount invested in the risk-free asset remains constant over $[t_m^+,\,t_{m+1}^-]$ and is updated at $t_{m+1}$ to reflect the interest accrued over $[t_m,\,t_{m+1}]$. Specifically, if the risk-free amount at $t_m^+$ is $B_{t_m^+} = b$, it remains at $b$ throughout $[t_m^+,\,t_{m+1}^-]$ and is updated to $b\,e^{r\,\Delta t}$ at $t_{m+1}$. Rebalancing then occurs immediately after settlement, i.e.\ over $[t_{m+1},\,t_{m+1}^+]$.}}

To capture more realistic behavior of the risky asset, we allow for both diffusion and jump components \cite{li2016optimal}.
We let the random variable $\xi$ be the jump multiplier.
If a jump occurs at time $t$, the amount invested in the risky asset
jumps from $S_{t^-}$ to $S_t = \xi \,S_{t^-}$. We adopt the Kou model \cite{kou01, kou2004}, in which $\ln(\xi)$ follows an asymmetric double-exponential distribution.
The probability density function of $\ln(\xi)$ is given by
\begin{eqnarray}
    \rho(\zeta)
    &=&
    p_{up}\,\eta_1\,e^{-\eta_1\,\zeta}\,\mathbb{I}_{\{\zeta \ge 0\}}
    \;+\;
    (1 - p_{up})\,\eta_2\,e^{\eta_2\,\zeta}\,\mathbb{I}_{\{\zeta < 0\}},
    \quad
    p_{up} \in [0,1],
    \ \eta_1 > 1,
    \ \eta_2 > 0.
    \label{eq: pdf for Kou model}
\end{eqnarray}
Between consecutive rebalancing times, in the absence of active control, the process $\{S_t\}$ evolves according to the jump-diffusion dynamics:
\begin{eqnarray}
    \frac{dS_t}{S_{t^-}}
    &=&
    \bigl(\mu - \lambda\,\kappa\bigr)\,dt
    \;+\;
    \sigma\,dZ_t
    \;+\;
    d\Bigl(\!\sum_{i=1}^{\pi_t}\!(\xi_i - 1)\Bigr),
    \quad
    t \in [\,t_m^+,\,t_{m+1}^-], \quad t_m \in \mathcal{T}.
    \label{eq: dSt}
\end{eqnarray}
Here, $\mu$ and $\sigma$ are the (inflation-adjusted) drift and instantaneous volatility, respectively, and $\{Z_{t}\}_{t \in [0, T]}$ is a standard Brownian motion.
The process $\{\pi_t\}_{0\le t \le T}$ is a Poisson process with a constant finite intensity rate $\lambda\geq 0$. All jump multiplier $\xi_i$ are i.i.d.\ with the same distribution as random variable $\xi$, and $\kappa = \mathbb{E}[\xi - 1]$ is the compensated drift adjustment, where $\mathbb{E}[\cdot]$ is the expectation taken under the real-world measure $\mathbb{P}$. It is further assumed that  $\{Z_{t}\}_{t \in [0, T]}$, $\{\pi_t\}_{0\le t \le T}$, and all $\{\xi_i\}$ are mutually independent. Note that GBM dynamics for $\{S_t\}$ can be recovered from \eqref{eq: dSt} by setting the intensity parameter $\lambda$ to zero.

\section{Risk measures}
\label{sc:risk_measures}
In this and the next sections, we introduce two risk measures, namely CVaR and bPoE, along with their corresponding mean-risk portfolio optimization formulations. For simplicity and clarity, we establish the following notational conventions. A subscript $j \in \{a,o\} $ is used to distinguish quantities related to the CVaR risk measure or CVaR-based formulation ($j = a$) from those corresponding to the bPoE counterparts ($j = o$). Additionally, a superscript  ``$p$'' denotes quantities associated with pre-commitment optimizations,
while a superscript ``$c$'' identifies those related to their time-consistent counterparts.

\subsection{Conditional Value-at-Risk}
Recalling that $W_{T}$ is the random variable representing terminal wealth, we let $\psi(w)$ denote its pdf. For a given confidence level $\alpha$, typically $0.01$ or $0.05$, the CVaR of $W_{T}$ at level $\alpha$ is defined as
\begin{eqnarray}
    \text{CVaR}_{\alpha}\left(W_{T}\right) & = & \tfrac{1}{\alpha}\int_{-\infty}^{\text{VaR}_{\alpha}\left( W_{T} \right)} w \  \psi(w)\,   dw.\label{eq: CVaR_alpha}
\end{eqnarray}
Here, $\text{VaR}_{\alpha}(W_T)$ denotes the Value-at-Risk (VaR) of $W_T$ at confidence level $\alpha$, given by
\begin{eqnarray}
    \text{VaR}_{\alpha}\left(W_{T}\right) & = & \left\{ \left.w \ \right| \ \mathbb{P}\left[ \  W_{T}\le w \ \right] = \alpha  \right\}.\label{eq: VaR_alpha}
\end{eqnarray}
That is,
\EQ
\label{eq: VaR_integral}
\int_{-\infty}^{\text{VaR}_{\alpha}\left( W_{T} \right)}  \psi(w)\,   dw = \alpha.
\EN
We can interpret $\text{VaR}_{\alpha}(W_T)$ as the threshold such that $W_T$ falls below this value with probability $\alpha$~\cite{rock2014}.
Intuitively, given a pre-specified $\alpha$, $\text{CVaR}_{\alpha}(W_T)$ represents the average level of $W_T$ in the worst $\alpha$-fraction of all possible outcomes, i.e.\ in the leftmost $\alpha$-quantile of the distribution of $W_T$.

As noted in \cite{RT2000}, the integral-based definition of $\text{CVaR}_{\alpha}\left(W_{T}\right)$ given in \eqref{eq: CVaR_alpha} often becomes cumbersome when embedded in optimization problems, particularly those involving complex or non-smooth probability distributions of terminal wealth.
To address this, \cite{RT2000} shows that $\text{CVaR}_{\alpha}\left(W_{T}\right)$ can be reformulated as a more computationally tractable optimization problem. The key idea in \cite{RT2000} is to introduce a candidate threshold $W_{a}$ that effectively partitions the probability distribution of $W_T$ into its lower tail and the remainder. Optimizing over this threshold leads to an equivalent formulation for $\text{CVaR}_{\alpha}\left(W_{T}\right)$:
\begin{eqnarray}
    \text{CVaR}_{\alpha}\left(W_{T}\right) & = & \sup_{W_{a}} \mathbb{E} \left[ W_{a} + \tfrac{1}{\alpha}\min\left( W_{T} - W_{a}, 0 \right) \right].
    \label{eq: CVaR_optimization}
\end{eqnarray}
{\purple{We note that the theoretical range of $W_{a}$ coincides with the set of all feasible values for $W_T$, which is $[0, \infty)$ under the no-leverage, no-shorting constrain.}}
We denote by $W_{a}^*$ the optimal threshold, i.e.\
\EQ
\label{eq:W_star_a}
    W_{a}^* =  \argmax_{W_{a}}\,
        \mathbb{E} \left[ W_{a} + \tfrac{1}{\alpha}\min\left( W_{T} - W_{\alpha}, 0 \right) \right].
\EN
Because $W_T$ is treated here as a gain-oriented variable, the worst outcomes lie in the left tail; hence, the formulation \eqref{eq: CVaR_optimization} employs $\sup(\cdot, \cdot)$ and $\min(\cdot,0)$ (rather than $\max(\cdot,0)$, which is more common in loss-oriented setups).
Notably, the optimal threshold $W_{a}^*$ that attains the optimum in \eqref{eq: CVaR_optimization} coincides with $\text{VaR}_{\alpha}(W_{T})$ \cite{RT2000}.

From a computational standpoint, the threshold-based formulation \eqref{eq: CVaR_optimization}
is often more tractable than the integral-based definition \eqref{eq: CVaR_alpha} of $\text{CVaR}_{\alpha}\left(W_{T}\right)$. Therefore, we adopt it for the remainder of this paper.
\subsection{Buffered Probability of Exceedance}
For a given disaster level $D$, we define the bPoE
of the terminal wealth $W_T$ at level $D$ by \cite{bpoe2018}
\begin{eqnarray}
    \text{bPoE}_{D}\left(W_{T}\right) & = & \mathbb{P}\left\{  \left. W_{T}\le w  \ \right| \ \mathbb{E}\left[ \left. W_{T} \ \right| \ W_{T}\le w \right]=D \right\}.\label{eq: bPoE}
\end{eqnarray}
That is, $\text{bPoE}_{D}(W_T)$ measures the probability of the left tail of the distribution of the random variable $W_T$, where the conditional mean of $W_T$ within this tail equals the disaster level~$D$. Hence, for a given disaster level~$D$, $\text{bPoE}_{D}(W_T)$ quantifies how large, i.e.\ how probable, this left tail is relative to $D$.

The formulation \eqref{eq: bPoE} is reminiscent of $\text{CVaR}_{\alpha}(W_T)$, which measures the average level of $W_T$ in the worst $\alpha$-fraction of all possible wealth outcomes. However, while $\text{CVaR}_{\alpha}(W_T)$ pre-specifies this probability~$\alpha$ and identifies the corresponding $\alpha$-quantile of the distribution along with its average, $\text{bPoE}_{D}(W_T)$ instead fixes the quantile mean, namely the disaster level $D$, and determines the probability of the portion of the distribution satisfying this condition. Hence, they can be regarded as inverse risk measures \cite{norton2017soft} (see Remark~\ref{rm:duality_cvar_bpoe}).

In addition, as noted earlier,  both risk measures are defined in terms of terminal wealth $W_T$ rather than losses. Thus, a larger $\text{CVaR}_{\alpha}(W_T)$ indicates a more favorable outcome. However, for $\text{bPoE}_{D}(W_T)$, a smaller value is preferred, as it signifies a lower probability of terminal wealth falling below the disaster level $D$.

Direct computation of \eqref{eq: bPoE} requires evaluating the conditional expectation of the tail of a distribution, which can be cumbersome in practice. However, analogous to CVaR's optimization-based formula, $\text{bPoE}_{D}\left(W_{T}\right)$ also admits an alternative infimum formulation (see \cite{bpoe2018, norton2019}):
\begin{eqnarray}
    \text{bPoE}_{D}\left(W_{T}\right)  = \underset{W_{o}\,>\,D}{\inf}  \ \mathbb{E} \left[ \max\left( 1 - \tfrac{W_{T}-D}{W_{o}-D} , 0 \right) \right], \label{eq: bpoe_optimization}
\end{eqnarray}
noting that $W_{o}$ ranges over $(D, \infty)$, covering all feasible values of $W_T$ above $D$.

Computationally, the bPoE formulation in \eqref{eq: bpoe_optimization} is simpler than evaluating the conditional expectation in \eqref{eq: bPoE}. Like the CVaR expression in \eqref{eq: CVaR_optimization}, it treats $W_o$ as a decision variable, enabling a systematic search for a threshold that yields a tail mean of $D$. This structure makes \eqref{eq: bpoe_optimization} \mbox{well-suited for numerical implementation.}

\begin{remark}[Monotonicity and optimal threshold in bPoE]
\label{rm:bpoE_Wstar}
{\purple{In practice, we often assume that the disaster level satisfies $D < \mathbb{E}[W_T]$,
ensuring that $D$ represents an adverse or undesirable outcome \cite{bpoe2018}.}} Let $W_{o}^*$ be the threshold minimizing the bPoE objective:
\EQ
\label{eq:W_star_o}
    W_{o}^*=  \argmin_{W_{o}\,>\, D}\,
        \mathbb{E} \left[ \max\left( 1 - \tfrac{W_{T}-D}{W_{o}-D} , 0 \right) \right].
\EN
It is shown in \cite{bpoe2018} that $W_o^*$ is indeed the VaR of the terminal wealth $W_T$
at a confidence level that corresponds to the disaster level $D$.
Below, in addition to this identity, we also show that the bPoE objective function is strictly decreasing on $(D,\,W_o^*)$ and strictly increasing on $(W_o^*,\infty)$.

Let $f(W_o)$ denote the bPoE objective function, which can be expressed as the integral
below for $W_o>D$
\EQ
\label{eq:bpoe_f}
f\left(W_o\right)  =\mathbb{E}\left[\max \left(1-\tfrac{W_T-D}{W_o-D}, 0\right)\right]  =\int_{-\infty}^{W_o}\left(1-\tfrac{w-D}{W_o-D}\right) \psi(w)\,  d w,
\EN
where $\psi(w)$ is the pdf of $W_T$. The derivative of $f\left(W_o\right)$ for $W_o>D$ is
$\tfrac{\partial f}{\partial W_o}=\tfrac{1}{\left(W_o-D\right)^2} \int_{-\infty}^{W_o}(w-D) \psi(w) \, dw$. Setting this derivative to zero gives $\int_{-\infty}^{W_o}(w-D)\, \psi(w) \, dw = 0$,
or equivalently,
\EQ
\label{eq:DCVar}
\frac{\int_{-\infty}^{W_o} w \, \psi(w)\,  d w}{\int_{-\infty}^{W_o} \psi(w) \, d w}=D,
\EN
which corresponds the defining relation between $\operatorname{CVaR}_\alpha(W_T)$
and $\operatorname{VaR}_\alpha(W_T)$ at a confidence level $\alpha$, where
$\operatorname{CVaR}_\alpha(W_T) = D$ (see \eqref{eq: CVaR_alpha} and \eqref{eq: VaR_integral}). Thus, the  point $W_o = \operatorname{VaR}_\alpha\left(W_T\right)$
solves the first-order optimality condition \eqref{eq:DCVar}.
Furthermore,  examining the sign of $\tfrac{\partial f}{\partial W_o}$ reveals
\[
\begin{aligned}
 W_o<\operatorname{VaR}_\alpha\left(W_T\right) \Rightarrow \tfrac{\partial f}{\partial W_o}<0
 \quad
 \text { and }
 \quad
 W_o>\operatorname{VaR}_\alpha\left(W_T\right) \Rightarrow \tfrac{\partial f}{\partial W_o}>0.
\end{aligned}
\]
Hence, $f(W_o)$ is  strictly decreasing  in $(D, \mathrm{VaR}_\alpha(W_T))$ and
strictly increasing for $(\mathrm{VaR}_\alpha(W_T),\infty)$,
giving it a ``V‐shaped'' profile centered at $\mathrm{VaR}_\alpha(W_T)$.
As a result, the point $W_o = W_o^* = \mathrm{VaR}_\alpha(W_T)$ is then
the unique minimizer of $f(W_o)$ in $(D,\infty)$, and $\text{bPoE}_D(W_T)$ is thus attained at this  point.
\end{remark}

\begin{remark}[Duality between CVaR and bPoE]
\label{rm:duality_cvar_bpoe}
Recall the bPoE objective function $f(W_o)$ in \eqref{eq:bpoe_f}.
By evaluating $f$ at $W_o = \operatorname{VaR}_\alpha(W_T)$, and using \eqref{eq:DCVar}, we obtain
\EQ
\label{eq:F_Var}
\inf_{W_o > D} f(W_o) = f\left(\operatorname{VaR}_\alpha(W_T)\right) = \alpha.
\EN
This equation reveals a  duality relationship  between the definitions of CVaR \eqref{eq: CVaR_optimization} and bPoE \eqref{eq: bpoe_optimization}.
Concretely, if we fix $\alpha$ and set $D = \operatorname{CVaR}{\alpha}(W_T)$ in the bPoE formulation \eqref{eq: bpoe_optimization}, then by \eqref{eq:F_Var}, choosing $W_o = \operatorname{VaR}{\alpha}(W_T)$ yields a bPoE value of exactly $\alpha$, and this threshold is optimal. Conversely, suppose we are given a disaster level $D$, and the corresponding bPoE value obtained from \eqref{eq: bpoe_optimization} is $\alpha$.
Then, by \eqref{eq:F_Var}, the optimal threshold in \eqref{eq: bpoe_optimization} is
$W_o^* = \operatorname{VaR}\alpha(W_T)$, and moreover,
the disaster level $D$ must be $\operatorname{CVaR}\alpha(W_T)$ due to \eqref{eq:DCVar}. Thus, $\operatorname{CVaR}_\alpha$ and $\operatorname{bPoE}_D$ are inverse risk measures: specifying one uniquely determines \mbox{the other}.
\end{remark}

\section{Pareto optimal points}
\label{sc:pareto}
Recall that $\{X_t\}_{t \in [0,T]}$, where $X_t = (S_t, B_t)$, represents the multi-dimensional controlled underlying process, and that $x = (s,b)$
denote a generic state of the system. For $t_m \in \mathcal{T}$,
we write  $X_{m}^- = X(t_m^-)$ and $X_{m}^+ = X(t_m^+)$.

We begin by examining the notion of Pareto optimality in the pre-commitment setting.
For subsequent use, we denote by $\mathbb{E}_{\mathcal{U}_0}^{X_{0}^+,t_0^+}\left[W_{T}\right]$ the mean of $W_{T}$ under the real-world measure, conditioned on the state $X_{0}^+ = (S_{0}^+, B_{0}^+)$ at time $t_0^+$ (after the cash injection $q_0$),
while using the control $\mathcal{U}_{0}$ over $[t_{0}, T]$.

Following \cite{DangForsythLi2016}, we introduce the concepts of the achievable objective sets, Pareto optimal points, and scalarization optimal sets.
We first address the bPoE risk measure.
\subsection{Mean--bPoE optimal sets}
To evaluate trade-offs under the Mean--bPoE framework, we first define the achievable Mean--bPoE \mbox{objective set}.
\begin{defn}[Achievable Mean--bPoE objective set]
Consider a disaster level $D > 0$, and let $(x_0,0) \equiv (X_{0}^-,t_{0}^-)$ be the initial state, where $x_0=(s_0,b_0)$. For an admissible control $\mathcal{U}_0 \in \mathcal{A}$,
we define the conditional bPoE at disaster level $D$
and the expected value of the terminal wealth $W_T$ as
\EQ
\label{eq:EbPoE}
\begin{aligned}
\text{bPoE}_{\mathcal{U}_0}^{x_0,t_0} &=
\inf_{W_{o} > D} \left\{\Ebb_{\mathcal{U}_0}^{X_{0}^+,t_{0}^+}
 \left[ \max\left( 1 - \tfrac{W_{T}-D}{W_{o}-D} , 0 \right) ~\bigg|~ X_{0}^- = x_0 \right]\right\},
 \\
\text{E}_{\mathcal{U}_0}^{x_0,t_0} &= \Ebb_{\mathcal{U}_0}^{X_{0}^+,t_{0}^+}
    \left[W_{T} ~\big|~ X_{0}^- = x_0 \right].
\end{aligned}
\EN
We define the achievable Mean--bPoE objective set as
\begin{equation} \label{eq: Y_o}
\mathcal{Y}_{o}(D)= \left\{\left( \text{bPoE}_{\mathcal{U}_0}^{x_0,t_0}\left[W_{T}\right],\,
\text{E}_{\mathcal{U}_0}^{x_0,t_0} \right): \mathcal{U}_0
\in \mathcal{A}\right\},
\end{equation}
and let $\overline{\mathcal{Y}}_{o}(D)$ be its closure in $\mathbb{R}^2$.
\end{defn}
\begin{remark}[Boundedness of $\mathcal{Y}_o$]
\label{rm:box}
Under no-leverage/no-short constraints, with $W_0 = s_0 + b_0 \ge 0$ and  $q_m \ge 0$ $(\forall m)$, it follows that $W_T$ is almost surely non-negative.
Hence, for any $\mathcal{U}_0 \in \mathcal{A}$,
we have $\text{E}_{\mathcal{U}_0}^{x_0,t_0} \ge 0$.

Next, we establish a finite upper bound for $\text{E}_{\mathcal{U}_0}^{x_0,t_0}$.
Assuming $\mu > r$, which is typical in real market data \cite{PMVS2019, PMVS2021, PMVS2021c, DangForsyth2014}, fully investing in the risky asset yields the highest expected value of terminal wealth. Let $\widehat{\mathcal{U}}_0(\cdot)$ denote the strategy that fully invests in the risky asset, and define $\mathcal{E}_{\max} = \text{E}_{\widehat{\mathcal{U}}_0(\cdot)}^{x_0,t_0}$. We have
\EQ
\label{eq:E_max}
\mathcal{E}_{\max} = \text{E}_{\widehat{\mathcal{U}}_0(\cdot)}^{x_0,t_0} =  W_0 e^{\mu T} +
  \sum_{m=0}^{M-1} \;q_m \, e^{\,\mu\,(T - m\,\Delta t)} < \infty.
\EN
We can  then have the bound $\text{E}_{\mathcal{U}_0}^{x_0,t_0}
\le
 \text{E}_{\widehat{\mathcal{U}}_0(\cdot)}^{x_0,t_0}  =
 \mathcal{E}_{\max}$.
Also, the bPoE risk measure satisfies
$\text{bPoE}_{\mathcal{U}_0}^{x_0,t_0} \in [0,1]$ by definition.
Therefore, the achievable Mean--bPoE objective set satisfies
\[
\mathcal{Y}_o(D) \subseteq [0,1] \times [0,\mathcal{E}_{\max}],
\quad \mathcal{E}_{\max} \text{ given by \eqref{eq:E_max}.}
\]
\end{remark}
\begin{defn}[Mean--bPoE Pareto optimal points]
\label{defn:P_o}
A point $\left(\mathcal{B}_*,\,  \mathcal{E}_*\right) \in \overline{\mathcal{Y}}_{o}(D)$
is a Pareto (optimal) point if there exists no admissible strategy
$\mathcal{U}_0 \in \mathcal{A}$ such that
\EQ
\label{eq:E_bPoE}
  \text{E}_{\mathcal{U}_0}^{x_0,t_0}  \geq \mathcal{E}_*,  \quad \text { and } \quad
  \text{bPoE}_{\mathcal{U}_0}^{x_0,t_0} \leq \mathcal{B}_*~,
\EN
and at least one of the inequalities in \eqref{eq:E_bPoE} is strict.
We denote the set of all Pareto optimal points for the Mean--bPoE problem by
$\mathcal{P}_o(D)$, so $\mathcal{P}_o(D) \subseteq \overline{\mathcal{Y}}_{o}(D)$.
\end{defn}
Intuitively, the set of all Pareto optimal points, $\mathcal{P}_o$, characterizes the efficient trade-offs between
the mean and bPoE of terminal wealth, ensuring that no further improvement
is possible without sacrificing one objective for the other.
In this sense, these points are efficient, as any point outside $\mathcal{P}_o$
is dominated by an alternative achievable Mean--bPoE outcome that provides at least as
much expected terminal wealth while also achieving a smaller probability
of wealth falling below the disaster level $D$, or vice versa.

Although the above definitions are intuitive, determining the points in $\mathcal{P}_o$
requires solving a challenging multi-objective optimization problem involving two conflicting
criteria. A standard scalarization method can be used to transform this into a single-objective
optimization problem.
Specifically, for a given scalarization parameter $\gamma > 0$,
we denote by  $\mathcal{S}_{o}(D, \gamma)$ the set of Mean--bPoE scalarization optimal points
corresponding to $\gamma$ for a given disaster level $D$. This set is defined by
\EQ
\label{eq:S_o_gamma}
\mathcal{S}_{o}(D, \gamma) = \big\{
\left(\mathcal{B}_*, \, \mathcal{E}_*\right)\in \overline{\mathcal{Y}}_{o}(D) ~\big|~
\gamma \mathcal{B}_* -  \mathcal{E}_* = \inf_{\left(\mathcal{B}, \,\mathcal{E}\right) \in
\mathcal{Y}_{o}(D)} \left(\gamma \mathcal{B} -  \mathcal{E}\right)
\big\}.
\EN
We then define the Mean--bPoE scalarization optimal set, denote by $\mathcal{S}_{o}(D)$,
as follows
\EQ
\label{eq:S_o}
\mathcal{S}_{o}(D) = \bigcup_{\gamma > 0} \mathcal{S}_{o}(D, \gamma).
\EN
Mathematically, the scalarization parameter $\gamma$ is simply a device for converting
the multi-objective Mean--bPoE problem into a single-objective optimization.
However, from a financial perspective, $\gamma$ naturally reflects
the investor's preference (or aversion) to bPoE risk.
As $\gamma \searrow 0$, the investor effectively ignores bPoE and prioritizes maximizing expected terminal wealth.
Conversely, as $\gamma \to \infty$,  bPoE is heavily penalized, and the investor's
preference shifts to minimizing the probability of wealth falling below the
disaster level $D$.

It is well known that the set of all Mean--bPoE Pareto optimal points
$\mathcal{P}_o(D)$, given in Definition~\ref{defn:P_o}, and the set of
Mean--bPoE scalarization optimal points $\mathcal{S}_{o}(D)$, defined in \eqref{eq:S_o},
satisfy the general relation
$\mathcal{S}_{o}(D) \subseteq \mathcal{P}_o(D)$.
However, the converse does not necessarily hold
if the achievable Mean--bPoE objective set $\mathcal{Y}_{o}(D)$ is not convex \cite{miettinen1999nonlinear, floudas2008encyclopedia}.
Following \cite{DangForsythLi2016}, we restrict our attention to determining $\mathcal{S}_{o}(D)$, hereafter referred to as the Mean-bPoE efficient frontier.

Next, we establish that the Mean--bPoE scalarization optimal set $\mathcal{S}_{o}(D, \gamma)$
is non-empty.
\begin{lemma}[Non-emptiness of $\mathcal{S}_{o}(D, \gamma)$]
\label{lem:non_emptiness}
Consider a disaster level $D>0$. For any $\gamma > 0$, the Mean--bPoE scalarization optimal set
$\mathcal{S}_{o}(D, \gamma)$ is non-empty, i.e.\
$\exists ( \mathcal{B}', \mathcal{E}') \in \overline{\mathcal{Y}}_{o}(D)$
such that
$\ds \gamma\,\mathcal{B}' - \mathcal{E}' =
\inf_{(\mathcal{B}, \mathcal{E}) \in
\mathcal{Y}_{o}(D)} \left(\gamma\,\mathcal{B} - \mathcal{E}\right)$.
\end{lemma}
\begin{proof}
From Remark~\ref{rm:box}, we know
$\overline{\mathcal{Y}}_{o}(D) \subseteq [0,1]\times[0,\mathcal{E}_{\max}]$
is a compact set.
The map $(\mathcal{B},\mathcal{E}) \mapsto \gamma\,\mathcal{B} - \mathcal{E}$
is continuous on $\overline{\mathcal{Y}}_{o}(D)$, so by the extreme value theorem,  it attains its infimum there.
Any point $(\mathcal{B}',\mathcal{E}')\in \overline{\mathcal{Y}}_{o}(D)$
that realizes this infimum belongs to $\mathcal{S}_{o}(D, \gamma)$.
Hence, $\mathcal{S}_{o}(D, \gamma)$ is non-empty for all $\gamma> 0$.
\end{proof}

\subsection{Mean--CVaR optimal sets}
We now outline the Mean--CVaR framework, mirroring the structure of the Mean--bPoE formulation.
Here, the investor seeks to balance the trade-off between the mean and CVaR of terminal wealth.
To start with, we formally define the achievable Mean--CVaR objective set.

\begin{defn}[Achievable Mean--CVaR objective set]
\label{defn:Y_a}
Consider a confidence level $\alpha \in (0,1)$, and let $(x_0,0) \equiv (X_{0}^-,t_{0}^-)$
be the initial state, where $x_0=(s_0,b_0)$. For an admissible control
$\mathcal{U}_0$ in $\mathcal{A}$, we define
\EQ
\label{eq:ECVaR}
\begin{aligned}
  \text{CVaR}_{\mathcal{U}_0}^{x_0,t_0} &=
  \sup_{W_{a} \,\in\,\mathbb{R}}
  \Ebb_{\mathcal{U}_0}^{X_{0}^+,t_{0}^+}
  \!\Bigl[
      W_{a} ~+~
      \tfrac{1}{\alpha}\,\min\!\bigl(W_{T} - W_{a},\,0\bigr)
    ~\big|\,
    X_{0}^- = x_0
  \Bigr],
   \\
  \text{E}_{\mathcal{U}_0}^{x_0,t_0}&=
  \Ebb_{\mathcal{U}_0}^{X_{0}^+,t_{0}^+}\!\bigl[ W_{T}\,\big|\,
  X_{0}^- = x_0 \bigr].
\end{aligned}
\EN
The achievable Mean--CVaR objective set is then defined as
\begin{equation}
\label{eq:Y_a}
\mathcal{Y}_{a}(\alpha) =
\Bigl\{\,\bigl(\text{CVaR}_{\mathcal{U}_0}^{x_0,t_0},~
\text{E}_{\mathcal{U}_0}^{x_0,t_0}\bigr)
~:~
\mathcal{U}_0\in \mathcal{A}\Bigr\},
\end{equation}
and let $\overline{\mathcal{Y}}_{a}(\alpha)$ be its closure in $\mathbb{R}^2$.
\end{defn}
\begin{remark}[Boundedness of $\mathcal{Y}_a(\alpha)$]
\label{rm:box_cvar}
We now establish a finite upper bound for
$\text{CVaR}_{\mathcal{U}_0}^{x_0,t_0}$
under no-leverage/no-short constraints.
Let $W_a^*\left(\mathcal{U}_0\right)$ be the optimal threshold
in the equivalent CVaR definition \eqref{eq: CVaR_optimization}
and let $g\bigl(w;\mathcal{U}_0\bigr)$ denote the pdf of $W_T$
under an admissible control $\mathcal{U}_0 \in \mathcal{A}$.
Recalling that $W_T \ge 0$ almost surely,
we have
\begin{eqnarray}
\text{CVaR}_{\mathcal{U}_0}^{x_0,t_0}
= \tfrac{1}{\alpha} \int_0^{W_a^*(\cdot)} w\,  g\bigl(w;\mathcal{U}_0\bigr) dw
\le \tfrac{1}{\alpha} \int_0^{\infty} w\, g\bigl(w;\mathcal{U}_0\bigr) dw
= \tfrac{\text{E}_{\mathcal{U}_0}^{x_0,t_0}}{\alpha}
\overset{\text{(i)}}{\le}
\tfrac{\mathcal{E}_{\max}}{\alpha}.
\end{eqnarray}
Here, (i) follows from Remark~\ref{rm:box}, and
 $\mathcal{E}_{\max}$ is a finite constant given in \eqref{eq:E_max}.
Therefore, the achievable Mean--CVaR objective set satisfies
\[
\mathcal{Y}_a(\alpha) \subseteq [0, \mathcal{E}_{\max}/\alpha] \times [0,\mathcal{E}_{\max}],
\quad \mathcal{E}_{\max} \text{ given by \eqref{eq:E_max}}.
\]
\end{remark}
\begin{defn}[Mean--CVaR Pareto optimal points]
\label{defn:P_a}
A point $(\mathcal{C}_*, \mathcal{E}_*) \in \overline{\mathcal{Y}}_{a}(\alpha)$
is a \emph{Pareto optimal point} if there is no admissible strategy
$\mathcal{U}_0\in \mathcal{A}$ with
\[
  \text{CVaR}_{\mathcal{U}_0}^{x_0,t_0} \,\ge\, \mathcal{C}_*
  \quad\text{and}\quad
  \text{E}_{\mathcal{U}_0}^{x_0,t_0} \,\ge\, \mathcal{E}_*,
\]
and at least one of these inequalities is strict.
We denote the set of all Pareto optimal points by
$\mathcal{P}_{a}(\alpha)\subseteq \overline{\mathcal{Y}}_{a}(\alpha)$.
\end{defn}
We also employ a scalarization approach to transform the Mean--CVaR Pareto optimization problem, which inherently requires a multi-objective framework, into a single-objective problem using a scalarization parameter $\gamma > 0$.
Formally, we define the Mean--CVaR scalarization optimal set
for a given $\gamma>0$ as follows.
\begin{equation}
\label{eq:S_a_gamma}
\mathcal{S}_{a}(\alpha, \gamma)
~=~
\Bigl\{
(\mathcal{C}_*,\mathcal{E}_*)\in \overline{\mathcal{Y}}_{a}(\alpha)
~\big|\,
\gamma\, \mathcal{C}_* + \mathcal{E}_*
~=~
\sup_{(\mathcal{C},\mathcal{E})\in \mathcal{Y}_a(\alpha)}
~\bigl(\gamma\,\mathcal{C} + \mathcal{E}\bigr)
\Bigr\}.
\end{equation}
The Mean--CVaR scalarization optimal set is then
\begin{equation}
\label{eq:S_a}
\mathcal{S}_{a}(\alpha)
~=~
\bigcup_{\gamma>0}\,\mathcal{S}_{a}(\alpha,\gamma).
\end{equation}
Just as in the Mean--bPoE framework, we have
$\mathcal{S}_{a}(\alpha) \subseteq \mathcal{P}_{a}(\alpha)$, with equality
if $\mathcal{Y}_a(\alpha)$ is convex \cite{miettinen1999nonlinear, floudas2008encyclopedia}.
We focus on determining $\mathcal{S}_{a}(\alpha, \gamma)$, and hence $\mathcal{S}_a(\alpha)$,
hereinafter referred to as the Mean-CVaR efficient frontier.
\begin{lemma}[Non-emptiness of $\mathcal{S}_{a}(\alpha, \gamma)$]
\label{lem:non_emptiness_cvar}
Consider a confidence level $\alpha \in (0, 1)$.
For any $\gamma > 0$, the Mean--CVaR scalarization optimal set
$\mathcal{S}_{a}(\alpha, \gamma)$ is non-empty.
\end{lemma}
\begin{proof}
This follows immediately from the compactness of
$\overline{\mathcal{Y}}_{a}(\alpha)$ (see Remark~\ref{rm:box_cvar})
and the continuity of the map
$(\mathcal{C},\mathcal{E}) \mapsto \gamma\,\mathcal{C} - \mathcal{E}$.
\end{proof}

\section{Pre-commitment Mean--bPoE and Mean--CVaR}
\label{sc:precommitment}
We now consider the problem of determining points in the scalarization optimal sets
$\mathcal{S}_{o}(D, \gamma)$ (for Mean--bPoE, defined in \eqref{eq:S_o_gamma}) and $\mathcal{S}_{a}(\alpha, \gamma)$ (for Mean--CVaR, defined in \eqref{eq:S_a_gamma})
in a form that can be solved using stochastic optimal control techniques.
We begin with $\mathcal{S}_{o}(D, \gamma)$.

\subsection{Pre-commitment Mean--bPoE}
Using the definitions in \eqref{eq:EbPoE}, we recast \eqref{eq:S_o_gamma}
as a control problem involving system dynamics, {\sifin{cash injection}}, and rebalancing constraints.
Specifically, for a given disaster level $D$ and scalarization parameter \mbox{$\gamma > 0$},
the pre-commitment Mean--bPoE problem $PCMo_{t_0}\bigl(D, \gamma\bigr)$
is defined in terms of the value function $V^p_o(s_0,b_0,t_{0}^-)$.
Its formulation is as follows:
\begin{equation}
\label{eq: PCMb}
PCMo_{t_{0}}\left(D, \gamma\right):  V^p_o\bigl(s_0,b_0,t_{0}^-\bigr) \coloneqq \inf_{\mathcal{U}_0 \in \mathcal{A}} \bigg\{ \inf_{W_{o}^p~>~D}  \Ebb_{\mathcal{U}_0}^{X_{0}^+,t_{0}^+}\bigg[ \gamma \max\bigg( 1 - \tfrac{W_{T}-D}{W_{o}^p-D} , 0 \bigg) \!-\! W_{T} \bigg| X_0^- = (s_0, b_0)  \bigg]\bigg\}
\end{equation}
{\sifin{
\begin{equation}
\text{subject to:}~
\begin{cases}
~~\text{(dynamics)} \quad~~~~\left(S_{t}, B_{t}\right) \text{ evolves via }
\eqref{eq: dBt}-\eqref{eq: dSt}
& t \notin \mathcal{T}, \\
\left.\begin{array}{ll}
\text{(cash injection)} &W_k^+ = W_k^- + q_k, \\
\text{(rebalancing)} &X_k^+ = (S_k^+, B_k^+),  \text{ where }
\\
&S_k^+ = u_{k} W_k^+,\quad B_k^+ = \left(1-u_{k}\right)W_k^+, \\
&u_{k}\in \mathcal{Z} = [0, 1],
\quad k = 0,\dots, M-1,
\end{array}\right\} & t_k \in \mathcal{T}.
    \end{cases}\label{eq: PCMb constraints}
\end{equation}}}
We denote by $\mathcal{U}_{0, o}^{p*} = \left\{u_{0, o}^{p\ast}, u_{1, o}^{p\ast}, \ldots, u_{M-1, o}^{p\ast}\right\} $ the optimal control of problem $PCMo_{t_{0}}\left(D, \gamma\right)$.

Following~\cite{PA2020a, strub2019, miller2017optimal}, we interchange the infima in~\eqref{eq: PCMb}, yielding a more computationally tractable form
\EQ
V^p_o\left(s_0,b_0,t_{0}^-\right) =
\inf_{W_{o}^p~>~D}\bigg\{ \inf_{\mathcal{U}_0 \in \mathcal{A}} \Ebb_{\mathcal{U}_0}^{X_{0}^+,t_{0}^+}\left[ \gamma \max\left( 1 - \frac{W_{T}-D}{W_{o}^p-D} , 0 \right) - W_{T}  \bigg| X_0^- = (s_0, b_0) \right] \bigg\}.\label{eq: PCMb_new}
\EN
In Lemma~\ref{lem:infF} below, we establish the existence of a finite optimal threshold $W_o^{p*}$ in $(D,\infty)$, which is essential for both the theoretical analysis and the development of numerical methods. We begin with a continuity result for the inner optimization in the pre-commitment Mean–bPoE problem
\begin{proposition}[Continuity of the inner optimization function]
\label{prop:F_continuous}
Fix $\gamma > 0$, and suppose the initial state $X_0^- = (s_0, b_0)$ and the disaster level $D$ are given. Define $F(W_{o}^p) \equiv F(W_{o}^p; (s_0, b_0))$ as the inner infimum of the value function $V^p_o(s_0,b_0,t_0^-)$
in \eqref{eq: PCMb_new}, i.e.\
\EQ
\label{eq: FWo}
F(W_{o}^p)
= \inf_{\mathcal{U}_0\in \mathcal{A}} \Ebb_{\mathcal{U}_0}^{X_{0}^+,t_{0}^+}\bigg[ \gamma \max\left( 1 - \frac{W_{T}-D}{W_{o}^p-D} , 0 \right) - W_{T}  \bigg| X_0^- = (s_0, b_0) \bigg],~
W_{o}^p > D, \text{ subject to \eqref{eq: PCMb constraints}}.
\EN
Then, the function $F(W_{o}^p)$ is continuous in $W_{o}^p$ for all $W_{o}^p > D$.
\end{proposition}
\noindent A  proof of Proposition~\ref{prop:F_continuous} is given in Appendix~\ref{app:F_continuous}.
\begin{lemma}[Existence of a finite optimal threshold $W_o^{p*}(s_0, b_0)$]
\label{lem:infF}
Let $ F(W_{o}^p; (s_0, b_0))$ be defined on $(D,\infty)$ as in \eqref{eq: FWo}.
Then, $\ds \inf_{W_o^p > D} F(W_{o}^p; (s_0, b_0))$ is finite and is attained at
some $W_o^{p*}(s_0, b_0)\in(D,\infty)$.
\end{lemma}

\noindent A  proof of Lemma~\ref{lem:infF} is given in Appendix~\ref{app:infF}.
\subsubsection{An equivalent time-consistent problem}
We now show that the pre-commitment Mean--bPoE portfolio optimization can be
reformulate as a time-consistent target-based portfolio problem.
To this end, for a fixed $\widehat{W_o}> D$, we define a time-consistent portfolio optimization problem induced by pre-commitment Mean--bPoE formulation.
This problem is denoted by $TCEo_{t_{m}}\!\bigl(D,\gamma;\,\widehat{W_o}\bigr)$
and is characterized by the value function
$\widehat{Q}^c_o(s,b,t_{m}; \widehat{W_o})$ as follows.
\EQ
\label{eq: TCMb*}
    TCEo_{t_{m}}(\cdot):
    \begin{cases}
       \ds \widehat{Q}^c_o(s,b,t_{m}; \widehat{W_o}) \!\coloneqq\! \inf_{\mathcal{U}_{m}}
    \bigg\{\!\Ebb_{\mathcal{U}_{m}}^{X_m^+,t_{m}^+}\bigg[
    \max(\widehat{W_o}-W_T,\,0)
  \!-\! (\widehat{W_o}-D)\,\tfrac{W_T}{\gamma}
  \bigg| X_m^- = (s, b)\bigg] \bigg\},
  \\
 \text{subject to system dynamics, cash injection, and rebalancing constraints in  \eqref{eq: PCMb constraints}}.
\end{cases}
\EN
In the $TCEo_{t_m}(\cdot;\,\widehat{W_o})$ problem, the terminal objective
$
\max(\widehat{W_o}-W_T,\,0)
  - (\widehat{W_o}-D)\,\tfrac{W_T}{\gamma}
$
depends only on the terminal wealth $W_T$ and it remains fixed throughout the investment horizon.
The expression is piecewise linear with a kink at the target level $\widehat{W_o}$,
giving it a clear target-based character: the first term, $\max(\widehat{W_o} - W_T,\, 0)$, penalizes terminal wealth shortfalls relative to the target $\widehat{W_o}$,
while the second term is a negative linear function of
$W_T$ that rewards higher wealth in the context of a minimization problem.
Together these two components create a trade-off that is typical of target-based
portfolio control: they strongly discourage catastrophic shortfalls while
preserving an incentive to exceed the target.


Since the objective is fixed at each rebalancing time $t_m$, with no re-optimization of risk parameters, dynamic programming can be applied over the control sequence $\mathcal{U}_m$, yielding a time-consistent policy. Such policies, which carry no incentive to deviate over time, are referred to as \emph{implementable} in the \mbox{literature \cite{PA2020a, strub2019}}.

Although $TCEo_{t_m}(\cdot;\,\widehat{W_o})$ yields a time-consistent control policy in this sense, it differs from fully time-consistent mean–risk formulations, which explicitly impose time-consistency constraints across the control sequence (see Section~\ref{sc:time_consistent}). Following the literature, we therefore refer to $TCEo_{t_m}(\cdot;\,\widehat{W_o})$ as a Mean--bPoE induced time-consistent optimization \cite{strub2019}.

We now establish the equivalence between the pre-commitment problem $PCMo_{t_{0}}(D, \gamma)$ and the induced time-consistent formulation $TCEo_{t_{m}}\bigl(D,\gamma;\widehat{W_o}\bigr)$, for an appropriate choice of $\widehat{W_o}$.
Specifically, define $W_o^{p*}(s_0, b_0)$ as the optimal threshold obtained from solving the pre-commitment problem at time $t_0$:
\EQ
\label{eq:W_op_W0}
W_o^{p*}(s_0, b_0) \in \argmin_{\,W_o^p > D} \left\{
\inf_{\mathcal{U}_0\in \mathcal{A}} \Ebb_{\mathcal{U}_0}^{X_{0}^+,t_{0}^+}\bigg[ \gamma \max\left( 1 - \frac{W_{T}-D}{W_{o}^p-D} , 0 \right) - W_{T}  \bigg| X_0^- = (s_0, b_0) \bigg]
\right\}.
\EN
It follows from Lemma~\ref{lem:infF} that such a $W_o^{p*}(s_0, b_0)$ exists and is finite.
The proposition below establishes that when $\widehat{W_o} = W_o^{p*}(s_0, b_0)$,
 $PCMo_{t_0}(D, \gamma)$  is equivalent to $TCEo_{t_m}(D, \gamma;\, W_o^{p*}(s_0, b_0))$.

\begin{proposition}[Equivalent time-consistent problem]
\label{prop:equiv}
Let $\mathcal{U}_{0, o}^{p*}$ be the optimal control of the pre-commitment
problem $PCMo_{t_{0}}\left(D, \gamma\right)$, defined in \eqref{eq: PCMb}-\eqref{eq: PCMb constraints}, obtained by solving the value function $V^p_o(s_0,b_0,t_{0}^-)$ with initial state $X_0^- = (s_0,b_0)$. Then $\mathcal{U}_{0, o}^{p*}$ is also the optimal time-consistent  control for the induced problem $TCEo_{t_{m}}\left(D,\gamma; W_o^{p*}(s_0, b_0)\right)$ defined in \eqref{eq: TCMb*}, where $W_o^{p*}(s_0, b_0)$ is given by \eqref{eq:W_op_W0}.
\end{proposition}
For a proof of Proposition~\ref{prop:equiv}, see Appendix~\ref{app:equiv}.
\begin{remark}[Pre-commitment Mean--bPoE implementability]
\label{rm:tceq_equivalence}
Among the family of induced time-consistent problems $TCEo_{t_m}(D, \gamma;, \widehat{W_o})$ for varying $\widehat{W_o} > D$, the case $\widehat{W_o} = W_o^{p*}(s_0,b_0)$ yields equivalence with the pre-commitment formulation $PCMo_{t_0}(D,\gamma)$. We refer to this particular case, $TCEo_{t_m}(D, \gamma;, W_o^{p*}(s_0,b_0))$, as the Mean--bPoE induced time-consistent \emph{equivalent} optimization. Consequently, the optimal control obtained from solving $PCMo_{t_0}(D,\gamma)$ is also a time-consistent strategy for the associated target-based single-objective problem, and is therefore implementable across the entire investment  horizon.

\end{remark}

\subsection{Pre-commitment Mean--CVaR}
Similarly to the pre-commitment Mean--bPoE case, using the definitions in \eqref{eq:ECVaR}, we now express the scalarization formulation \eqref{eq:S_a_gamma} as a control problem involving both system dynamics and rebalancing constraints.
For a given confidence level $\alpha$ and scalarization parameter $\gamma > 0$,
the pre-commitment Mean--CVaR problem $PCMa_{t_{0}}\left(\alpha, \gamma\right)$
is given in terms of the value function $V^p_a(s_0,b_0,t_{0}^-)$:
\begin{equation}
\label{eq: PCMC}
    PCMa_{t_{0}}\left(\alpha, \gamma\right): V_a^p(s_0,b_0, t_{0}^-) \!\coloneqq \! \sup_{\mathcal{U}_0 \in \mathcal{A}}\bigg\{ \sup_{W_a^p} \Ebb_{\mathcal{U}_0}^{X_{0}^+,t_{0}^+}
    \big[\gamma \big(W_ a^p+\tfrac{1}{\alpha}\min\big( W_{T} - W_a^p, 0 \big)\big) + W_{T}
    \big| X_0^- \!=\! (s_0,b_0)\big]  \bigg\}\\
\end{equation}
{\sifin{\begin{equation}
\label{eq: PCMC constraints}
\text{subject to the system dynamics, cash injection, and rebalancing constraints in \eqref{eq: PCMb constraints}.}
\end{equation}
}}\noindent We denote by $\mathcal{U}_{0, a}^{p*} = \left\{u_{0, a}^{p\ast}, u_{1, a}^{p\ast}, \ldots, u_{M-1, a}^{p\ast}\right\} $ the optimal control of problem $PCMa_{t_{0}}\left(\alpha, \gamma\right)$. As in the Mean--bPoE case, we interchange the order of optimization in \eqref{eq: PCMC}.
The value function $V_a^p(s_0,b_0,t_0^-)$ can be equivalently written as
\EQ
V^p_a\left(s_0,b_0,t_{0}^-\right) =
\sup_{W_a^p}\bigg\{ \sup_{\mathcal{U}_0 \in \mathcal{A}} \Ebb_{\mathcal{U}_0}^{X_{0}^+,t_{0}^+}
    \left[\gamma \left(W_a^p+\frac{1}{\alpha}\min\left( W_{T} - W_a^p, 0 \right)\right) + W_{T} \bigg| X_0^- = (s_0, b_0) \right] \bigg\}\label{eq: PCMa_new}
\EN
Although prior work, such as~\cite{PA2020a}, defines the threshold using the upper semicontinuous envelope of the value function, no formal existence proof is given, and uniqueness is implicitly assumed. For completeness, we establish existence in the lemma below, noting that uniqueness is not guaranteed.
\begin{lemma}[Existence of a finite optimal threshold for pre-commitment Mean--CVaR]
\label{lem:CVaR_threshold_existence}
Fix $\gamma > 0$, and let the initial state $X_0^- = (s_0, b_0)$ and confidence level $\alpha \in (0,1)$ be given.
Define $F_a(W_a^p)$ as the inner supremum of the pre-commitment Mean--CVaR value function
$V_a^p(s_0,b_0,t_0^-)$ in \eqref{eq: PCMb_new}, i.e.\
\[
  F_a(W_a^p) \;=\;
  \sup_{\mathcal{U}_0\in \mathcal{A}}
  \,\Ebb_{\mathcal{U}_0}^{\,X_0^+,\,t_0^+}\!\Bigl[
    \gamma \Bigl( W_a^p + \tfrac{1}{\alpha}\,\min(W_T - W_a^p,\,0)\Bigr)
    \;+\; W_T
    \;\Big|\; X_0^-=(s_0,b_0)
  \Bigr],
  \quad
  W_a^p \ge 0.
\]
Then, $\sup_{\,W_a^p \,\ge 0} \,F_a(W_a^p)$ is finite and is attained by
some $W_a^{p*}(s_0,b_0)\in[0,\infty)$.
\end{lemma}
The argument parallels the Mean–bPoE case (Lemma~\ref{lem:infF}) with full details given in Appendix~\ref{app:CVaR_threshold_existence}.
\begin{remark}[Pre-commitment Mean–CVaR implementability] \label{rm:cvar_time_consistent}
Given the existence of an optimal threshold $W_a^{p*}$ in the pre-commitment Mean--CVaR formulation (Lemma~\ref{lem:CVaR_threshold_existence}), one can construct a corresponding time-consistent optimization problem with a fixed, target-based single objective—mirroring the approach used in the Mean--bPoE setting. As a result, the optimal control from $PCMa_{t_{0}}\left(\alpha, \gamma\right)$ defines a time-consistent and implementable strategy. For further discussion, also see \cite{PA2020a}.
\end{remark}

\subsection{Equivalence between pre-commitment Mean--bPoE and  Mean--CVaR}
\label{ssc:equiv}
In  the analysis, for brevity,  we will occasionally write $\Ebb_{\mathcal{U}_0}[\cdot]$ in place of the full notation $\Ebb_{\mathcal{U}_0}^{X_0^+, t_0^+}[\cdot ~| \cdot]$.
To distinguish between the scalarization parameter values used in the pre-commitment Mean–bPoE and Mean–CVaR formulations, we use $\gamma_o$ and $\gamma_a$, respectively, following the notational convention adopted earlier.

We now formalize the equivalence between the pre-commitment Mean–CVaR and Mean–bPoE scalarization optimal sets by establishing a one-to-one correspondence between
$\mathcal{S}_a(\alpha, \gamma_a)$ and $\mathcal{S}_o(D, \gamma_o)$
under appropriately calibrated parameters.
{\sifin{Moreover, we will show that, for each such pair of corresponding points,
there exists a control/threshold pair, not necessarily unique, that is
optimal in both formulations and attains these points.}}
Each direction of this equivalence is captured in Lemmas~\ref{lem: Equiv_CVaR_to_bPoE} and~\ref{lem: Equiv_bPoE_to_CVaR}.

\subsubsection{Mean--CVaR to Mean--bPoE direction}
Before presenting the result for the Mean--CVaR to Mean--bPoE direction,
we provide a brief heuristic illustrating how the parameter $\gamma_o$,
the disaster level $D$, and the corresponding point in $\mathcal{S}_o(D, \gamma_o)$
\mbox{naturally emerge.}

Suppose we begin with a point
$\bigl(\mathcal{C}_a^*, \,\mathcal{E}_a^*\bigr)$
in the Mean--CVaR scalarization optimal set $\mathcal{S}_{a}(\alpha, \gamma_a)$,
attained by the threshold/control pair
$\bigl(W_a^{p*}, \,\mathcal{U}_{0,a}^{p\ast}\bigr)$.
The associated Mean--CVaR scalarization objective~is
\[
\gamma_a \left(W_a^{p\ast} + \frac{1}{\alpha} \mathbb{E}_{\mathcal{U}_{0,a}^{p\ast}} \bigl[ \min(W_T - W_a^{p\ast}, 0) \bigr] \right) + \mathbb{E}_{\mathcal{U}_{0,a}^{p\ast}}[W_T].
\]
Rewriting this in a bPoE-style form by introducing a disaster level $D$,
we obtain
\[
  -\tfrac{\gamma_a}{\alpha}\,(W_a^{p\ast} - D)\,
  \Ebb_{\mathcal{U}_{0,a}^{p\ast}}\!
  \Bigl[
    \max\Bigl(
      1 \,-\,\frac{W_T - D}{\,W_a^{p\ast} - D\,},
      0
    \Bigr)
  \Bigr]
  \;+\;
  \Ebb_{\mathcal{U}_{0,a}^{p\ast}}[W_T]
  \;+\;
  \gamma_a\,W_a^{p\ast},
\]
which we compare to the Mean--bPoE objective under the same threshold/control pair
$\bigl(W_o^p,\,\mathcal{U}_{0}\bigr) = \bigl(W_a^{p\ast},\,\mathcal{U}_{0,a}^{p\ast}\bigr)$:\footnote{At this stage, we do not claim
$\bigl(W_a^{p*}, \,\mathcal{U}_{0,a}^{p\ast}\bigr)$ is also optimal
for Mean--bPoE problem; we only use it for rewriting the objective.}
\[
  \gamma_o \,\Ebb_{\mathcal{U}_{0,a}^{p\ast}}\!
  \Bigl[
    \max\Bigl(
      1 \,-\,\tfrac{W_T - D}{\,W_a^{p*} - D\,},
      0
    \Bigr)
  \Bigr]
  \;-\;
  \Ebb_{\mathcal{U}_{0,a}^{p\ast}}[W_T].
\]
Matching terms suggests the relation
$ \gamma_o  = \frac{\gamma_a}{\,\alpha\,}\,\bigl(\,W_a^{p\ast} \,-\,D\bigr)$.
Motivated by the CVaR-bPoE duality result in Remark~\ref{rm:duality_cvar_bpoe}, we set
\mbox{$D = \mathcal{C}_a^* = \operatorname{CVaR}_\alpha(W_T)$},
leading to the formula for $\gamma_{o}$:
$\gamma_{o}
    =
    \tfrac{\gamma_{a}(\,W_{a}^{p*} \;-\; \mathcal{C}_a^*\,)}{\alpha}$.
By Remark~\ref{rm:duality_cvar_bpoe}, setting $D = \operatorname{CVaR}_\alpha(W_T)$ implies
$\text{bPoE}_{D}(W_T) = \alpha$.
{\sifin{In the proof of this direction, we further show that the resulting bPoE point can be attained by the same threshold/control pair
$\bigl(W_a^{p*}, \mathcal{U}_{0,a}^{p\ast}\bigr)$ that attains
$\bigl(\mathcal{C}_a^*, \mathcal{E}_a^*\bigr)$ in the Mean--CVaR formulation,
and that it yields the same expected terminal wealth, i.e.\
$\mathcal{E}_o^* = \mathcal{E}_a^*$.
}}
Hence
$(\mathcal{B}_o^*, \mathcal{E}_o^*) = (\alpha, \mathcal{E}_a^*)$ belongs to
$\mathcal{S}_o(D, \gamma_o)$. This point can be viewed as the image of
$(\mathcal{C}_a^*, \mathcal{E}_a^*)$ under the correspondence from
Mean-CVaR to Mean-bPoE scalarization optimal points.
The above provides the heuristic motivation for the parameter transformation used in
Lemma~\ref{lem: Equiv_CVaR_to_bPoE} below.
\begin{lemma}[Equivalence direction: Mean–CVaR to Mean–bPoE]
\label{lem: Equiv_CVaR_to_bPoE}
Consider a confidence level $\alpha \in (0,1)$ and scalarization parameter $\gamma_a>0$.
Let $\bigl(\mathcal{C}_a^*,\,\mathcal{E}_a^*\bigr)$ be any point in $\mathcal{S}_{a}(\alpha, \gamma_a)$.
{\sifin{Suppose $(W_a^{p*}, \mathcal{U}_{0,a}^{p\ast})$ is an optimal
threshold/control pair associated with a solution that attains
$\bigl(\mathcal{C}_a^*, \mathcal{E}_a^*\bigr)$.}}
In particular, we have $W_{a}^{p*}\ge \mathcal{C}_a^*$.

Define
\EQ
\label{eq:gamma_o}
    \gamma_{o}
    \;=\;
    \frac{\gamma_{a}(\,W_{a}^{p*} \;-\; \mathcal{C}_a^*\,)}{\alpha}\,.
\EN
Then the point $\bigl(\mathcal{B}_o^*,\,\mathcal{E}_o^*\bigr)$,
where $\mathcal{B}_o^*=\alpha$ and $\mathcal{E}_o^*=\mathcal{E}_a^*$,
lies in $\mathcal{S}_{o}\!\bigl(\mathcal{C}_a^*, \gamma_o)$
{\sifin{and can be attained by the threshold/control pair $(W_a^{p*}, \,\mathcal{U}_{0,a}^{p\ast})$.}}
Hence, $\bigl(\mathcal{C}_a^*,\mathcal{E}_a^*\bigr)$ and $\bigl(\mathcal{B}_o^*,\mathcal{E}_o^*\bigr)$
represent the same ``efficient'' portfolio outcome,
mapped from $\mathcal{S}_{a}\!\bigl(\alpha, \gamma_a\bigr)$ to $\mathcal{S}_{o}\!\bigl(\mathcal{C}_a^*, \gamma_{o}\bigr)$.
\end{lemma}
A complete proof of Lemma~\ref{lem: Equiv_CVaR_to_bPoE} is given in Appendix~\ref{app:Equiv_CVaR_to_bPoE}.

\subsubsection{Mean--bPoE  to Mean--CVaR direction}
As a brief heuristic, we observe that a formula for $\gamma_a$ can be derived
by inverting the relationship from the Mean--CVaR to Mean--bPoE
direction in the expression for $\gamma_o$ given in \eqref{eq:gamma_o}.
Specifically, rearranging \eqref{eq:gamma_o} to solve for $\gamma_a$,
and then identifying
$\alpha=\mathcal{B}_{o}^*$, $W_{a}^{p*}=W_{o}^{p*}$,
and $\mathcal{C}_{a}^*=D$, immediately yields
$\gamma_{a}
  \;=\;
  \frac{\gamma_{o}\,\mathcal{B}_{o}^*}{\,W_{o}^{p*}-D\,}$.
Consequently, the point $\bigl(\mathcal{C}_a^*, \mathcal{E}_a^*\bigr)
= \bigl(D, \mathcal{E}_o^*\bigr)$ belongs to
$\mathcal{S}_{a}\!\bigl(\mathcal{B}_o^*, \gamma_{a}\bigr)$ and
can be viewed as the image of $\bigl(\mathcal{B}_o^*, \mathcal{E}_o^*\bigr)$
under the correspondence from Mean--bPoE to Mean--CVaR scalarization optimal points.
We now formally establish this direction in the lemma below.
\begin{lemma}[Equivalence direction: Mean–bPoE to Mean–CVaR]
\label{lem: Equiv_bPoE_to_CVaR}
Consider a disaster level $D>0$ and scalarization parameter $\gamma_o$.
Let $\bigl(\mathcal{B}_o^*,\,\mathcal{E}_o^*\bigr)$ be any point in
$\mathcal{S}_{o}(D, \gamma_o)$.
{\sifin{Suppose $(W_o^{p*}, \,\mathcal{U}_{0,o}^{p\ast})$ is an optimal threshold/control pair associated with a solution that attains $\bigl(\mathcal{B}_o^*,\,\mathcal{E}_o^*\bigr)$.}}
In particular, we have $W_{o}^{p*}> D$.

Define
\EQ
\label{eq:gamma_a}
   \gamma_{a}
   \;=\;
   \tfrac{\gamma_o\, \mathcal{B}_o^*}{W_{o}^{p*}- D}.
\EN
Then the point $\bigl(\mathcal{C}_a^*,\,\mathcal{E}_a^*\bigr)$,
where $\mathcal{C}_a^* = D$ and $\mathcal{E}_a^*=\mathcal{E}_o^*$,
lies in $\mathcal{S}_{a}\!\bigl(\mathcal{B}_o^*,\,\gamma_{a}\bigr)$
{\sifin{and can be attained by the threshold/control pair $(W_o^{p*}, \,\mathcal{U}_{0,o}^{p\ast})$.}}
Hence, $\bigl(\mathcal{B}_o^*,\mathcal{E}_o^*\bigr)$ and $\bigl(\mathcal{C}_a^*,\mathcal{E}_a^*\bigr)$
represent the same ``efficient'' portfolio outcome,
mapped from $\mathcal{S}_{o}(D,\gamma_o)$ to $\mathcal{S}_{a}\bigl(\mathcal{B}_o^*,\gamma_{a}\bigr)$.
\end{lemma}
Full details of a proof of Lemma~\ref{lem: Equiv_bPoE_to_CVaR} are given in Appendix~\ref{app:Equiv_bPoE_to_CVaR}.

\begin{remark}
Lemmas~\ref{lem: Equiv_CVaR_to_bPoE} and~\ref{lem: Equiv_bPoE_to_CVaR} jointly establish a one-to-one correspondence between the efficient frontiers in the Mean--CVaR and Mean--bPoE frameworks. These results carry both mathematical and practical significance. Mathematically, they show that, with appropriately calibrated parameters, both formulations yield precisely the same efficient frontier. Practically, bPoE's dollar-based disaster level provides a more intuitive and communicable alternative to probability-based CVaR---especially for retail investors---{\sifin{while preserving the same efficient frontier.}} Therefore, this equivalence enables seamless integration of pre-commitment Mean--bPoE into existing Mean--CVaR workflows, without altering broader frameworks or results.
\end{remark}
We conclude this section by noting that, although our analysis has focused on terminal wealth objectives, analogous arguments apply when maximizing Expected Withdrawals, a popular reward measure in DC superannuation \cite{forsyth2022stochastic, forsyth2022short}. Minor adjustments to handle interim withdrawals do not \mbox{affect the equivalences.}

\section{Time-consistent Mean--bPoE and Mean--CVaR}
\label{sc:time_consistent}
We now introduce the time-consistent Mean--bPoE (TCMb) and Mean--CVaR (TCMa) problems,
each incorporating an explicit time-consistency constraint
\cite{basak2010, bjork2017, bjork2014}. For a given disaster level $D$ and scalarization parameter $\gamma>0$,
let $V^c_o(s,b,t_m^-)$ be the value function for TCMo at state
$X_m^- = (S_m^-, B_m^-) = (s,b)$ and \mbox{time $t_m^-$}.
The problem $TCMo_{t_{m}}\bigl(D, \gamma\bigr)$ is then formulated as follows:
\begin{equation}
\label{eq: TCMb}
 TCMo_{t_{m}}\left(D, \gamma\right): V^c_o\left(s,b,t_{m}\right) \coloneqq \inf_{W_o^c>D}\bigg\{ \inf_{\mathcal{U}_{m} \in \mathcal{A}}  \Ebb_{\mathcal{U}_{m}}^{X_{m}^+,t_{m}^+}\left[\gamma\,  \max\left( 1 - \tfrac{W_{T}-D}{W_o^c-D}, 0 \right) -  W_{T} \right]
 \bigg| X_{m}^- = (s, b) \bigg\}
\end{equation}
\vspace*{+3pt}
\begin{equation}
\label{eq: TCMb time-consistent constraint}
\begin{aligned}
\text{subject to: } \mathcal{U}_{m} &= \left\{ u_{m},  \mathcal{U}_{m+1, o}^{c,\ast} \right\} = \left\{u_{m}, u_{m+1, o}^{c\ast}, \ldots, u_{M-1, o}^{c\ast}\right\}, 
\\
&\qquad \qquad\qquad \text{where } \mathcal{U}_{m+1, o}^{c\ast} \text{ is the optimal control for }TCMo_{t_{m+1}}\left(D, \gamma\right),
\end{aligned}
\end{equation}
\vspace*{+4pt}
{\sifin{
\begin{equation}
\label{eq: TCMb constraints}
\qquad \qquad\qquad\qquad~\text{and}~
\begin{cases}
~~\text{(dynamics)} \quad~~~~\left(S_{t}, B_{t}\right) \text{ evolves via }
\eqref{eq: dBt}-\eqref{eq: dSt}
& t \notin \mathcal{T}, \\
\left.\begin{array}{ll}
\text{(cash injection)} &W_m^+ = W_m^- + q_m = s + b + q_m, \\
\text{(rebalancing)} &X_m^+ = (S_m^+, B_m^+),  \text{ where }
\\
&S_m^+ = u_m W_m^+,\quad B_m^+ = \left(1-u_m\right)W_m^+, \\
&u_m\in \mathcal{Z} = [0, 1].
\end{array}\right\}
    \end{cases}
\end{equation}}}
Similarly, for a given confidence level $\alpha$ and scalarization parameter $\gamma > 0$,
let $V^c_a(s,b,t_m^-)$ denote the value function for the TCMa problem at state $X_m^- = (S_m^-, B_m^-) = (s,b)$ and time $t_m^-$.
The problem $TCMa_{t_m}\bigl(\alpha, \gamma\bigr)$ is then formulated as follows \cite{PA2020a}:
\begin{equation}
\label{eq: TCMC}
TCMa_{t_{m}}\left(\alpha, \gamma\right): V_a^c\left(s,b,t_{m}^-\right) \coloneqq \sup_{W_a^c\ge 0}\bigg\{ \sup_{\mathcal{U}_{m} \in \mathcal{A}} \Ebb_{\mathcal{U}_{m}}^{X_m^+,t_{m}^+}\big[ \gamma \bigl(W_a^c+\tfrac{1}{\alpha}\min( W_{T} - W_a^c, 0)\bigr) +W_{T} \big] \bigg| X_m^- = (s, b) \bigg\}
\end{equation}
\vspace*{+3pt}
\begin{equation}
\label{eq: TCMC time-consistent constraint}
\begin{aligned}
\text{subject to: } \mathcal{U}_{m} &= \left\{ u_{m},  \mathcal{U}_{m+1, a}^{c\ast} \right\} = \left\{u_{m}, u_{m+1, a}^{c\ast}, \ldots, u_{M-1, a}^{c\ast}\right\},
\\
&\qquad \qquad\qquad\text{where }
\mathcal{U}_{m+1, a}^{c \ast} \text{ is the optimal control for }TCMa_{t_{m+1}}\left(\alpha, \gamma\right),
\end{aligned}
\end{equation}
\vspace*{+4pt}
{\sifin{
\begin{equation}
\label{eq: TCMC constraints}
\qquad \qquad\qquad\qquad~
\text{and the system dynamics, cash injection, and rebalancing constraints in \eqref{eq: TCMb constraints}.}
\end{equation}
}} \noindent Here, the time-consistency constraints \eqref{eq: TCMb time-consistent constraint} and \eqref{eq: TCMC time-consistent constraint}
ensure that the TCMo and TCMa optimal strategies,
 \[
 \mathcal{U}_{m, o}^{c\ast}= \left\{u_{m, o}^{c\ast}, u_{m+1, o}^{c\ast}, \ldots, u_{M-1, o}^{c\ast}\right\}, \quad \text {and} \quad
 \mathcal{U}_{m, a}^{c\ast}= \left\{u_{m, a}^{c\ast}, u_{m+1, a}^{c\ast}, \ldots, u_{M-1, a}^{c\ast}\right\}, \quad
 \]
are indeed time-consistent, so that dynamic programming applies directly
\cite{basak2010, bjork2017, bjork2014, landriault2018equilibrium, wang2011continuous, hu2012time}.

By contrast, the pre-commitment formulations $PCMo_{t_0}(D,\gamma)$ and
$PCMa_{t_0}(\alpha,\gamma)$ in
\eqref{eq: PCMb}--\eqref{eq: PCMb constraints} and
\eqref{eq: PCMC}--\eqref{eq: PCMC constraints}
are specified at the initial time $t_0$ and treat
the entire horizon $[t_0, T]$ as one multi-period optimization
without imposing time consistency. Thus, an investor commits
to a strategy at $t_0$ but does not require that rebalancing
decisions remain optimal upon reconsideration at later times.
Hence, these pre-commitment problems can be solved in a single pass at $t_0$,
ignoring the explicit stagewise optimization; in contrast,
the time-consistent versions $TCMo_{t_m}(D,\gamma)$ and $TCMa_{t_m}(\alpha,\gamma)$
must be solved iteratively across the rebalancing points
$t_m\in\mathcal{T}$ to maintain optimality at each future stage.

We now highlight a key structural difference between the time-consistent Mean–CVaR and Mean–bPoE formulations under lump-sum investment (i.e.\ no {\sifin{cash injections}}).
\begin{lemma}
\label{lem:scaling}
Suppose that no {\sifin{cash injections}} occur, i.e.\ $q_m = 0$ for all $t_m \in \mathcal{T}$
(lump-sum investment). Let $\lambda > 0$ be a scalar. Then:
\begin{itemize}[noitemsep, topsep=0pt, leftmargin=*]
\item The value function $V_a^c(s, b, t_m^-)$ of the time-consistent Mean–CVaR problem defined in~\eqref{eq: TCMC}–\eqref{eq: TCMC constraints} satisfies
\[
V_a^c(\lambda s, \lambda b, t_m^-) = \lambda\, V_a^c(s, b, t_m^-).
\]
Hence, $V_a^c$ is positively homogeneous of degree $1$.
As a consequence, the optimal rebalancing control is $u_{m, a}^{c\ast}(w) = u_{m, a}^{c\ast}(\lambda w)$, where $w = s+b$.

\item The value function $V_o^c(s, b, t_m^-)$ of the time-consistent Mean–bPoE problem defined in~\eqref{eq: TCMb}–\eqref{eq: TCMb constraints} does not, in general, satisfy this positive homogeneity:
\[
V_o^c(\lambda s, \lambda b, t_m^-) \ne \lambda\, V_o^c(s, b, t_m^-).
\]
Consequently, the TCMb optimal rebalancing control does depend on the absolute wealth $w = s+b$.
\end{itemize}
\end{lemma}
\noindent
A full, rigorous proof via dynamic programming and induction appears in~\cite{PA2020a}
(which addresses the Mean–CVaR case only). Specifically, in the TCMa formulation, each subproblem (at any $t_m \in \mathcal{T}$) involves an objective of the form
$
\gamma_a \left(W_a^c + \tfrac{1}{\alpha} \min(W_T - W_a^c,\, 0) \right) + W_T$.
Under lump‐sum conditions (i.e.\ no cash injections), scaling $W_T \mapsto \lambda W_T$ and $W_a^c \mapsto \lambda W_a^c$ yields
\[
\gamma_a \left( \lambda W_a^c + \tfrac{1}{\alpha} \min(\lambda W_T - \lambda W_a^c,\, 0) \right) + \lambda W_T
= \lambda \left[\gamma_a \left( W_a^c + \tfrac{1}{\alpha} \min(W_T - W_a^c,\, 0) \right) + W_T \right].
\]
Hence, the TCMa objective factors out $\lambda$.
This positive homogeneity propagates through the dynamic programming recursion, implying wealth‐independent rebalancing controls, namely
\EQ
\label{eq:TCMa_control}
\text{for each $t_m \in \mathcal{T}$ and } \forall \lambda > 0, ~w>0,
\quad
u_{m,a}^{c\ast}(w)= u_{m,a}^{c\ast}(\lambda w).
\EN
In practical terms, this means the optimal fraction in the risky asset does not change
if total wealth is scaled.

By contrast, the TCMb objective includes a fixed disaster level $D$, leading to terms like
$\max\!\left(1 - \tfrac{W_T - D}{W_o^c - D},\, 0\right)$.
Scaling $W_T \mapsto \lambda W_T$ and $W_o^c \mapsto \lambda W_o^c$ does not eliminate the constant $D$, so the expression fails to factor out~$\lambda$. As a result, the bPoE objective is not positively homogeneous, and the optimal rebalancing control in TCMb depends on the absolute wealth level. That is,
\EQ
\label{eq:TCMb_control}
\exists \text{ some $t_m \in \mathcal{T}$ and some } \lambda>0,\; w>0
  \quad\text{such that}\quad
  u_{m,o}^{c\ast}(\lambda\,w)
  \;\neq\;
  u_{m,o}^{c\ast}(w).
\EN

\begin{remark}[{\sifin{No equivalence under lump-sum when $D < e^{rT} w_0$}}]
\label{rm:no_equiv}
{\sifin{
Throughout this remark, we assume that the disaster level $D$ satisfies
$D < e^{rT} w_0$, i.e.\ it lies below the risk–free growth of the initial
lump–sum $w_0$, as in Remark~\ref{rm:bpoE_Wstar}.
}}

Suppose we are in a lump-sum setting (i.e., $q_m = 0$ for all $t_m \in \mathcal{T}$), and let $\Phi$ be a hypothetical global mapping from $(\alpha, \gamma_a)$ to $(D, \gamma_o)$ such that the time-consistent problems $TCMa_{t_0}(\alpha, \gamma_a)$ and $TCMo_{t_m}(D, \gamma_o)$ yield the same efficient frontier, attained by the same optimal rebalancing control and threshold pair.

In the TCMa formulation, the control $u_{m,a}^{c\ast}(w)$ is wealth-independent in a lump-sum environment, as per \eqref{eq:TCMa_control}. By contrast, the TCMb control $u_{m,o}^{c\ast}(w)$ is wealth-dependent, as in \eqref{eq:TCMb_control}, due to the fixed disaster level $D$. Consequently, there must exist some scaled wealth level $\lambda w$ at which the TCMb control differs from that of TCMa, a contradiction. Although this argument focuses on mapping from TCMa to TCMb, a symmetric consideration applies in the reverse direction. Hence, no global parameter mapping $\Phi$ can yield the same efficient frontiers and optimal control/threshold pairs across the entire investment horizon.\footnote{In contrast, in the pre-commitment setting, Lemmas~\eqref{lem: Equiv_CVaR_to_bPoE} and~\eqref{lem: Equiv_bPoE_to_CVaR} establish the existence of a mapping $\Phi$ between the two formulations such that equivalent points on the efficient frontiers are attained by the same optimal control and threshold pair.}

In our numerical examples, we consider the practical setting where the initial investment is zero and the investor contributes a fixed amount (in real terms) at each rebalancing date. This departs from the lump-sum assumption in Lemma~\ref{lem:scaling}. However, for states at time $t_m$ where the future discounted value of remaining contributions is small relative to current wealth, that is,
$w = s + b \gg \sum_{\iota = m}^{M-1} e^{-r(T - t_\iota)} q_\iota$,
the control $u_{m,a}^{c\ast}(w)$ is expected to depend only weakly on $w$.
This observation suggests that a global parameter mapping $\Phi$ is unlikely to preserve
the same frontier and optimal control/threshold structure between the TCMa and TCMb formulations, even beyond the strict lump-sum setting.
\end{remark}
{\sifin{
We emphasize that this non–equivalence statement in Remark~\ref{rm:no_equiv}
is specific to the regime $D < e^{rT} w_0$: if $D$ is chosen so large that this condition fails, then time--consistent Mean–bPoE may degenerate to the fully risky
strategy $u^{c*}\equiv 1$ on parts (or all) of the state space, and in
such cases we do not claim a qualitative distinction from the
wealth–independent Mean–CVaR controls.
}}

\section{Numerical methods}
\label{sc:num_methods}
This section presents provably convergent numerical integration methods for both pre-commitment and time-consistent multi-period Mean–bPoE and Mean–CVaR problems. To the best of our knowledge, no convergent schemes have previously been established for this class of mean–risk formulations. Existing work, such as~\cite{PA2020a}, addresses only the Mean–CVaR case and does not provide a convergence analysis. In contrast to the PDE-based approach of~\cite{PA2020a}, which enforces monotonicity via a user-defined parameter~$\varepsilon$, our method guarantees strict monotonicity by construction—a property critical for convergence in stochastic control problems~\cite{kushner2001numerical}.

To treat Mean--bPoE and Mean--CVaR in a unified framework, we introduce a generic threshold variable: $W^p \in \{W_o^p,\,W_a^p\}$ for pre-commitment and $W^c \in \{W_o^c,\,W_a^c\}$ for time-consistent formulations, collectively referred to as $W^\bullet$, and include it in the state vector. We also apply a log transformation to the risky asset, resulting in the augmented controlled process $\{\widehat{X}_t\}$, where $\widehat{X}_t = (Y_t, B_t,\,W^\bullet)$ with $Y_t = \ln(S_t)$.
For subsequent use, given $u_m \in \mathcal{Z}$, we define the intervention operator $\mathcal{M}(u_m)$, applied to a function $F(\cdot)$ defined on the augmented state $(y,b, W^\bullet, t)$, as
\EQ
\label{eq:Operator_M}
\mathcal{M}(u_m)~\big\{F\left(y,b, W^\bullet, t_m\right)\big\} =
  F\big(y^+(s, b, u_m), b^+(s, b, u_m), W^\bullet, t_m^+\big),
\EN
where, by \eqref{eq: W=W+q}, we have
\EQ
\label{eq:sbplus}
y^+(y, b, u_m) = \ln(u_m \, (e^y+b + q_m)) , \quad b^+(e^y, b, u_m) = (1-u) \, (e^y+b + q_m).
\EN
For each feasible $W^\bullet$, we define the payoff function
\EQ
\label{eq:payoff_l}
f(w, W^\bullet) =
\begin{cases}
 \gamma_o \max\left( 1 - \frac{w-D}{W^\bullet-D} , 0 \right) - w,  & \text{Mean--bPoE},
\\
 \gamma_a \left( W^\bullet + \tfrac{1}{\alpha} \min\left(w - W^\bullet , 0 \right)\right)  + w,  & \text{Mean--CVaR}.
 \end{cases}
\EN
\subsection{Pre-commitment Mean--bPoE and Mean--CVaR}
\label{ssc:pre}
Recall that the pre-commitment Mean–bPoE formulation $PCMo_{t_0}(D,\gamma)$ in \eqref{eq: PCMb}–\eqref{eq: PCMb constraints} and its Mean–CVaR counterpart in \eqref{eq: PCMC}–\eqref{eq: PCMC constraints} share a similar structure and can both be tackled using a ``lifted'' state approach~\cite{miller2017optimal}:
the threshold is fixed at candidate values, the corresponding inner control problem is solved over $[0, T]$ for each value, and an outer search is performed at time $t_0$ to determine the optimal threshold. With this in mind, we define a generic auxiliary function $\widehat{V}^p(\cdot)$ on the augmented state, where $W^p$ remains fixed throughout:
\begin{eqnarray}
\widehat{V}^p\left(y,b, W^p, t_m^-\right) &\coloneqq& \Ebb_{\mathcal{U}_m}^{\widehat{X}_{m}^+, t_{m}^+}\bigg[f\left(e^{Y_T} + B_T, W^p\right) \bigg| \widehat{X}_m^- = (y, b, W^p)  \bigg].
\label{eq: PCMb_l}
\end{eqnarray}
The terminal condition for this  auxiliary function at time $t = T$ is given by
\EQ
\label{eq:terminal}
\widehat{V}^p\left(y,b, W^p, T\right) = f(e^y+b, W^p),
\quad \text{ $f(\cdot)$ is defined in \eqref{eq:payoff_l}}.
\EN
For each $t_m \in \mathcal{T}$, the interest accrued over $[t_m,\,t_{m+1}]$ is settled during $[t_{m+1}^-,\,t_{m+1}]$, leading to
\EQ
\label{eq:interest}
\widehat{V}^p\left(y,\, b, \, W^p,  t_{m+1}^-\right)  = \widehat{V}^p\left(y, \, b\,e^{r\,\Delta t}, \, W^p, \, t_{m+1}\right).
\EN
Over $[t_m^+,\,t_{m+1}^-]$, the risk-free amount $b$ is constant, and $\{Y_t\}=\{\ln(S_t)\}$ evolves under the log-dynamics of~\eqref{eq: dSt}. Hence,
for a fixed $b$, the backward recursion takes the form of a convolution integral:
\EQ
\label{eq:integral}
\widehat{V}^p\bigl(y,b, \, W^p,  t_m^+\bigr)
=
\int_{-\infty}^{\infty}
\widehat{V}^p\bigl(y',\,b, \, W^p, \,t_{m+1}^-\bigr)\,\,
g(y- y',\Delta t) \,dy'.
\EN
Here, $g(y - y', \Delta t)$ is the transition density of the log-state $Y_t$ from $y$ at $t_m^+$ to $y'$ at $t_{m+1}^-$, and depends only on the displacement $y - y'$ and timestep $\Delta t$ due to the spatial and temporal homogeneity of~\eqref{eq: dSt}.
\begin{remark}[Transition density for the Kou model]
\label{rm:Kou}
In the Kou model, the transition density $g(y, \Delta t)$ admits an infinite series representation of the form $g(y, \Delta t) = \sum_{\ell = 0}^{\infty} g_\ell(y, \Delta t)$, where each term $g_\ell$ is non-negative~\cite{zhang2024monotone}. Full details of this representation are provided in Appendix~\ref{app:num_scheme}.
For numerical implementation, the series is truncated to a finite number of terms to form a partial sum; see Subsection~\ref{ssc:num_scheme}.
\end{remark}

\noindent From $t_m^+$ to $t_m$, the optimal rebalancing $u_m^{p\ast}(\cdot|W^p)$ corresponding to the fixed threshold $W^p$ is determined by
\EQ
\label{eq:control}
u_m^{p\ast}(\cdot|W^p) =
\begin{cases}
u_{m, o}^{p\ast}(\cdot|W_o^p) = \underset{u_m \in \mathcal{Z}}{\arginf}\, \mathcal{M}(u_m)~\big\{\widehat{V}_o^p\left(y,b, W_o^p,  t_m^+\right)\big\},
& W^p = W_o^p,
\\
u_{m, a}^{p\ast}(\cdot|W_a^p) = \underset{u_m \in \mathcal{Z}}{\argsup}\,
\mathcal{M}(u_m)~\big\{\widehat{V}_a^p\left(y,b, W_a^p,  t_m^+\right)\big\},    & W^p = W_a^p.
 \end{cases}
\EN
where $\mathcal{M}(u_m)$ is given in \eqref{eq:Operator_M}.
The auxiliary function is then updated via
\begin{equation}
\label{eq:intervention}
\widehat{V}^p\bigl(y,b,W^p, t_m\bigr)
=
\mathcal{M}(u_m^{p\ast}(\cdot|W^p)) \, \big\{\widehat{V}^p\bigl(y,b, W^p, t_m^+\bigr)\big\}.
\end{equation}
At time $t_0$, we find the optimal threshold $W^{p\ast}$ via
an outer exhaustive search over all feasible $W^p$
\EQ
\label{eq:vhat_l}
W^{p\ast} =
\begin{cases}
W_o^{p\ast}
=
 \underset{W_o^p>D}{\arginf}\,  \widehat{V}_o^p\left(y_0,b_0, W_o^p, t_0\right), & \text{Mean-bPoE},
 \\
W_a^{p\ast} = \underset{W_a^p \ge 0}{\argsup}\,  \widehat{V}_a^p\left(y_0,b_0, W_a^p, t_0\right),  & \text{Mean-CVaR},
\end{cases}
\EN
The pre-commitment optimal control is then determined by
$\mathcal{U}_0^{p\ast}=\bigl\{u_0^{p\ast}(\cdot | W^{p\ast}), \ldots,u_{M-1}^{p\ast}(\cdot|W^{p\ast})\bigr\}$.

Finally, at $t_0^-$, we set the Mean-bPoE and Mean-CVaR value functions respectively~as
\EQ
\label{eq:vhat_l_V}
V_o^p(s_0, b_0, t_0^-) = \widehat{V}^p\left(y_0,b_0, W_o^{p\ast}, t_0^-\right)\quad
\text{ and }\quad
V_a^p(s_0, b_0, t_0^-) = \widehat{V}^p\left(y_0,b_0, W_a^{p\ast}, t_0^-\right)
\EN
\begin{proposition}[Lifted formulation equivalence]
\label{prop:equiv_lifted_precommitment}
Under the log transformation $y = \ln(s)$,
the formulation \eqref{eq:payoff_l}--\eqref{eq:vhat_l_V} is equivalent
to $PCMo_{t_0}(D,\gamma_o)$ in
\eqref{eq: PCMb}--\eqref{eq: PCMb constraints} when $W^p = W_o^p$,
and to $PCMa_{t_0}(\alpha,\gamma_a)$ in \mbox{\eqref{eq: PCMC}--\eqref{eq: PCMC constraints}}
when $W^p = W_a^p$.
\end{proposition}
\noindent The proof of the proposition follows directly by substitution and application of the backward recursion.
\begin{remark}[Computation of expectation]
\label{rm:expectation}
We obtain the expectation $\text{E}_{\mathcal{U}_0^{p\ast}}^{x_0,t_0}$ under
$(W^{p\ast},\mathcal{U}_0^{p\ast})$ by defining an auxiliary function $\widehat{E}(y,b, W^{p\ast}, t)$ with terminal condition
$\widehat{E}(y,b,\cdot, T)=e^y+b$, and evolving it backward with the same convolution
equation and rebalancing decision as $\widehat{V}^p$:
\[
\widehat{E}^p\bigl(y,b, \cdot, t_m^+\bigr)
=
\int_{-\infty}^{\infty}
\widehat{E}^p(y', b, \cdot, t_{m+1}^-)
g(y- y',\Delta t)dy',
\quad
\widehat{E}^p\bigl(y,b, \cdot, t_m\bigr)
=
\mathcal{M}(u_m^{p\ast}(\cdot|W^{p\ast})) \, \big\{\widehat{E}^p\bigl(y,b, \cdot, t_m^+\bigr)\big\}.
\]
The interest settlement during $[t_{m+1}^-,\,t_{m+1}]$ is in the same fashion as \eqref{eq:interest} by updating $b$ to $be^{r \Delta t}$.
At time~$t_0^-$, we have
$\text{E}_{\mathcal{U}_0^{p\ast}}^{x_0,t_0} = \widehat{E}^p\bigl(y_0,b_0,W^{p\ast}, t_0^-\bigr)$. Thus, 
\[
\text{bPoE}_{\mathcal{U}_0^{p\ast}}^{x_0,t_0} = \gamma_o \widehat{V}_o^p\left(y_0,b_0, t_0^-\right) - \widehat{E}^p\bigl(y_0,b_0,W_o^{p\ast}, t_0^-\bigr)
\quad
\text{and}
\quad
\text{CVaR}_{\mathcal{U}_0^{p\ast}}^{x_0,t_0}= \gamma_a \widehat{V}_a^p\left(y_0,b_0, t_0^-\right) - \widehat{E}^p\bigl(y_0,b_0,W_a^{p\ast}, t_0^-\bigr).
\]
\end{remark}
\subsubsection{Localization and problem definition}
\label{ssc:num_pro}
Since $W^{p\ast}$ is finite by Lemmas~\ref{lem:infF} and~\ref{lem:CVaR_threshold_existence},
we truncate the threshold domain to $\Gamma \equiv \Gamma_o =  [D,\,W_o^{\max}]$ (Mean–bPoE) or $\Gamma \equiv \Gamma_a =[0,\,W_a^{\max}]$ (Mean–CVaR), where $W_o^{\max}$ and $W_a^{\max}$ are sufficiently large. We localize the $(y,b)$ domain to
$\Omega = \big[y_{\min}^{\dagger}, y_{\max}^{\dagger}\big] \times \left[0, b_{\max}\right]$,
where $y_{\min}^{\dagger} < y_{\min} < 0 < y_{\max} < y_{\max}^{\dagger}$
and $b_{\max} > 0$ are chosen sufficiently large in magnitude
to ensure negligible boundary errors \cite{zhang2024monotone, DangForsyth2014}.
We partition $\Omega$ into
\begin{align*}
\Omega_{\myin} = (y_{\min}, y_{\max}) \times [0, b_{\max}],\quad
\Omega_{y_{\min}} = [y_{\min}^{\dagger}, y_{\min}] \times [0, b_{\max}], \quad
\Omega_{y_{\max}} = [y_{\max}, y_{\max}^{\dagger}] \times [0, b_{\max}].
\end{align*}
On $\Omega_{\myin}$, for each fixed  $W^p\in\Gamma$,
we truncate the integral \eqref{eq:integral} to
\EQ
\label{eq:integral_trunc}
\widehat{V}^p\bigl(y,b, \, W^p,  t_m^+\bigr)
=
\int_{y_{\min}^{\dagger}}^{y_{\max}^{\dagger}}
\widehat{V}^p\bigl(y',\,b, \, W^p, \,t_{m+1}^-\bigr)\,\,
g(y- y',\Delta t) \,dy',
\quad y \in (y_{\min}, y_{\max}).
\EN
For $(y, b) \in \Omega_{y_{\min}}$ (resp.\  $\Omega_{y_{\max}}$),
by the payoff function \eqref{eq:terminal},
we assume $\widehat{V}^p(y,b, \cdot, t)$ has the form $A_0(t) b$ (resp.\ $A_1(t) e^y$) for some unknown $A_0(t)$ (resp.\ $A_1(t)$) to approximate the behavior as $y \to -\infty$ (resp.\ $y \to \infty$). Substituting this form into~\eqref{eq:integral}, and applying standard properties of exponential L\'evy processes, yields
\EQ
\label{eq:boundary}
\Omega_{y_{\min}}: \widehat{V}^p(y,b, \cdot, t_m^+) = \widehat{V}^p(y,b, \cdot, t_{m+1}^-)
\quad
\big(\text{resp.}~ \Omega_{y_{\max}}: \widehat{V}^p(y,b, \cdot, t_m^+) = e^{\mu \Delta t}\, \widehat{V}^p(y,b, \cdot, t_{m+1}^-)  \big).
\EN
Finally, for interest-settlement usage in \eqref{eq:interest},
we define $\Omega_{b_{\max}} = (y_{\min}, y_{\max}) \times (b_{\max}, b_{\max}e^{rT}]$
and approximate the solution there using linear extrapolation
\EQ
\label{eq:bmax}
\widehat{V}^p(y,b, \cdot, t_m^+) = \tfrac{b}{b_{\max}}\widehat{V}^p (y, b_{\max},  \cdot, t_m^+).
\EN
\begin{defn}[Localized pre-commitment formulations]
\label{def:glwb}
The value function of the pre-commitment Mean-bPoE/CVaR problem at time
$t_0^{-}$ is given by $V^p(s_0, b_0, t_0^-) = \widehat{V}^p(y_0,b_0, W^{p\ast}, t_0^-)$, where
$y_0 = \ln(s_0)$, $W^{p\ast}$ is the optimal threshold determined via the outer search~\eqref{eq:vhat_l}, and $\mathcal{U}_0^{p\ast} = \bigl\{u_0^{p\ast}(\cdot|W^{p\ast}), \ldots, u_{M-1}^{p\ast}(\cdot|W^{p\ast})\bigr\}$ is the associated \mbox{optimal control}.

The function $\widehat{V}^p(y,b, W^{p}, t)$ is defined on  $\Omega \times \Gamma \times \{\mathcal{T} \cup \{T\}\}$ as follows.
At each $t_{m} \in \mathcal{T}$, $\widehat{V}^p(\cdot, t_m)$ is given by the rebalancing optimization \eqref{eq:control}–\eqref{eq:intervention},
where $\widehat{V}^p(\cdot, t_m^+)$ satisfies:
(i) the  integral \eqref{eq:integral_trunc} on $\Omega_{\myin}\times\{t_m^+\}$,
(ii) the boundary conditions \eqref{eq:boundary} and \eqref{eq:bmax} on
$\{\Omega_{y_{\min}}, \Omega_{y_{\max}}, \Omega_{b_{\max}}\} \times\{t_{m}^+\}$.
In addition, $\widehat{V}^p(\cdot,t)$ satisfies the terminal condition~\eqref{eq:terminal} on $\Omega \times \{T\}$ and the interest-settlement update~\eqref{eq:interest} on $\Omega \times \{t_{m+1}^-\}$.
\end{defn}
Since the spatial domain $\Omega \subset \mathbb{R}^2$ and the threshold interval $\Gamma \subset \mathbb{R}$ are both bounded, the arguments in~\cite[Proposition 3.1]{zhang2024monotone} can be applied to show that, for each fixed threshold $W^p \in \Gamma$, the discrete-time function $\widehat{V}^p(y, b, W^p, t_m)$ is unique, bounded, and continuous on $\Omega_{\myin} \times \{t_m\}$ for each $t_m \in \mathcal{T} \cup {T}$. Continuity may not extend across the boundaries $y = y_{\min}$ and $y = y_{\max}$ due to imposed boundary conditions.

\subsubsection{Numerical schemes and convergence}
\label{ssc:num_scheme}
We first discretize the domain and then apply a numerical scheme for the localized problem in Definition~\ref{def:glwb}. The admissible control set $\mathcal{Z} = [0,1]$ is discretized using $N_u$ nodes, yielding $\{u_i\}_{i=0}^{N_u}$. The threshold domain $\Gamma$ is partitioned into $N_w$ unequally spaced intervals, denoted by $\{W_k\}_{k=0}^{N_w}$.

Without loss of generality, for convenience, we assume that  $|y_{\min}|$ and $y_{\max}$ are chosen sufficiently large with:
$y^{\dagger}_{\min} = y_{\min} - \tfrac{y_{\max} - y_{\min}}{2}$,
and
$y^{\dagger}_{\max} =  y_{\max} + \tfrac{y_{\max} - y_{\min}}{2}$.
We discretize $[y_{\min} , y_{\max}]$ and $[y^{\dagger}_{\min}, y^{\dagger}_{\max}]$
using uniform partitions with $N_y$  and $N_y^\dagger = 2 N_y$ subintervals, respectively.
This yields a single set of $y$-coordinates, denoted by $\{y_n\}_{n = -N_y^\dagger/2}^{N_y^\dagger/2}$, with $\{y_n\}_{n = -N_y/2+1}^{N_y/2-1}$
corresponding to the interior interval $(y_{\min}, y_{\max})$.
We use a nonuniform partition in $b$, consisting of $N_b$ subintervals, denoted by $\{b_j\}_{j = 0}^{N_b}$.

For each fixed $W_k$, we approximate the function $\widehat{V}^p(y_n, b_j, W_k, t_m^\star)$ using a discrete scheme that produces the numerical approximation $\widehat{V}_h^p(y_n, b_j, W_k, t_m^\star)$, where $h$ denotes a unified discretization parameter (i.e.\ $h \to 0$ implies $N_y^\dagger, N_b, N_w, N_u \to \infty$); $(y_n, b_j, W_k, t_m)$ is a reference node; and $t_m^\star \in \{t_m, t_m^+, t_m^-\}$ indicates the evaluation time. The numerical scheme proceeds as follows.

\noindent \textbf{Inner optimization (fixed $\boldsymbol{W_k}$).} This consists of the steps below.
\begin{itemize}[noitemsep, topsep=0pt, leftmargin=*]
\item Terminal condition \eqref{eq:terminal}: we set
\EQ
\label{eq:terminal_sym}
\widehat{V}_h^p(y_n, b_j, W_k, T) = f(e^{y_n} + b_j, W_k).
\EN

\item Interest-settlement update\eqref{eq:interest}: we apply interpolation or extrapolation as needed to compute
\EQ
\label{eq:interest*_sym}
\widehat{V}_h^p(y_n, b_j, W_k, t_{m+1}^-)  =  \widehat{V}_h^p(y_n, b_je^{ r\Delta t}, W_k, t_{m+1}).
\EN

\item Time-advancement via integral \eqref{eq:integral_trunc}:
For numerical implementation, we truncate the infinite series representation of $g(\cdot, \Delta t)$ (see Remark~\ref{rm:Kou}) after $N_g$ terms (typically $N_g = 10\text{--}15$), forming the partial sum $g(y, \Delta t; N_g) = \sum_{\ell = 0}^{N_g} g_\ell(y, \Delta t)$,
where each term $g_\ell$ is non-negative.
This partial sum is used to approximate~\eqref{eq:integral} via a discrete convolution along the $y$-dimension: for each $(y_n, b_j) \in \Omega_{\myin}$, we compute
\EQ
\label{eq:scheme*_sym}
\widehat{V}_h^p(y_n, b_j, W_k, t_m^+)
 = \textstyle  \sum_{l=-N^{\dagger}/2}^{N^{\dagger}/2}~ \varphi_{l}~
~g(y_n - y_l, \Delta t; N_g)~
\widehat{V}_h^p(y_l, b_j, W_k, t_{m+1}^-),
\EN
where $\{\varphi_l\}$ are the composite trapezoidal weights.
As the discretization parameter $h \to 0$, we also let
$N_g \to \infty$ so that $|g - g(\cdot;\,N_g)|\to 0$, thus ensuring no truncation error in the limit. Full details of this truncated representation and its error bounds appear in
Appendix~\ref{app:num_scheme}.
\item
Boundary condition \eqref{eq:boundary}: we enforce
\EQ
\label{eq:bdry_adv_sym}
\Omega_{y_{\min}}:
\widehat{V}_h^p(y_n, b_j, \cdot, t_m^+) = \widehat{V}_h^p(y_n, b_j, \cdot, t_{m+1}^-),
\quad
\Omega_{y_{\max}}:\widehat{V}_h^p(y_n, b_j, \cdot, t_m^+) =  e^{\mu \Delta t} \widehat{V}_h^p(y_n, b_j, \cdot, t_{m+1}^-).
\EN

\item
Rebalancing \eqref{eq:control}–\eqref{eq:intervention}:
we solve the optimization problem by exhaustive search,
interpolating as needed:
\EQ
\label{eq:control_sym}
\widehat{V}_h^p(y_n, b_j, W_k, t_m)
=
\begin{aligned}[t]
&\begin{cases}
\underset{\{u_i\}}{\min}\,~
\widehat{V}_h^p(y_n^+, b_j^+, W_k, t_m^+), &\text{Mean--bPoE},
\\
\underset{\{u_i\}}{\max}\,~
\widehat{V}_h^p(y_n^+, b_j^+, W_k, t_m^+) , & \text{Mean--CVaR},
\end{cases}
\\
&y_n^+=\ln(u_i(e^{y_n}+b_j+q_m)),~ b_j^+=(1-u_i)(e^{y_n}+b_j+q_m)
\text{ as given in~\eqref{eq:sbplus}.}
\end{aligned}
\EN
This step yields the numerically computed optimal rebalancing control $u_{m,h}^{p\ast}(\cdot| W_k)$ at each node $(y_n, b_j, t_m)$.
\end{itemize}

\noindent \textbf{Outer optimization {\boldmath\eqref{eq:vhat_l}}.}
After computing $\widehat{V}_h^p(y, b, W_k, t)$ for each discretized value $W_k$, we perform an exhaustive search over $\{W_k\}$ to obtain the Mean–bPoE and Mean–CVaR results, respectively, as follows:
\EQ
\label{eq:outer_search_sim}
\widehat{V}_{o, h}^p(y_0, b_0, t_0^-) =
\underset{\{W_k\}}{\min}\, \widehat{V}_{o, h}^p(y_0, b_0, W_k, t_0),
\quad \text{ and }\quad
\widehat{V}_{a, h}^p(y_0, b_0, t_0^-)= \underset{\{W_k\}}{\max}\,
 \widehat{V}_{a, h}^p(y_0, b_0, W_k, t_0).
\EN
Finally, the pre-commitment numerical value function at inception
for $(s_0, b_0) \in  \{(s, b) \mid (\ln s, b) \in \Omega_{\myin}\}$ is\footnote{If the model specifies $s_0 = 0$, we approximate it by $e^{y_{-N_y/2 + 1}}$, which corresponds to the first interior $y$-node. As $h \to 0$, we have $e^{y_{-N_y/2 + 1}} \searrow e^{y_{\min}}$. Since $|y_{\min}|$ is chosen sufficiently large, the resulting error in evaluating $V_h^p(0, b_0, t_0^-)$ using $V_h^p(e^{y_{-N_y/2 + 1}}, b_0, t_0^-)$ is negligible.}
\EQ
\label{eq:incept_pre}
V_h^p(s_0,b_0,t_0^-) = \widehat{V}_h^p(y_0,b_0,W_h^{p\ast},t_0^-),
\EN
where $W_h^{p\ast}$ is the computed optimal threshold from \eqref{eq:outer_search_sim}.
The scheme also yields the associated computed optimal control $\mathcal{U}_{0,h}^{p\ast} = \{u_{0,h}^{p\ast}(\cdot|W_h^{p\ast}), \ldots, u_{M-1,h}^{p\ast}(\cdot|W_h^{p\ast})\}$, with $u_{m,h}^{p\ast}(\cdot|W_h^{p\ast})$ \mbox{obtained from \eqref{eq:control_sym}}.

\noindent \textbf{Expectation of $\boldsymbol{W_T}$.} We approximate this as discussed in Remark~\ref{rm:expectation}, using the same grids, and boundary conditions. Time advancement is carried out via discrete convolution, as in \eqref{eq:scheme*_sym}, along with interpolation/extrapolation similar to \eqref{eq:control_sym} and \eqref{eq:interest*_sym} to handle intervention and interest settlement. 

We next state the convergence result for the numerical scheme in the pre-commitment setting.
\begin{theorem}[Pre-commitment scheme convergence]
\label{thm:convergence}
Consider the pre-commitment Mean--bPoE/CVaR problem defined in
Definition~\ref{def:glwb} on the localized domain
$\Omega \times \Gamma \times \{\mathcal{T} \cup \{T\}\}$.
Suppose the threshold domain $\Gamma$ is chosen sufficiently
large to contain the optimal threshold $W^{p\ast}$.
Also suppose linear interpolation is used for the intervention (rebalancing) step.
As the discretization parameter $h \to 0$ (i.e.\ $N_y^\dagger, N_b, N_w, N_u, N_g \to \infty$)
the numerical scheme \eqref{eq:terminal_sym}-\eqref{eq:incept_pre} for $V_h^p(\cdot)$ converges in both the value function and the optimal threshold.
\begin{itemize}[noitemsep, topsep=0pt, leftmargin=*]
\item Value function convergence:
$\displaystyle\lim_{h\to 0}
     \bigl\lvert
       V_h^p(s_0,b_0,t_0^-)
       -
       V^p(s_0,b_0,t_0^-)
     \bigr\rvert =
     0$, for $(s_0, b_0) \in  \{(s, b) \mid (\ln s, b) \in \Omega_{\myin}\}$
\item Threshold convergence: As $h \to 0$, any sequence of computed optimal thresholds $\{W_{h}^{p\ast}\}$ has a subsequence converging to $W^{p,\ast}$.
\end{itemize}
\end{theorem}
\noindent A detailed proof is given in Appendix~\ref{app:convergence}.

\subsection{Time-consistent Mean--bPoE and Mean--CVaR}
In contrast to the pre-commitment problems in Subsection~\ref{ssc:pre},
where the optimal threshold $W^p$ is obtained via  an outer optimization
at time $t_0$, time-consistent formulations
re-optimize the threshold $W^c\in \{W_o^c, W_a^c\} $ at each rebalancing time $t_m \in \mathcal{T}$.
Specifically, we use an embedding technique, lifting the state space to
\mbox{$\widehat{X}_t \;=\; (Y_t,\,B_t,\,W^c)$} where $Y_t=\ln(S_t)$ and $W^c$ remains
a decision variable that can be re-optimized at each rebalancing time.
We then define an auxiliary function $\widehat{V}^c(\cdot)$ by
\begin{eqnarray}
\widehat{V}^c\bigl(y,b,W^c,t_m^-\bigr)
&\coloneqq&
\Ebb_{\mathcal{U}_m}^{\,\widehat{X}_m^+,\,t_m^+}
\Bigl[
  {\sifin{f\bigl(e^{Y_T}+B_T,\;W^c\bigr)}}
  \;\Big|\;
  \widehat{X}_m^-=(y,b,W^c)
\Bigr].
\label{eq: PCMC_l}
\end{eqnarray}
The terminal condition at time $t = T$ for each feasible threshold value $W^c$ is given by
\EQ
\label{eq:terminal_c}
\widehat{V}^c\left(y,b, W^c, T\right) = f(e^y+b, W^c),
\quad \text{ $f(\cdot)$ is defined in \eqref{eq:payoff_l}}.
\EN
The interest accrued over $[t_m,\,t_{m+1}]$ is settled during $[t_{m+1}^-,\,t_{m+1}]$ for all feasible threshold values $W^c$:
\begin{equation}
\label{eq:interest_c}
\widehat{V}^c\bigl(y,b,W^c,t_{m+1}^-\bigr)
\;=\;
\widehat{V}^c\!\bigl(y,\;b\,e^{r\,\Delta t},\;W^c,\;t_{m+1}\bigr).
\end{equation}
Over $[t_m^+,\,t_{m+1}^-]$, the backward recursion for each feasible value of $W^c$ is given by
the integral as in \eqref{eq:integral}
\begin{equation}
\label{eq:integral_c}
\widehat{V}^c\bigl(y,b,W^c,t_m^+\bigr)
\;=\;
\int_{-\infty}^{\infty}
\widehat{V}^c\bigl(y',\,b,\,W^c,\,t_{m+1}^-\bigr)\,g\bigl(y-y',\Delta t\bigr)\,dy'.
\end{equation}
{\sifin{Unlike the pre-commitment case, in the time-consistent formulation we must,
at each rebalancing time $t_m$, re-optimize both the rebalancing
control $u_m$ and the threshold $W^c$ at the state $(y,b)$.
Let $(u_m^{c\ast}(y,b), W_m^{c\ast}(y,b))$ denote an optimal pair.
Then the update from $t_m^{+}$ to $t_m$ for the function $\widehat{V}^c(\cdot)$ is
}}
{\sifin{
\begin{equation}
\label{eq:intervention_c}
\widehat{V}^c\bigl(y,b,W^c,t_m\bigr)
=
\mathcal{M}\bigl(u_m^{c\ast}(y,b)\bigr)\,
\Big\{\widehat{V}^c\bigl(y,b,W^c,t_m^+\bigr)\Big\}.
\end{equation}
}}
%
{\sifin{The associated optimal rebalancing control and threshold
$(u_m^{c\ast}(y,b), W_m^{c\ast}(y,b))$ are given by
\begin{equation}
\label{eq:control_c}
\begin{cases}
\displaystyle
\bigl(u_m^{c\ast}(y,b),\,W_m^{c\ast}(y,b)\bigr)
\in
\underset{\substack{u_m \in \mathcal{Z}\\[1pt] W_o^c > D}}{\arginf}
\Bigl\{
  \mathcal{M}(u_m)\,\widehat{V}_o^c\bigl(y,b,W_o^c,t_m^+\bigr)
\Bigr\},
& \text{(Mean--bPoE)},\\[10pt]
\displaystyle
\bigl(u_m^{c\ast}(y,b),\,W_m^{c\ast}(y,b)\bigr)
\in
\underset{\substack{u_m \in \mathcal{Z}\\[1pt] W_a^c \ge 0}}{\argsup}
\Bigl\{
  \mathcal{M}(u_m)\,\widehat{V}_a^c\bigl(y,b,W_a^c,t_m^+\bigr)
\Bigr\},
& \text{(Mean--CVaR)}.
\end{cases}
\end{equation}
Here, $\mathcal{M}(u_m)$ is the intervention operator defined in
\eqref{eq:Operator_M}.}}

Finally, at the initial time $t_0^-$, we set  the Mean-bPoE and Mean-CVaR value functions
respectively as
\begin{equation}
\label{eq:vhat_l_c}
{\purple{\widehat{V}_o^c\bigl(y_0,b_0,t_0^-\bigr) = \inf_{\,W_o^c > D}
\,\widehat{V}^c\bigl(y_0,b_0, W_o^c, t_0 \bigr),
\quad\text{and} \quad
\widehat{V}_a^c\bigl(y_0,b_0,t_0^-\bigr) = \sup_{\,W_a^p \ge 0}
\,\widehat{V}^c\bigl(y_0,b_0, W_a^c, t_0\bigr).}}
\end{equation}
\begin{proposition}
\label{prop:equiv_lifted_timeconsistency}
The formulation \eqref{eq:terminal_c}--\eqref{eq:vhat_l_c} is equivalent, under the log transformation $y = \ln(s)$,
to (i) $TCMo_{t_0}(D,\gamma_o)$ in
\eqref{eq: TCMb}--\eqref{eq: TCMb constraints} when $W^c = W_o^c$,
and (ii) $TCMa_{t_0}(\alpha,\gamma_a)$ in \eqref{eq: TCMC}--\eqref{eq: TCMC constraints}
when $W^c = W_a^c$.
\end{proposition}
\noindent The proof of the proposition follows directly by substitution and application of the backward recursion.
\subsubsection{Localization and problem definition}
We adopt the same localized spatial domain $\Omega \subset \mathbb{R}^2$ and threshold domain $\Gamma \subset \mathbb{R}$ as in the pre-commitment setting (see Subsection~\ref{ssc:num_pro}). Specifically,
$\Omega$ is a finite region in the $(y,b)$-plane, with interior sub-domain $\Omega_{\myin}$ and boundary regions $\Omega_{y_{\min}}, \Omega_{y_{\max}}, \Omega_{b_{\max}}$;
$\Gamma \subset \mathbb{R}$ is chosen sufficiently large to contain optimal time-consistent threshold values $W^{c\ast}$.
On each boundary region in $\{\Omega_{y_{\min}}, \Omega_{y_{\max}}, \Omega_{b_{\max}}\}$, we impose the same conditions as specified in~\eqref{eq:boundary}–\eqref{eq:bmax} in Subsection~\ref{ssc:num_pro}. In particular, boundary conditions in $y$ are handled via approximate asymptotic conditions, while those in $b$ are treated by extrapolation.
We now define the time-consistent localized problem.
\begin{defn}[Localized time-consistent  formulations]
\label{def:glwb_c}
The value function of the time-consistent
\\
Mean–bPoE/CVaR problem at time $t_m \in \mathcal{T}$ and state $(s, b)$ is given by $V^c(s, b, t_m) = \widehat{V}^c(y, b, W_m^{c\ast}, t_m)$, where $y = \ln(s)$, and $\{W_m^{c\ast}, \ldots, W_{M-1}^{c\ast}\}$ is the sequence of optimal thresholds from $t_m$ onward, determined via~\eqref{eq:control_c}.
The associated optimal control is  $\mathcal{U}_m^{c\ast} = \{u_m^{c\ast}, \ldots, u_{M-1}^{c\ast}\}$.

The function $\widehat{V}^c(y, b, W^c, t)$ is defined on  $\Omega \times \Gamma \times \{\mathcal{T} \cup \{t_M = T\}\}$ as follows.
At each $t_{m} \in \mathcal{T}$, $\widehat{V}^c(\cdot, W^c, t_m)$ is given by
the rebalancing/threshold optimization \eqref{eq:intervention_c}--\eqref{eq:control_c},
where $\widehat{V}^c(\cdot, W^c, t_m^+)$ satisfies:
(i) the  integral \eqref{eq:integral_c}on $\Omega_{\myin}\times\{t_m^+\}$,
(ii) the boundary conditions \eqref{eq:boundary} and \eqref{eq:bmax} on
\mbox{$\{\Omega_{y_{\min}}, \Omega_{y_{\max}}, \Omega_{b_{\max}}\} \times\{t_{m}^+\}$}.
In addition, $\widehat{V}^c(\cdot, W^c, t)$ satisfies the terminal condition~\eqref{eq:terminal_c} on $\Omega \times \{t_M = T\}$ and the interest-settlement update~\eqref{eq:interest_c} on $\Omega \times \{t_{m+1}^-\}$.
\end{defn}
\subsubsection{Numerical schemes and convergence}
We now present the numerical scheme for approximating the function $\widehat{V}^c(y, b, W^c, t)$ from Definition~\ref{def:glwb_c}. The scheme uses the same domain discretization and convolution technique as in the pre-commitment case (Subsection~\ref{ssc:num_pro}), with the key distinction that both $u_m$ and $W^c$ are re-optimized at each $t_m\in \mathcal{T}$.

The localized domain $\Omega = [y_{\min}^\dagger, y_{\max}^\dagger] \times [0, b_{\max}]$ is discretized into a grid $\{(y_n, b_j)\}$, with interior sub-domain $\Omega_{\myin}$ and boundary regions $\{\Omega_{y_{\min}}, \Omega_{y_{\max}}, \Omega_{b_{\max}}\}$, where boundary conditions remain as in~\eqref{eq:boundary}–\eqref{eq:bmax}. The threshold domain $\Gamma \subset \mathbb{R}$ is discretized into $\{W_k\}_{k=0}^{N_w}$, identical to the pre-commitment case. The control set $\mathcal{Z} = [0,1]$ is likewise discretized to $\{u_i\}_{i=0}^{N_u}$.
For each fixed $(y_n,b_j,W_k,t_m)$, the exact function $\widehat{V}^c$ is approximated by a discrete solution $\widehat{V}_h^c(y_n,b_j,W_k,t_m^\star)$, where $t_m^\star\in \{t_m, t_m^+, t_m^-\}$ indicates the \mbox{evaluation time}.

\begin{itemize}[noitemsep, topsep=0pt, leftmargin=*]
\item  Terminal condition \eqref{eq:terminal_c}: we set
\EQ
\label{eq:terminal_c_sym}
\widehat{V}_h^c(y_n, b_j, W_k, T) = f\bigl(e^{y_n}+b_j,\;W_k\bigr).
\EN

\item Interest settlement \eqref{eq:interest_c}: is handled by same step  \eqref{eq:interest*_sym} in the pre-commitment scheme, via
interpolation
\EQ
\label{eq:interest*_c_sym}
\widehat{V}_h^c(y_n, b_j, W_k, t_{m+1}^-)  =  \widehat{V}_h^c(y_n, b_je^{ r\Delta t}, W_k, t_{m+1}).
\EN

\item Time-advancement via integral \eqref{eq:integral_c}: we use the same discrete convolution approach in \eqref{eq:scheme*_sym}:
\EQ
\label{eq:scheme*_c_sym}
\widehat{V}_h^c(y_n,b_j,W_k,t_m^+) =
     \sum_{l=-N^\dagger/2}^{\,N^\dagger/2}
        \varphi_l\,\,
        g\bigl(y_n - y_l,\,\Delta t; N_g\bigr)\,\,
        \widehat{V}_h^c\bigl(y_l,\,b_j,\,W_k,\,t_{m+1}^-\bigr).
\EN

\item  On the boundary regions $\{\Omega_{y_{\min}},\,\Omega_{y_{\max}},\,\Omega_{b_{\max}}\}\times \Gamma\times \{t_m^+\}$, we impose the same asymptotic or extrapolation conditions from
    \eqref{eq:bdry_adv_sym} and \eqref{eq:bmax}.

\item {\sifin{Rebalancing/threshold re-optimization \eqref{eq:intervention_c}--\eqref{eq:control_c}:
at each node $(y_n,b_j,t_m)$, both $u_m$ and $W^c$ must be re-optimized.
Numerically, we search over the discrete control grid $\{u_i\}_{i=0}^{N_u}$
and threshold grid $\{W_k\}_{k=0}^{N_w}$ and select the pair that optimize
(minimizes for bPoE, maximizes for CVaR) the discrete function $\widehat{V}_{h}^c(\cdot)$ at $t_m^+$:
\begin{equation}
\label{eq:outer_search_sim_c_sym}
\begin{cases}
\displaystyle
\bigl(u_{m,h}^{c\ast}(y_n,b_j),\,W_{m,h}^{c\ast}(y_n,b_j)\bigr)
\in
\underset{\substack{u_i \in \{u_0,\dots,u_{N_u}\}\\[1pt]
                    W_\iota \in \{W_0,\dots,W_{N_w}\}}}{\arg\min}
\;
\widehat{V}_{o,h}^c\bigl(y_n^+(u_i), b_j^+(u_i), W_\iota, t_m^+\bigr),
& \text{Mean--bPoE},\\[10pt]
\displaystyle
\bigl(u_{m,h}^{c\ast}(y_n,b_j),\,W_{m,h}^{c\ast}(y_n,b_j)\bigr)
\in
\underset{\substack{u_i \in \{u_0,\dots,u_{N_u}\}\\[1pt]
                    W_\iota \in \{W_0,\dots,W_{N_w}\}}}{\arg\max}
\;
\widehat{V}_{a,h}^c\bigl(y_n^+(u_i), b_j^+(u_i), W_\iota, t_m^+\bigr),
& \text{Mean--CVaR}.
\end{cases}
\end{equation}
Here, $y_n^+(u_i)=\ln\bigl(u_i(e^{y_n}+b_j+q_m)\bigr)$ and
$b_j^+(u_i)=(1-u_i)(e^{y_n}+b_j+q_m)$, as in \eqref{eq:sbplus}.
This step yields a numerically computed optimal pair
$(u_{m,h}^{c\ast}, W_{m,h}^{c\ast})$ at each node $(y_n, b_j, t_m)$:
$u_{m,h}^{c\ast} = u_{m,h}^{c\ast}(w_{n, j})$ and
$W_{m,h}^{c\ast} = W_{m,h}^{c\ast}(y_n,b_j)$, where
$w_{n, j} = e^{y_n} + b_j + q_m$ is the total wealth after cash
injection at this node.

Given the optimal control $u_{m,h}^{c\ast}(y_n,b_j)$ obtained from
\eqref{eq:outer_search_sim_c_sym}, the discrete analogue of the
intervention update \eqref{eq:intervention_c} (from $t_m^{+}$ to $t_m$)
is then applied for each threshold grid point $W_k$:
\begin{equation}
\label{eq:whVhc}
\widehat{V}_h^c(y_n,b_j,W_k,t_m)
=
\widehat{V}_h^c\bigl(
  y_n^+\bigl(u_{m,h}^{c\ast}(y_n,b_j)\bigr),
  b_j^+\bigl(u_{m,h}^{c\ast}(y_n,b_j)\bigr),
  W_k,
  t_m^+
\bigr),
\quad k = 0,\ldots,N_w.
\end{equation}}}
%

\item After \eqref{eq:outer_search_sim_c_sym}-\eqref{eq:whVhc} are completed, the time-consistent numerical value function
$V_h^c(s_n,b_j,t_m)$, with $s_n = e^{y_n}$, is given by
{\sifin{
\begin{equation}
\label{eq:incept_pre_c_sym}
V_h^c(s_n,b_j,t_m)
  = \widehat{V}_h^c\bigl(y_n,b_j,W_{m,h}^{c\ast}(y_n,b_j),t_m\bigr),
  \qquad s_n = e^{y_n}.
\end{equation}
Here,  $W_{m,h}^{c\ast}(y_n,b_j)$ is the locally optimal threshold at node
$(y_n,b_j,t_m)$, obtained from the discrete optimization
\eqref{eq:outer_search_sim_c_sym}.}}



\item Once the optimal pair $(u_{m,h}^{c\ast}, W_{m,h}^{c\ast})$ is computed at each node $(y_n, b_j, t_m)$, we approximate $\mathbb{E}[W_T]$ as in the pre-commitment case (see Remark~\ref{rm:expectation}).
\end{itemize}
Let $\Omega^h$ be the computational grid parameterized by $h$, with $\Omega^h \to \Omega$ as $h \to 0$, and let $\Omega_{\myin}^h$ denote the interior subgrid. To map the log-domain $\Omega_{\myin}$ to the original $(s, b)$ coordinates, define
$\widetilde{\Omega}_{\myin} := \{(s, b) \mid (\ln s, b) \in \Omega_{\myin}\}$,
and similarly define the discrete version $\widetilde{\Omega}_{\myin}^h$.
We define the exact time-consistent optimal threshold in original coordinates by $\widetilde{W}_m^{c\ast}(s, b) := W_m^{c\ast}(y, b)$, where $y = \ln(s)$ and $W_m^{c\ast}(y, b)$ is the exact threshold on the $(y, b)$ grid. The corresponding discrete optimal threshold is denoted by $\widetilde{W}_{m,h}^{c\ast}(s_h, b_h)$.

We now state the convergence result in these original variables.
\begin{theorem}[Time-consistent scheme convergence]
\label{thm:convergence_tc}
Consider the time-consistent Mean--bPoE/CVaR problem defined in
Definition~\ref{def:glwb_c}.
Suppose the threshold domain $\Gamma$ is chosen sufficiently large to contain
all optimal time-consistent thresholds.
Also suppose linear interpolation is used for the intervention (rebalancing) step.
As the discretization parameter $h \to 0$ (i.e.\ $N_y^\dagger, N_b, N_u, N_w, N_g \to \infty$), the numerical scheme for $V_h^c(\cdot)$, defined
 in~\eqref{eq:terminal_c_sym}–\eqref{eq:incept_pre_c_sym} using a log-transformation
 on $s$, converges in both value function and optimal thresholds.
\begin{itemize}[noitemsep, topsep=0pt, leftmargin=*]
\item
Value function convergence: For any fixed $(s', b', t_m) \in \widetilde{\Omega}_{\myin} \times \mathcal{T}$,
\EQ
\label{eq:time_consistent_v_conv}
\lim_{\substack{h \to 0 \\ (s_h, b_h) \to (s', b')}}
\bigl| V_h^c(s_h, b_h, t_m) - V^c(s', b', t_m) \bigr| = 0,
\quad \text{where } (s_h, b_h) \in \widetilde{\Omega}_{\myin}^h \text{ for each } h.
\EN
\item Threshold convergence: For any fixed $(s', b', t_m) \in \widetilde{\Omega}_{\myin} \times \mathcal{T}$, let $\widetilde{W}_m^{c\ast}(s', b')$ be an associated optimal threshold.

    Then, for any sequence $\{(s_h, b_h)\}$ such that $(s_h, b_h) \in \widetilde{\Omega}_{\myin}^h$ for each $h$, and $(s_h, b_h) \underset{h \to 0}{\rightarrow} (s', b')$, the corresponding computed thresholds $\{\widetilde{W}_{m,h}^{c\ast}(s_h, b_h)\}$ have a subsequence converging to $\widetilde{W}_m^{c\ast}(s', b')$.
\end{itemize}
\end{theorem}
\noindent A detailed proof is given in Appendix~\ref{app:convergence_tc}.

\section{Numerical results}
\label{sc:num}
\subsection{Empirical data and calibration}
To calibrate the parameters specified in dynamics~\eqref{eq: dSt} and~\eqref{eq: dBt}, we employ the same data sources and calibration techniques as described in \cite{DM2016semi, ForsythVetzal2016, PMVS2021c}. The risky asset data is based on daily total return series of the VWD index from the Center for Research in Security Prices (CRSP), covering the period 1926:1 - 2014:12.\footnote{The results presented here were calculated based on data from Historical Indexes, \copyright 2015 Center for Research in Security Prices (CRSP), The University of Chicago Booth School of Business. Wharton Research Data Services was used in preparing this article. This service and the data available thereon constitute valuable intellectual property and trade secrets of WRDS and/or its third-party suppliers.} This is a capitalization-weighted index of all domestic stocks on major US exchanges, including dividends and other distributions in the total return. The risk-free asset is represented by 3-month US T-bill rates covering the period 1934:1-2014:12.\footnote{See \texttt{http://research.stlouisfed.org/fred2/series/TB3MS}.} To account for the impact of the 1929 crash, we supplement this data with short-term government bond yields from the National Bureau of Economic Research (NBER) over the period 1926:1 - 1933:12.\footnote{See \texttt{http://www.nber.org/databases/macrohistory/contents/chapter13.html}.} To ensure all parameters correspond to their inflation-adjusted counterparts, the annual average CPI-U index (inflation for urban consumers) from the US Bureau of Labor Statistics is used to adjust the time series for inflation.\footnote{CPI data from the U.S.\ Bureau of Labor Statistics. Specifically, we use the annual average of the all urban consumers (CPI-U)index. See \texttt{http://www.bls.gov/cpi}.} The resulting calibrated parameters are provided in Table~\ref{tab: parameters for Kou3}.

\noindent
\begin{table}[!hbt]
\vspace*{-0.4cm}
\caption{Calibrated parameters for asset dynamics. {\sifin{Sample period 1926:1 to 2014:12.}}}
\label{tab: parameters for Kou3}
\vspace*{-0.2cm}
\centering{}
\begin{tabular}{c c c c c c c }
\hline
$\mu$ & $\sigma$ & $\lambda$ & $p_{up}$ & $\eta_{1}$ & $\eta_{2}$ & $r$ \\
\hline
0.0874 & 0.1452 & 0.3483 & 0.2903 & 4.7941 & 5.4349 & 0.00623\\
\hline
\end{tabular}
\end{table}


To illustrate the accumulation phase of a DC plan, we consider a 35-year-old investor with an annual salary of \$100,000 and the total contribution to the plan account is 20\% of the salary each year. This investor plans to retire at age 65, yielding a 30-year money saving horizon \cite{PA2020a}. The investment scenario considered is summarized in Table~\ref{tab: investment scenario from Peter's Mean-CVaR}, and the numerical discretization parameters used in this work are listed in Table~\ref{tab: localized parameters from Peter's Mean-CVaR}.


\noindent
\begin{table}[!hbt]
\vspace*{-0.4cm}
\centering
  \begin{minipage}[c]{0.48\textwidth}
    \caption{Investment scenario}
    \label{tab: investment scenario from Peter's Mean-CVaR}
    \vspace*{-0.2cm}
    \centering
    \begin{tabular}{l|c}
    \hline
    Investment horizon $T$ & 30 years \\
    Rebalancing frequency & yearly\\
    Initial wealth $W_{0}$ & 0\\
    Cashflow $\left\{q_{m}\right\}_{m=0,1,\ldots, 29}$ & 20,000\\
    \hline
    \end{tabular}
  \end{minipage}
  \hfill
  \begin{minipage}[c]{0.48\textwidth}
    \caption{Discretization parameters}
    \label{tab: localized parameters from Peter's Mean-CVaR}
    \vspace*{-0.2cm}
    \centering
    \begin{tabular}{c|c||c|c}
    \hline
    $y_{\min}^{\dagger}$ & $\log{\left(10^{5}\right)}-16$ & $b_{\max}$ & $5\times 10^{8}$ \\
    $y_{\min}$ & $\log{\left(10^{5}\right)}-8$ & $W_{a}^{\max}$ & $5\times 10^{8}$ \\
    $y_{\max}$ & $\log{\left(10^{5}\right)}+8$ & $W_{o}^{\max}$ & $5\times 10^{8}$ \\
    $y_{\max}^{\dagger}$ & $\log{\left(10^{5}\right)}+16$ & $N_{b}$ & $333$ \\
    $N_{y}$ & $512$ & $N_{w}$ & $333$ \\
    $N_{y}^{\dagger}$ & $1024$ & $N_{u}$ & $333$  \\
    \hline
    \end{tabular}
  \end{minipage}
\vspace*{-0.4cm}
\end{table}

\subsection{Precommitment Mean--bPoE and Mean--CVaR}
We now numerically illustrate the equivalence between the
pre-commitment Mean–bPoE and Mean–CVaR formulations, as established in Lemmas~\ref{lem: Equiv_CVaR_to_bPoE} and~\ref{lem: Equiv_bPoE_to_CVaR}. Specifically, in Subsection~\ref{ssc:Investment_outcomes}, we examine terminal wealth distributions and key performance metrics, including the mean, CVaR, bPoE, and the 5th, 50th, and 95th percentiles. Subsection~\ref{ssc:Optimal_controls} analyzes the structure of optimal investment strategies, and Subsection~\ref{ssc:Efficient_frontiers} presents the efficient frontiers.
{\sifin{
To support these comparisons, Subsection~\ref{ssc:Convergence_behaviour}
first reports a representative convergence check to indicate the magnitude of
discretization error.}}

{\sifin{\subsubsection{Convergence behaviour}
\label{ssc:Convergence_behaviour}
We take the confidence level
$\alpha=0.05$ (a standard choice in practice) and scalarization parameter
$\gamma_a=10$. Table~\ref{tab:convergence_results} reports a grid refinement test
for the pre-commitment Mean--CVaR problem $PCMa_{t_0}(\alpha,\gamma_a)$.
For each grid size, we first compute the optimal feedback policy and its
associated threshold $W_a^{p\ast}$ using the proposed numerical scheme, and then
evaluate this fixed policy by Monte Carlo simulation with $2.56\times 10^{6}$
paths to estimate $\mathbb{E}[W_T]$ and $\mathrm{CVaR}_\alpha(W_T)$.
The finest grid in Table~\ref{tab:convergence_results} corresponds to the
discretization listed in Table~\ref{tab: localized parameters from Peter's Mean-CVaR}.

\begin{table}[!htbp]
\caption{{\sifin{Convergence test for $PCMa_{t_0}\!\left(\alpha,\gamma_a\right)$ with
$\alpha=0.05$ and $\gamma_a=10$.
Grid size denotes the discretization used in the numerical scheme ($N_y\times N_b$),
where $N_y$ is the number of nodes in the $y=\log S$ direction and $N_b$ is the
number of nodes in the $B$ direction; $N_w$ and $N_u$ are set equal to $N_b$
for all refinement levels. Units: thousands of dollars (real).}}}
\label{tab:convergence_results}
\vspace*{-0.1cm}
\centering
\begin{tabular}{c c c c}
\hline
Grid size & $E[W_T]$ & $\mathrm{CVaR}_{\alpha}(W_T)$ & $W_a^{p\ast}$ \\
\hline
$128\times 84$  & $2387.57$ & $663.62$ & $723.20$ \\
$256\times 167$ & $2432.20$ & $666.82$ & $731.20$ \\
$512\times 333$ & $2441.27$ & $668.81$ & $750.00$ \\
\hline
\end{tabular}
\vspace*{-0.25cm}
\end{table}
As seen in Table~\ref{tab:convergence_results}, between the two finest grids the
changes in $E[W_T]$ and $\mathrm{CVaR}_\alpha(W_T)$ are below $0.4\%$ and about
$0.3\%$, respectively. The optimal threshold $W_a^{p\ast}$ is chosen from a
discrete threshold grid, so shifts of a few grid points under refinement are
expected even when the underlying objective is well resolved. A similar
grid refinement check for the $PCMo$ formulation (not reported here) shows
changes of comparable magnitude, so for brevity we only present the $PCMa$ case.
Overall, this indicates that discretization error in the main reported
risk--return metrics at the finest grid is small (below about one percent) for
this representative case, which is adequate for the numerical comparisons
reported in this section. Unless otherwise stated, all numerical results in this
paper use the discretization listed in
Table~\ref{tab: localized parameters from Peter's Mean-CVaR}.}}

\subsubsection{Investment outcomes}
\label{ssc:Investment_outcomes}
{\sifin{We now report investment outcomes, keeping
$\alpha=0.05$ and $\gamma_a=10$ as in the convergence study above.}}
Solving $PCMa_{t_{0}}\left( \alpha, \gamma_{a}\right)$ yields the optimal threshold and optimal control pair $\left( W_{a}^{p \ast}, \ \mathcal{U}_{0, a}^{p \ast} \right)$ and a point
$\left(\mathcal{C}_{a}^{\ast}=\text{CVaR}_{\alpha}\left(W_{T}\right), \  \mathcal{E}_{a}^{\ast}=E\left[W_{T}\right]\right)$ in the Mean-CVaR scalarization optimal set $\mathcal{S}_{a} \left(\alpha, \gamma_{a}\right)$. We then let $D=\mathcal{C}_{a}^{\ast}$ and compute the corresponding $\gamma_{o}$ via the relationship~\eqref{eq:gamma_o} proposed in Lemma~\ref{lem: Equiv_CVaR_to_bPoE}. Knowing $D$ and $\gamma_{o}$ allows us to solve $PCMo_{t_{0}}\left(D, \gamma_{o}\right)$ and obtain a point $\left(\mathcal{B}_{o}^{\ast}=\text{bPoE}_{D}\left(W_{T}\right), \  \mathcal{E}_{o}^{\ast}=E\left[W_{T}\right]\right)$ in the Mean-bPoE scalarization optimal set $\mathcal{S}_{o} \left(D, \gamma_{o}\right)$.

Table~\ref{tab: Precommitment results} presents the optimal thresholds and relevant statistics obtained by conducting the numerical experiment described above.
{\purple{The statistics are computed via Monte Carlo simulation of the portfolio using the optimal control from the numerical scheme. Overall, the reported investment outcomes are virtually identical, with relative errors under 1.3\% across all metrics, and under 1\% for most---consistent with Monte Carlo simulation error.}}

\noindent
\begin{table}[!hbt]
\vspace*{-0.4cm}
\caption{Investment outcomes of $PCMa_{t_{0}}\left( \alpha, \gamma_{a} \right)$ and $PCMo_{t_{0}}\left(D, \gamma_{o} \right)$. $\alpha=0.05$, $D=668.81$, $\gamma_{a} = 10$, $\gamma_{o}=1.6238\times 10^{4}$. Results computed using Monte Carlo simulations with $2.56 \times 10^{6}$ paths. Units: thousands of dollars (real).}
\label{tab: Precommitment results}
\vspace*{-0.2cm}
\centering{}
\begin{tabular}{|c|c|c|c|c|c|c|c|}
\hline
 & $E\left[W_{T}\right]$ & $\text{CVaR}_{\alpha}\left(W_{T}\right)$ &
$\text{bPoE}_{D}\left(W_{T}\right)$  & $W_{a \ (o)}^{p \ast}$ & 5th $W_T$ & 50th $W_T$ & 95th $W_T$\\
\hline
$PCMa$ & $2441.27$ & $668.81$ & $5.00 \%$ & $750.00$ & $759.14$ & $1107.78$ & $7977.45$ \\
\hline
$PCMo$ & $2462.99$ & $665.45$ & $5.17 \%$ & $759.39$ & $766.31$ & $1121.58$ & $8024.51$ \\
\hline
\end{tabular}
\vspace*{-0.2cm}
\end{table}

Specifically, using the input $D = \text{CVaR}_{\alpha}(W_T)$, the solution to $PCMo$ achieves the same expected terminal wealth as $PCMa$, and the resulting $\text{bPoE}_{D}(W_T)$ matches the pre-specified confidence level $\alpha$. This confirms the mapping from the point $\left(\mathcal{C}_a^{\ast}, \ \mathcal{E}_a^{\ast}\right) \in \mathcal{S}_a(\alpha, \gamma_a)$ to the corresponding point $\left(\mathcal{B}_o^{\ast} = \alpha, \ \mathcal{E}_o^{\ast} = \mathcal{E}_a^{\ast}\right) \in \mathcal{S}_o(D = \mathcal{C}_a^{\ast}, \gamma_o)$.
Moreover, both formulations yield the same optimal threshold, $W_a^{p\ast} = W_o^{p\ast}$, which numerically coincides with the 5th percentile of the terminal wealth distribution—confirming the interpretation of the optimal threshold as $\text{VaR}_{\alpha}(W_T)$. Although our experiment proceeds by mapping from $PCMa$ to $PCMo$, the reverse direction can be carried out analogously.

\noindent
\begin{minipage}[t]{0.5\textwidth}
Finally, as shown in Figure~\ref{fig:PCMa_PCMo_comp}, the terminal wealth distributions under $PCMa$ and $PCMo$ nearly overlap, indicating that their investment outcomes are essentially indistinguishable across the entire distribution. This visual agreement complements the close alignment of all key statistics in Table~\ref{tab: Precommitment results}, and confirms the theoretical equivalence established by Lemmas~\ref{lem: Equiv_CVaR_to_bPoE} and~\ref{lem: Equiv_bPoE_to_CVaR}.

\subsubsection{Optimal rebalancing controls}
\label{ssc:Optimal_controls}
Having shown that both $PCMa$ and $PCMo$ lead to numerically identical investment outcomes, we now examine their optimal rebalancing controls.
\end{minipage}
\hfill
\begin{minipage}[t]{0.5\textwidth}
  \centering\raisebox{\dimexpr \topskip-\height}{%
  \includegraphics[width=0.8\textwidth, height=0.5\textwidth]{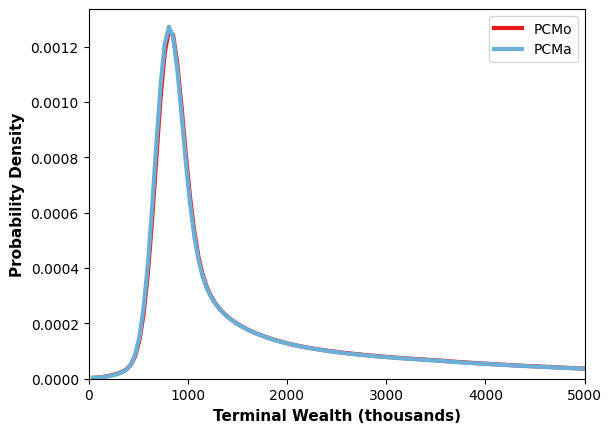}}
\captionof{figure}{Terminal wealth distribution comparison--PCMa vs.\ PCMo}
\label{fig:PCMa_PCMo_comp}
\end{minipage}

\vspace*{+0.25cm}
\noindent Figure~\ref{Fig: Precommitment optimal control heat maps} presents the heat maps of the optimal controls for both frameworks. A direct side‐by‐side comparison reveals virtually identical controls under these two frameworks.
\begin{figure}[!htb]
    \centering
    \subfigure[$PCMa$ optimal control]{
        \label{fig: PCMa_Kou3_heatmap}
        \includegraphics[width = 0.45\textwidth]{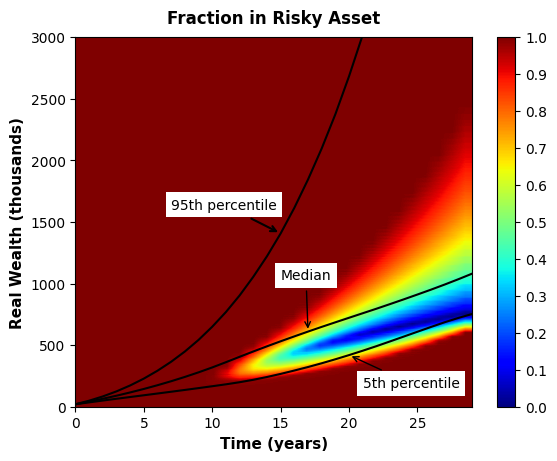}}
    \subfigure[$PCMo$ optimal control]{
        \label{fig: PCMo_Kou3_heatmap}
        \includegraphics[width = 0.45\textwidth]{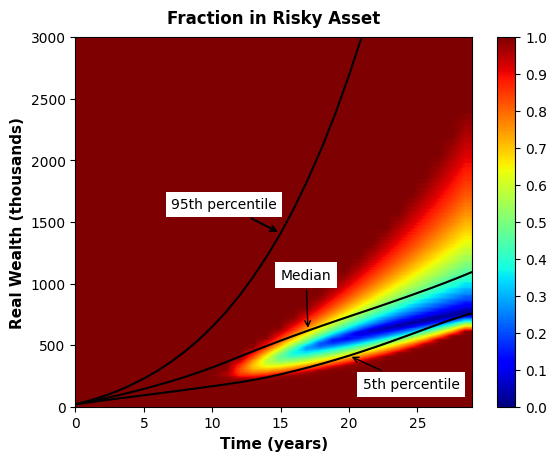}}
\caption{Precommitment optimal control heat maps}
\label{Fig: Precommitment optimal control heat maps}
\vspace*{-0.4cm}
\end{figure}
These heat maps illustrate the optimal proportion in the risky asset as a function of current real wealth and time. The optimal proportion can be determined by comparing the color of each point to the legend on the right-hand side. Redder colors indicate a higher proportion should be invested in the risky asset, while bluer colors suggest allocating more funds to the risk-free asset.

We observe from those heatmaps in Figure~\ref{Fig: Precommitment optimal control heat maps}  that initially the strategies recommend allocating all funds to the risky asset to maximize return. However, about 15 years later, the investment strategies begin to depend on the current wealth realized. We also note that a distinct triangular band appears in the second half of the investment horizon. This band indicates that investors should move their funds from risky to risk-free asset as their realized wealth approaches \$0.6-0.8 million. Recall that the optimal threshold $W^{p \ast}$ is around \$0.75 million, see Table~\ref{tab: Precommitment results}. Financially, this threshold can be interpreted as the dollar level below which investors want to avoid falling. Therefore, if the current wealth is below this threshold, a higher allocation to risky asset is suggested to pursue the return. In contrast, if the current wealth is already at the nearby level, funds should be shifted to risk-free asset to ensure that the terminal wealth does not fall below this amount. The reason for the reallocation to risky asset in the upper right corner is that, in this area, investors have accumulated \$2-3 million, which is far above the threshold level. Thus, they do not worry about falling into a poor situation and can allocate all funds to the risky asset to pursue enhanced return.

{\sifin{
With the numerically estimated optimal controls in
Figure~\ref{Fig: Precommitment optimal control heat maps} and the numerical
thresholds in Table~\ref{tab: Precommitment results}, the results indicate that the $PCMa$ and $PCMo$ solutions are practically indistinguishable. This is consistent with Lemmas~\ref{lem: Equiv_CVaR_to_bPoE} and~\ref{lem: Equiv_bPoE_to_CVaR}, which show that, for each corresponding pair of efficient points, there exists a threshold/control pair, not necessarily unique, that is simultaneously optimal for both $PCMa$ and $PCMo$ and attains those points.}}

\subsubsection{Efficient frontiers}
\label{ssc:Efficient_frontiers}
Table~\ref{tab: Precommitment results} and Figure~\ref{Fig: Precommitment optimal control heat maps} illustrate the detailed one-to-one correspondence between $PCMa$ and $PCMo$ for a specific scalarization-optimal set. To demonstrate that this correspondence holds across the full efficient frontier, we vary the value of $\gamma$ and repeat the procedure described above. {\purple{The resulting efficient frontiers are reported in Figure~\ref{Fig: Equivalent efficient frontier}.}}

{\purple{Although the two efficient frontiers shown in Figure~\ref{Fig: Equivalent efficient frontier} appear as mirror images, they plot different risk measures along the $x$-axis. In {\sifin{Figure~\ref{Fig: Equivalent efficient frontier}(a)}},
risk is quantified by $\text{CVaR}_{\alpha}(W_T)$; increasing $\gamma_a$ places greater emphasis on risk reduction, resulting in strategies that increase $\text{CVaR}_{\alpha}(W_T)$ (i.e.\ reduce downside risk) at the cost of lower expected terminal wealth $E[W_T]$. This drives the ($\text{CVaR}_{\alpha}(W_T), E[W_T])$ efficient frontier downward and to the right. In contrast, {\sifin{Figure~\ref{Fig: Equivalent efficient frontier}(b)}}
uses $\text{bPoE}_{D}(W_T)$ as the risk measure; increasing $\gamma_o$ places more emphasis on lowering $\text{bPoE}_{D}(W_T)$) at the cost of reduced $E[W_T]$, which drives the ($\text{bPoE}_{D}(W_T), E[W_T]$) efficient frontier downward and to the left.}}

\begin{figure}[!hbt]
    \centering
    \subfigure[Direction: $PCMa$ to $PCMo$]{
        \label{fig: EF_CVaR_bPOE}
        \includegraphics[width = 0.45\textwidth]{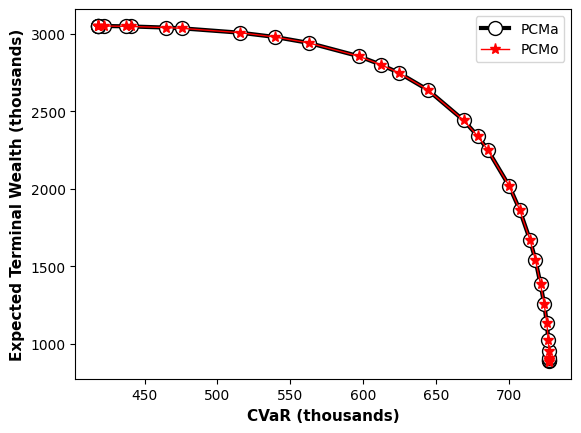}}
    \subfigure[Direction: $PCMo$ to $PCMa$]{
        \label{fig: EF_bPOE_CVaR}
        \includegraphics[width = 0.45\textwidth]{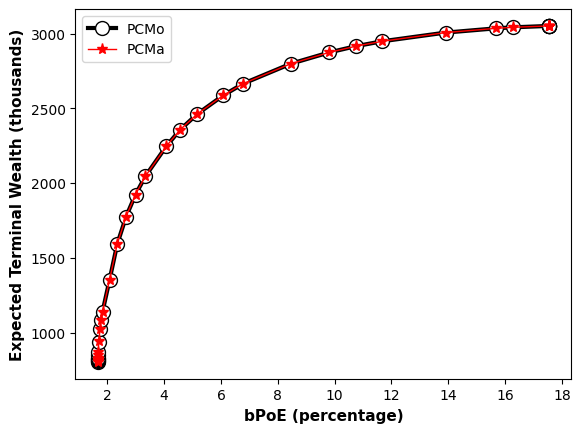}}
\caption{Equivalent efficient frontiers of $PCMa$ and $PCMo$}
\label{Fig: Equivalent efficient frontier}
\vspace*{-0.4cm}
\end{figure}
{\sifin{Figure~\ref{Fig: Equivalent efficient frontier}(a)}},
each black point generated by $PCMa$ is mapped via Lemma~\ref{lem: Equiv_CVaR_to_bPoE} into a red point of $PCMa$. Conversely, Figure~\ref{Fig: Equivalent efficient frontier}(b)
starts from $PCMo$ then applies Lemma~\ref{lem: Equiv_bPoE_to_CVaR} to map to $PCMa$ points. The perfect overlap of these two frontiers in both mapping directions provides excellent numerical confirmation of the theoretical equivalence between $PCMa$ and $PCMo$ established in Lemma~\ref{lem: Equiv_CVaR_to_bPoE} and~\ref{lem: Equiv_bPoE_to_CVaR}.

\subsection{Time-consistent Mean--bPoE and Mean--CVaR}
{\sifin{
Unless otherwise stated, all time-consistent results below use the finest
discretization listed in Table~\ref{tab: localized parameters from Peter's Mean-CVaR}.
In addition, we carried out a grid refinement check analogous to
Table~\ref{tab:convergence_results} for representative time-consistent instances
(not reported here) and observed changes of comparable magnitude in the main
risk--return metrics. For brevity, we do not include a separate convergence
table for the time-consistent case.}}
\subsubsection{Investment outcomes}
{\sifin{We now report investment outcomes for the time-consistent Mean--bPoE and
Mean--CVaR formulations.}} In the time-consistent setting, Remark~\ref{rm:no_equiv} suggests that no global parameter mapping can preserve the same efficient frontier and optimal control/threshold structure between the TCMa and TCMb formulations.
Nevertheless, we compare the investment outcomes of $TCMa$ and $TCMo$ by first matching their expected terminal wealth and then comparing other statistics, such as CVaR, bPoE, and percentiles. Specifically, for $TCMo$, we adopt the same disaster level $D$ and parameter $\gamma_o$ as in $PCMo$, enabling a direct comparison between $PCMo$ and $TCMo$. We then solve the $TCMo$ problem using the proposed numerical scheme to obtain its expected terminal wealth. For $TCMa$, we numerically determine the corresponding $\gamma_a$ using Newton's method to match this value.
\noindent
\begin{table}[!hbt]
\vspace*{-0.1cm}
\caption{Investment outcomes of $TCMa_{t_{0}}\left( \alpha, \gamma_{a} \right)$ and $TCMo_{t_{0}}\left(D, \gamma_{o} \right)$. $\alpha=0.05$, $D=668.81$, $\gamma_{a} = 0.4$, $\gamma_{o}=1.6238\times 10^{4}$. Results computed using Monte Carlo simulations with $2.56 \times 10^{6}$ paths. Units: thousands of dollars (real).}
\label{tab: Time-consistent results}
\vspace*{-0.1cm}
\centering{}
\begin{tabular}{|c|c|c|c|c|c|c|}
\hline
 & $E\left[W_{T}\right]$ & $\text{CVaR}_{\alpha}\left(W_{T}\right)$ &
$\text{bPoE}_{D}\left(W_{T}\right)$  & 5th $W_T$ & 50th $W_T$ & 95th $W_T$\\
\hline
$TCMa$ & $2140.11$ & $476.49$ & $19.00 \%$ & $575.84$ & $1618.12$ & $5508.10$ \\ 
\hline
$TCMo$ & $2133.97$ & $672.50$ & $4.62 \%$ & $717.92$ & $952.42$ & $7176.15$ \\
\hline
\end{tabular}
\vspace*{-0.1cm}
\end{table}

\noindent
\begin{minipage}[t]{0.5\textwidth}
{\purple{Table~\ref{tab: Time-consistent results} presents the investment outcomes of $TCMa$ and $TCMo$. With the same expected terminal wealth, $TCMa$ produces a substantially lower $\text{CVaR}_{\alpha}\left(W_{T}\right)$ and a much higher $\text{bPoE}_{D}\left(W_{T}\right)$. Since a larger $\text{CVaR}_{\alpha}\left(W_{T}\right)$ and a smaller $\text{bPoE}_{D}\left(W_{T}\right)$ indicate a more favorable outcome, $TCMo$ dominates $TCMa$ in both the Mean-CVaR and Mean-bPoE sense. Thus, although the expected terminal wealth is matched, the differing risk levels imply that the scalarization optimal sets of $TCMa$ and $TCMo$ are not equivalent--as $PCMa$ and $PCMo$ are in the pre-commitment case--and no pair of efficient points coincides.}}
\end{minipage}
\hfill
\begin{minipage}[t]{0.5\textwidth}
  \centering\raisebox{\dimexpr \topskip-\height}{%
  \includegraphics[width=0.9\textwidth, height=0.6\textwidth]{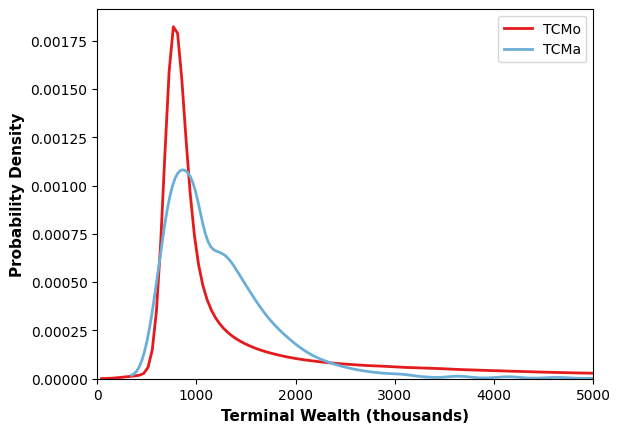}}
\captionof{figure}{Terminal wealth distribution comparison--TCMa vs.\ TCMo}
\label{fig: TCMa_TCMo}
\end{minipage}

\vspace*{+0.25cm}
{\purple{We note that among the other statistics reported in Table~\ref{tab: Time-consistent results}, the only metric in which $TCMa$ outperforms $TCMo$ is the median of terminal wealth. To investigate this further, we compare their terminal wealth density functions, shown in Figure~\ref{fig: TCMa_TCMo}. An interesting observation from this figure is that $TCMa$ concentrates more probability mass around the middle range (e.g., \$0.5–2 million), whereas $TCMo$ spreads its density more into the tails.

This behavior reflects the key difference between these two risk measures.
In $TCMa$, the strategy aims to minimize the average of the worst $\alpha$ fraction
of outcomes at each rebalancing time, hence intuitively  encouraging the portfolio to compress its distribution toward the middle to reduce extreme downside risk, which in turn elevates the median metric. By contrast, $TCMo$ enforces a fixed disaster level $D$, ensuring the wealth avoids falling below this specific wealth floor, while allowing the remainder of the distribution (including the median and upper tail) to spread more widely, if needed.
Consequently, the portfolio may accept greater variance in the median or mid-range outcomes so that it can preserve or even enhance both the upside (e.g.\ high-end returns) and downside protection. The result is a potentially lower median outcome than Mean–CVaR, but better tail performance at the extremes (both upside and downside).

This suggests that bPoE is a strictly tail-oriented measure that prioritizes
guarding against catastrophic shortfalls while still permitting meaningful upside exposure--making it especially appealing for investors focused on long-term wealth security rather than distributional tightness.}}

{\purple{
Further comparison between the $TCMo$ results in Table~\ref{tab: Time-consistent results} and the $PCMo$ results in Table~\ref{tab: Precommitment results} indicates that, under the same inputs $D$ and $\gamma_{o}$, the investment outcomes of $TCMo$ closely resemble those of its precommitment counterpart. In particular, the resulting $\text{CVaR}_{\alpha}\left(W_{T}\right)$ of $TCMo$ is nearly equal to the input disaster level $D$, and its $\text{bPoE}_{D}\left(W_{T}\right)$ is close to the pre-specified confidence level $\alpha = 5\%$. This consistency across the $TCMo$ and $PCMo$ formulations motivates the next investigation, which focuses on a detailed heatmap analysis of the $TCMo$ optimal rebalancing control and its comparison with that of $PCMo$.}}

\subsubsection{Optimal rebalancing controls}
{\purple{Figure~\ref{Fig: Time-consistent optimal control heat maps} displays the heat maps of the optimal controls for $TCMa$ and $TCMo$. In contrast to the precommitment case—where the optimal control heatmaps of $PCMa$ and $PCMo$ are identical—the control behaviors of $TCMa$ and $TCMo$ differ significantly.}}
\begin{figure}[!hbt]
    \centering
    \subfigure[$TCMa$ optimal control]{
    \label{fig: TCMa_Kou3_heatmap}
        \includegraphics[width=0.48\textwidth]{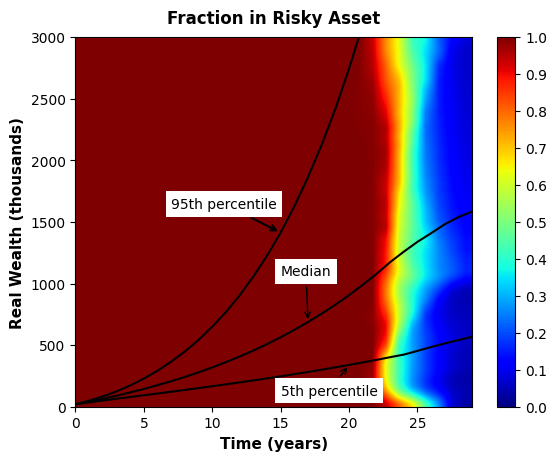}}
    \subfigure[$TCMo$ optimal control]{
        \label{fig: TCMo_Kou3_heatmap}
        \includegraphics[width=0.48\textwidth]{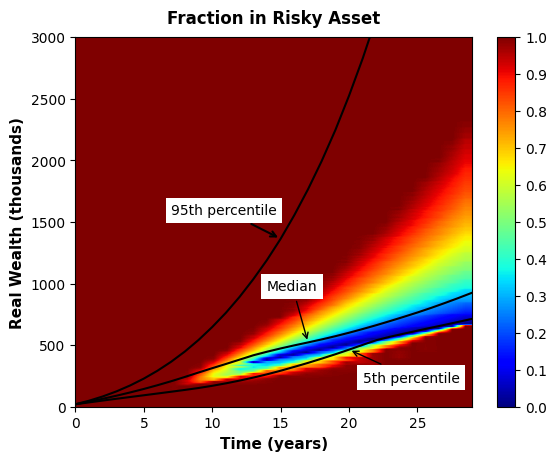}}
\caption{Time-consistent optimal control heat maps}
\label{Fig: Time-consistent optimal control heat maps}
\vspace*{-0.4cm}
\end{figure}
\noindent
{\purple{We highlight two key observations below.
\begin{itemize}[noitemsep, topsep=0pt, leftmargin=*]
\item
As shown in {\sifin{Figure~\ref{Fig: Time-consistent optimal control heat maps}(a)}},
the $TCMa$ control depends primarily on time and exhibits a weak sensitivity to current wealth. Although our numerical experiment includes annual contributions rather than a lump-sum investment, the optimal rebalancing control remains nearly wealth-independent, consistent with the scaling property discussed earlier in Lemma~\ref{lem:scaling} and Remark~\ref{rm:no_equiv}.

    Consequently, the $TCMa$ control resembles a glide-path strategy. It begins with full allocation to the risky asset to pursue growth, and gradually shifts to the risk-free asset over time to reduce risk exposure. This mirrors the principle of glide-path design: aggressive investment in early periods followed by systematic de-risking as the investment horizon shortens.

\item
The $TCMo$ control shown in {\sifin{Figure~\ref{Fig: Time-consistent optimal control heat maps}(b)}}
closely resembles the $PCMo$ control in {\sifin{Figure~\ref{Fig: Precommitment optimal control heat maps}(b)}}.
This similarity supports the idea that $TCMo$ closely mirrors $PCMo$ not only in outcomes but also in the structure of optimal rebalancing decisions. By contrast, the Mean–CVaR case shows a clear divergence between $PCMa$ and $TCMa$, with markedly different control behavior.
\end{itemize}
To better understand the structural differences between $TCMa$ and $TCMo$, we now examine the behavior of their optimal  thresholds $W^{c\ast}$, which serve as the key drivers of each strategy.}}

\subsubsection{\texorpdfstring{Optimal thresholds $\boldsymbol{W^{c \ast}}$}{Threshold}}
{\purple{Figure~\ref{Fig: Time-consistent optimal threshold heat maps} displays heat maps of the optimal thresholds $W^{c \ast}$ for $TCMa$ and $TCMo$. In the $TCMa$ case
({\sifin{Figure~\ref{Fig: Time-consistent optimal threshold heat maps}(a))}},
$W_{a}^{c \ast}$ varies substantially, sweeping the full range from \$0 to over \$3.5 million as wealth increases from \$0 to \$3 million and time evolves.

\begin{figure}[!hbt]
    \centering
    \subfigure[$TCMa$ optimal threshold]{
    \label{fig: TCMa_outer_Kou3_heatmap}
        \includegraphics[width = 0.48\textwidth]{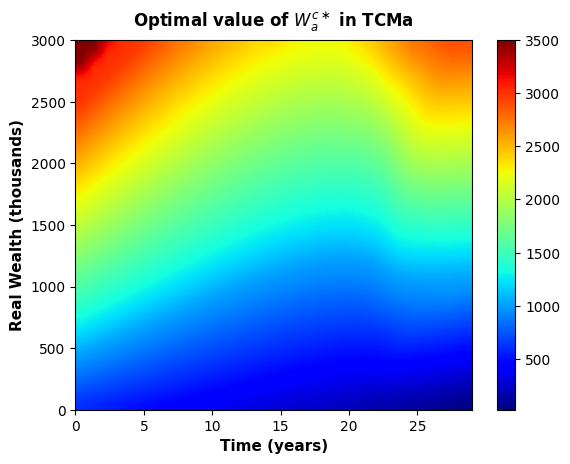}}
    \subfigure[$TCMo$ optimal threshold]{
    \label{fig: TCMo_outer_Kou3_heatmap}
        \includegraphics[width=0.48\textwidth]{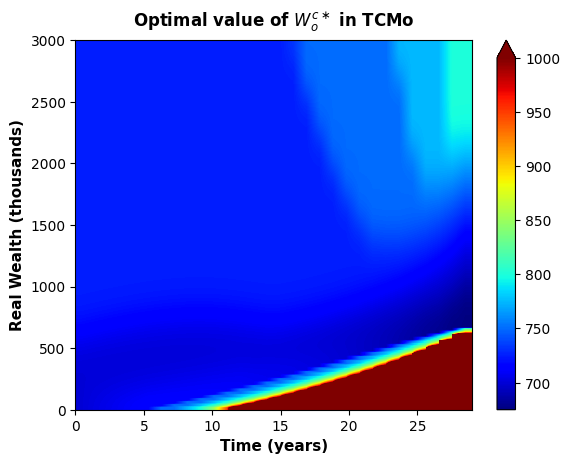}}
\caption{Time-consistent optimal threshold heat maps}
\label{Fig: Time-consistent optimal threshold heat maps}
\vspace*{-0.4cm}
\end{figure}
{\sifin{
By contrast, in the $TCMo$ case (Figure~\ref{Fig: Time-consistent optimal threshold heat maps}(b)), more than 80\% of the computed values of $W_{o}^{c \ast}$ lie in a narrow band around the disaster level, typically between $D \approx \$0.67$ million and about \$0.8 million. This confirms that, on most of the state space, the time–consistent bPoE problem behaves like a near--constant--threshold policy.

The dark red region in the lower–right corner of
Figure~\ref{Fig: Time-consistent optimal threshold heat maps}(b) corresponds to states where the current wealth is very low and the remaining horizon is short. In this region, even the most aggressive admissible strategy (investing fully in the risky asset) cannot generate a terminal wealth distribution whose lower tail is aligned with the disaster level $D$.
Numerically, the optimisation over
$W_o^c$ then favors very large thresholds and selects the aggressive control
$u_{o,h}^{c\ast} = 1$ (see the corresponding region in
Figure~\ref{Fig: Time-consistent optimal control heat maps}(b)), since in this regime the bPoE component is essentially saturated and thus
offers little scope for improvement, so there is no incentive to reduce risky
exposure; the resulting optimizer takes $u_{o,h}^{c\ast}=1$ to maximize
expected terminal wealth.


In the numerical scheme, we work with a truncated threshold domain
$\Gamma_o = [D, W_{o}^{\max}]$. In the red region described above, the
discrete optimisation therefore selects $W_{o,h}^{c\ast} \approx W_{o}^{\max}$
and $u_{o,h}^{c\ast} = 1$, so the truncation affects only the reported
value of the threshold (which appears as the saturated red area), not the
optimal control or the backward value update: once $u_{o}^{c\ast}=1$, the
intervention step \eqref{eq:intervention_c} depends only on this control and is
essentially insensitive to how large $W_{o}^{c\ast}$ is, provided $W_{o}^{\max}$ is chosen sufficiently large. On the remainder of the domain,
where a finite optimal threshold exists, $W_{o}^{c\ast}$ lies strictly inside
$\Gamma_o$ and our convergence theorem applies directly.
}}

Given that $PCMo$ and $TCMo$ produce remarkably similar rebalancing controls and investment outcomes, investors may interchange the precommitment and time-consistent solutions with minimal loss of accuracy. This behavioral similarity stands in sharp contrast to the Mean–CVaR case, where time consistency leads to counterintuitive, wealth-independent controls that diverge substantially from their precommitment counterparts. The key difference lies in the structural role of the disaster level $D$: by fixing $D$ from the outset, bPoE preserves a stable tail-risk objective across time, thereby avoiding the shifting risk preferences and unintuitive rebalancing behavior commonly caused by time-consistency constraints. In this sense, bPoE not only offers intuitive tail protection but also resolves the well-known dynamic inconsistency challenges that arise in multi-period mean–risk portfolio optimization.

}}

\section{Conclusion and future work}
\label{sc:conclude}
{\purple{This paper investigates the use of bPoE as a viable and intuitive alternative to CVaR in multi-period portfolio optimization for Defined Contribution plans. We formulate both pre-commitment and time-consistent Mean–bPoE and Mean–CVaR problems under realistic investment constraints and jump-diffusion dynamics, and develop a provably convergent numerical framework capable of solving all formulations.

In the pre-commitment setting, we establish a one-to-one correspondence between the scalarization optimal sets of Mean–bPoE and Mean–CVaR. {\sifin{For each pair of corresponding frontier points, there exists a threshold/rebalancing control pair
that is optimal for both formulations and attains these points, allowing bPoE to
be seamlessly integrated into existing CVaR--based workflows.}}

In the time-consistent setting, however, this equivalence no longer holds. Mean–CVaR strategies often exhibit counterintuitive, wealth-independent controls due to dynamic threshold re-optimization. By contrast, Mean-bPoE maintains a fixed disaster level, resulting in stable shortfall thresholds and wealth-dependent rebalancing behavior that better aligns with investor preferences for a minimum acceptable terminal wealth. As a result, time-consistent Mean–bPoE closely resembles its pre-commitment counterpart---both in control structure and investment outcomes---and consistently delivers superior tail performance across a range of key metrics.

These findings show that bPoE retains the computational tractability of CVaR while offering practical advantages for long-horizon retirement planning. Its stable threshold and control structure, intuitive interpretation, and ability to resolve key limitations of time-consistent Mean–CVaR formulation make it a compelling risk measure for DC lifecycle portfolio design.

Future work may extend the bPoE framework to incorporate interim withdrawals, inflation-linked liabilities, or stochastic mortality, to assess its robustness in more realistic lifecycle settings.

}}

\section*{Appendices}
\appendix
\section{Proof of Proposition~\ref{prop:F_continuous}}
\label{app:F_continuous}
Consider a fixed control $\mathcal{U}_0 \in \mathcal{A}$. Define
the function
\EQ
\label{eq:fW}
f(W_{o}^p;\mathcal{U}_0) = \Ebb_{\mathcal{U}_0}^{X_{0}^+,t_{0}^+}\left[
h(W_{o}^p;W_T)\bigg| X_0^- = (s_0, b_0) \right],
\EN
where $h(W_{o}^p;w) = \gamma \max\left( 1 - \frac{w-D}{W_{o}^p-D} , 0 \right) - w$.
Clearly, $h(W_{o}^p;w) $ is continuous in $W_{o}^p$ for each fixed~$w$.
Note that $F(W_o^p) = \inf_{\mathcal{U}_0\in \mathcal{A}} f(W_o^p;\mathcal{U}_0)$.

We now show that $f(W_{o}^p;\mathcal{U}_0)$ is continuous in $W_{o}^p$
for this fixed control $\mathcal{U}_0$. Specifically, for any fixed $\widehat{W}_o \in (D, \infty)$, and any sequence $\{W_{o}^{p\, (n)}\}$  with $W_{o}^{p\,(n)} > D$ for all $n$ and $W_{o}^{p\, (n)} \to \widehat{W}_o$, we will prove
\EQ
\label{eq:f_lim}
\lim_{n\to\infty}
f\bigl(W_{o}^{p\, (n)};\mathcal{U}_0\bigr)
~=\;
f\bigl(\widehat{W}_o;\mathcal{U}_0\bigr).
\EN
Equivalently, we will show
\EQ
\label{eq:E_lim}
\lim_{n\to\infty}
\Ebb_{\mathcal{U}_0}^{X_{0}^+,t_{0}^+}\Bigl[
  h\bigl(W_{o}^{p\,(n)}; W_T(\omega)\bigr) ~\big|~\cdot~
\Bigr]
= \Ebb_{\mathcal{U}_0}^{X_{0}^+,t_{0}^+}\Bigl[
  h\bigl(\widehat{W}_o; W_T(\omega)\bigr)~ \big|~\cdot~
\Bigr].
\EN
Once \eqref{eq:E_lim} is established, by the compactness of $\mathcal{A}$ and the continuity argument in parametric optimization \cite{bonnans2000perturbation},
we conclude that $F(W_{o}^p) =  \inf_{\mathcal{U}_0\in \mathcal{A}} f(W_{o}^p;\mathcal{U}_0)$ is also continuous for all $W_{o}^p > D$.

It remains to prove \eqref{eq:E_lim}, which we will do using a dominated convergence argument.
In the remainder of the proofs, let $C$ and $C'$ denote bounded constants that potentially depend only on the fixed parameters $\widehat{W}_o$, $\gamma$, and $D$, and are independent of $n$ and $w$. Their values may vary from line to line.

Let $\delta = \frac{1}{2}\,\bigl(\widehat{W}_o-D\bigr) >0$.
For sufficiently large $n$, say $n \ge N_*$, $W_{o}^{p\,(n)} \in
[\widehat{W}_o-\delta,\;\widehat{W}_o+\delta\,] ~\subset~ (D,\infty)$.
Thus, for $n\ge N_*$, we have
\EQ
\label{eq:W_oD}
W_{o}^{p\,(n)} - D
~\ge~
\widehat{W}_o- D - \delta
~=~
\tfrac12\,(\widehat{W}_o-D).
\EN
In addition,
\EQ
\label{eq:W_oD_less}
W_{o}^{p\,(n)}
~\le~
\widehat{W}_o+\delta
~=~
\widehat{W}_o + \tfrac12\,(\widehat{W}_o-D)
~\le~
2\,(\widehat{W}_o + |D|).
\EN
Now, we bound $h(W_{o}^{p\,(n)};w)$ for sufficiently large $n$, i.e.\ $n\ge N_*$, so that $W_{o}^{p\,(n)} \in
[\widehat{W}_o-\delta,\;\widehat{W}_o+\delta\,]$.
\begin{itemize}
\item If $w\ge W_{o}^{p\,(n)}$: we have $h(W_{o}^{p\,(n)};w) = \gamma \cdot 0 -w  = -w$,
so $|h(W_{o}^{p\,(n)};w)| \le |w|$.
\item If $w< W_{o}^{p\,(n)}$: $\max\left(1-\tfrac{w-D}{\,W_{o}^{p\,(n)}-D},0\right)=1-\tfrac{w-D}{\,W_{o}^{p\,(n)}-D}$, so
    $h\bigl(W_{o}^{p\,(n)};w\bigr)
    ~=~
    \gamma\,\Bigl(\tfrac{W_{o}^{p\,(n)}-w}{\,W_{o}^{p\,(n)}-D}\Bigr)
    \;-\;w$.

    By \eqref{eq:W_oD}, for sufficiently large $n$, the denominator
    $(W_{o}^{p\,(n)} - D) \ge \tfrac12\,(\widehat{W}_o-D)$.
    The numerator $(W_{o}^{p\,(n)}-w) \le (W_{o}^{p\,(n)} +|w|)$.
    Thus,
    \EQS
    \begin{aligned}
    \Bigl|\gamma\,\tfrac{W_{o}^{p\,(n)}-w}{W_{o}^{p\,(n)}-D} -w\Bigr|
    &\le
    \gamma\,\tfrac{W_{o}^{p\,(n)} +|w|}{\tfrac12(\widehat{W}_o-D)}
    \;+\;
    |w|
    \overset{\text{\eqref{eq:W_oD_less}}}{\le}
    \tfrac{4\gamma\,\bigl(\widehat{W}_o + |D| +|w|\bigr)}{\widehat{W}_o-D}  + |w|
    \overset{\text{(ii)}}{=} C + C'\,|w| + |w|.
  \end{aligned}
  \ENS
\end{itemize}
Hence, in either of the two cases, we can conclude that $\forall n \ge N_*$ and $w\ge 0$
\[
\bigl|h\bigl(W_{o}^{p\,(n)}; w\bigr)\bigr|
~\le~
C + (C'+1)\,\lvert w\rvert.
\]
For $\omega \in \Omega$, we define
$B(\omega) = C + (C'+1)\,W_T(\omega)$.
Hence, for $n \ge N_*$,
$\bigl|h(W_{o}^{p\,(n)};W_T(\omega))\bigr| \;\le\; B(\omega)$.
By Remark~\ref{rm:box},
$\Ebb_{\mathcal{U}_0}^{X_{0}^+,t_{0}^+}
    \left[W_{T} ~\big|~ \cdot \right]< \mathcal{E}_{\max}$, so
the random variable $B(\omega)$ is integrable.

By the Dominated Convergence Theorem and the pointwise continuity in $W_o$, we conclude that
\[
\lim_{n\to\infty}
\Ebb_{\mathcal{U}_0}^{X_{0}^+,t_{0}^+}\Bigl[
  h\bigl(W_{o}^{p\,(n)}; W_T(\omega)\bigr) ~\big|~\cdot~
\Bigr]
=
\Ebb_{\mathcal{U}_0}^{X_{0}^+,t_{0}^+}\Bigl[
  \lim_{n\to\infty}
   h\bigl(W_{o}^{p\,(n)}; W_T(\omega)\bigr) ~\big|~\cdot~
\Bigr]
~=\;
\Ebb_{\mathcal{U}_0}^{X_{0}^+,t_{0}^+}\Bigl[
  h\bigl(\widehat{W}_o; W_T(\omega)\bigr)~ \big|~\cdot~
\Bigr].
\]
This is \eqref{eq:E_lim}. 
In conclusion, for each fixed $\mathcal{U}_0 \in \mathcal{A}$, the function $f(W_o^p;\mathcal{U}_0)$ is continuous in $W_o^p$.

\section{Proof of Lemma~\ref{lem:infF}}
\label{app:infF}
For brevity, we will write $F(W_o^p)$, with the dependence on
the initial state $(s_0, b_0)$ understood implicitly.
Let $f(W_o^p;\mathcal{U}_0)$ be as in \eqref{eq:fW},
so $F(W_o^p) = \inf_{\mathcal{U}_0\in \mathcal{A}} f(W_o^p;\mathcal{U}_0)$.

First, we examine the behavior of $F(W_o^p)$ as $W_o^p \to D^+$ and $W_o^p \to \infty$.
As $W_o^p \to D^+$, the term $\frac{1}{\,W_o^p-D\,}$ blows up, causing
$F(W_o^p) \to +\infty$.

For $W_{o}^{p} \rightarrow +\infty$, the term $\max \biggl(1 - \tfrac{W_T - D}{\,W_o^p - D\,},\;0\biggr) \rightarrow 1$, and the function $F\left(W_{o}^{p}\right)$ becomes
\begin{eqnarray}
    \lim_{W_{o}^{p} \rightarrow +\infty} F\left(W_{o}^{p}\right) & = &
      \inf_{\mathcal{U}_{0} \in \mathcal{A}}
      \mathbb{E}_{\mathcal{U}_{0}}^{\,X_{0}^+,\,t_{0}^+}
      \bigl[
        \gamma - W_T
       \bigg| X_0^- = (s_0, b_0) \bigr].\label{eq: W_o^p goes infty}
\end{eqnarray}
The only target for optimization problem~\eqref{eq: W_o^p goes infty} is to maximize the expected return, so the optimal control is the strategy that fully invests in the risky asset,
as noted in Remark~\ref{rm:box}. Denote by $\widehat{\mathcal{U}_{0}}$ the optimal control in this case. Thus, recalling $\mathcal{E}_{\max}$ from \eqref{eq:E_max}, we have
\begin{eqnarray*}
    \lim_{W_{o}^{p} \rightarrow +\infty} F\left(W_{o}^{p}\right) & = & \gamma - \mathcal{E}_{\max}.
\end{eqnarray*}
However, as noted in Remark~\ref{rm:bpoE_Wstar}, there exists a threshold value $W_{o}^{\prime} =\text{VaR}_{\mathcal{B}^{\prime}}(W_{T};\widehat{\mathcal{U}_{0}})\in (D, +\infty)$, where $\mathcal{B^{\prime}} = \mathbb{E}_{\widehat{\mathcal{U}_{0}}}^{X_{0}^{+}, t_{0}^{+}}\left[ \max\left( 1-\frac{W_{T}-D}{W_{o}^{\prime}-D} ,0 \right) \right] \in (0,1)$, such that
\begin{eqnarray}
    F\left( W_{o}^{\prime} \right) \le f\left( W_{o}^{\prime}; \widehat{\mathcal{U}_{0}} \right) = \gamma \mathcal{B} - \mathcal{E}_{\max} < \gamma - \mathcal{E}_{\max}.
\end{eqnarray}
Therefore, we show that $W_{o}^{p} \rightarrow +\infty$ can not be the minimizer.


By Proposition~\ref{prop:F_continuous}, $F(W_o^p)$ is continuous on $(D,\infty)$.
Combined with the boundary behavior, where $F(W_o^p) \to +\infty$ as
$W_o^p \to D^+$ and no further decrease is possible for $F(W_o^p)$
as $W_o^p \to \infty$, we conclude that $\inf_{W_o^p>D}F(W_o^p)$ is finite
and is attained at some finite point
$W_o^{p*}(s_0, b_0)\in(D,\infty)$.

\section{Proof of Proposition~\ref{prop:equiv}}
\label{app:equiv}
From \eqref{eq:W_op_W0} and \eqref{eq: PCMb_new}, the pre-commitment value function becomes
\EQ
\label{eq:precommit_obj}
V^p_o\left(s_0, b_0,t_{0}^-\right) =
\inf_{\mathcal{U}_0 \in \mathcal{A}} \Ebb_{\mathcal{U}_0}^{X_{0}^+,t_{0}^+}\left[ \gamma \max\left( 1 - \frac{W_{T}-D}{W_o^{p*}(s_0, b_0)-D} , 0 \right) - W_{T}  \bigg| X_0^- = (s_0, b_0) \right].
\EN
Meanwhile, for $TCEQ_{t_0}(D,\gamma;W_o^{p*}(s_0, b_0))$ at $t_0$, the value function reads
\EQ
\label{eq:tceq_obj}
\widehat{Q}^c_o(s_0, b_0,t_{m}) \coloneqq \inf_{\mathcal{U}_{0}}
    \bigg\{\Ebb_{\mathcal{U}_{m}}^{X_0^+,t_{0}^+}\bigg[ \max(\widehat{W_o}-W_T,\,0)
  - (\widehat{W_o}-D)\,\tfrac{W_T}{\gamma}
  \bigg| X_0^- = (s_0, b_0)\bigg] \bigg\}.
\EN
Since $\gamma>0$, and $\widehat{W_o}-D>0$,  any optimal control $\mathcal{U}_{0,o}^{p*}$ for \eqref{eq:precommit_obj}
also minimizes \eqref{eq:tceq_obj}, and vice versa.
Moreover, \eqref{eq:tceq_obj} is solvable via standard dynamic programming
over the controls. Hence, $\mathcal{U}_{0,o}^{p*}$ itself is time-consistent.
This completes the proof.

\section{Proof of Lemma~\ref{lem:CVaR_threshold_existence}}
\label{app:CVaR_threshold_existence}
For each fixed admissible  control $\mathcal{U}_0 \in \mathcal{A}$, we define
\[
 f_a\bigl(W_a^p;\mathcal{U}_0\bigr)
       \;=\;
       \Ebb_{\mathcal{U}_0}^{\,X_0^+,\,t_0^+}\!\Bigl[\,
         \gamma\Bigl(W_a^p + \tfrac{1}{\alpha}\min(W_T - W_a^p,\,0)\Bigr)
         + W_T \big|~ \cdot~
      \Bigr].
\]
A standard dominated convergence argument (analogous to Proposition~\ref{prop:F_continuous}
for the pre-commitment Mean–bPoE case) shows that $f_a(W_a^p;\mathcal{U}_0)$ is continuous in $W_a^p$. By compactness of $\mathcal{A}$ and parametric optimization results
\cite{bonnans2000perturbation}, the function $F_a(W_a^p) = \sup_{\mathcal{U}_0\in \mathcal{A}} f_a(W_a^p;\,\mathcal{U}_0)$ is continuous for all $W_a^p$ in $[0,\infty)$.

We next verify that $F_a(W_a^p)$ cannot blow up as $W_a^p\to 0^+$ or $W_a^p\to \infty$. As $W_a^p\to 0^+$, $f_a$ behaves like $\Ebb_{\mathcal{U}_0} [\gamma \,\min(W_T,\,0)/\alpha + W_T| \cdot]$. Since $W_T$ is integrable (Remark~\ref{rm:box_cvar}),  this expectation is finite,
so no unbounded positive growth occurs.

For the case $W_a^p\to \infty$, for each fixed $W_a^p$, we partition the sample space $\Omega$ into
\[
\Omega_1(W_a^p) \;=\;\{\omega \in \Omega | \,W_T(\omega) \le W_a^p\}
\quad\text{and}\quad
\Omega_2(W_a^p) \;=\;\{ \omega \in \Omega | \,W_T (\omega) > W_a^p\}.
\]
On $\Omega_1(W_a^p)$, we have $W_T \le W_a^p$, so
$W_a^p \;+\;
\frac{1}{\alpha}\,\min(W_T -W_a^p,\,0)
~=~
W_a^p\,\Bigl(1-\tfrac1\alpha\Bigr)
\;+\;
\tfrac{W_T}{\alpha}$. Since $1-\tfrac1\alpha< 0$, pushing $W_a^p$ larger
actually lowers this part.

\noindent On $\Omega_2(W_a^p)$, $W_T>W_a^p$, so  $\min(W_T-W_a^p,0)=0$. Therefore,
\EQS
\begin{aligned}
f_a\bigl(W_a^p;\mathcal{U}_0\bigr) &= \Ebb_{\mathcal{U}_0}^{\,X_0^+,\,t_0^+} \left[\, \gamma\,W_a^p \;+\; W_T\right |~ \cdot~]
=
\int_{\{W_T(\omega)>W_a^p\}} \bigl[\gamma\,W_a^p+W_T(\omega)\bigr] \,d\Pbb(\omega)
\\
&= \gamma\,W_a^p
 \,\Pbb\!\bigl(W_T>W_a^p\bigr)
+
\int_{\{W_T(\omega)>W_a^p\}}\!\!W_T(\omega)\,d\Pbb(\omega).
\\
&\overset{(i)}{\le}
 \gamma\,W_a^p
\frac{\Ebb_{\mathcal{U}_0}^{\,X_0^+,\,t_0^+}[\,W_T\,]}{\,W_a^p\,} +
\Ebb_{\mathcal{U}_0}^{\,X_0^+,\,t_0^+}\big[\,W_T\mathbb{I}_{\{W_T>W_a^p\}}\big]
\\
&\overset{(ii)}{=} \gamma  \Ebb_{\mathcal{U}_0}^{\,X_0^+,\,t_0^+}[\,W_T\,]
+
\int_{W_a^p}^\infty \Pbb(W_T > u) du + W_a^p \Pbb(W_T > W_a^p).
\end{aligned}
\ENS
Here, (i) follows from Markov's inequality,  while (ii) is due to
Fubini's theorem. Since $W_T$ has finite expectation,
standard tail‐integration arguments imply that, as $W_a^p\to \infty$,
both  $\int_{W_a^p}^\infty \Pbb(W_T > u) du$ and
and $W_a^p \Pbb(W_T > W_a^p)$ vanish.
Therefore, in all cases, no blow-up can occur as $W_a^p\to \infty$.

In summary, $F_a(W_a^p)$ remains finite and is continuous on $[0,\infty)$, and it does not blow up at either $0$ or $+\infty$. Therefore, $\sup_{W_a^p \ge 0} F_a(W_a^p)$ must be attained at some finite $W_a^{p*} \in [0,\infty)$. This completes the proof.

\section{Proof of Lemma~\ref{lem: Equiv_CVaR_to_bPoE}}
\label{app:Equiv_CVaR_to_bPoE}
By Lemma~\ref{lem:non_emptiness_cvar}, the scalarization optimal set $\mathcal{S}_{a}(\alpha, \gamma_a)$ for the Mean--CVaR problem is nonempty. As such, the existence of a point $\bigl(\mathcal{C}_a^*,\,\mathcal{E}_a^*\bigr) \in
\mathcal{S}_{a}(\alpha, \gamma_a)$ is guaranteed. Furthermore, the associated optimal threshold $W_{a}^{p*}$ exists by Lemma~\ref{lem:CVaR_threshold_existence}.
Therefore, $\gamma_{o}$ given by \eqref{eq:gamma_o} is well-defined.

Let $\mathcal{U}_{0, a}^{p\ast}$ denote the optimal control associated with $\bigl(\mathcal{C}_a^*,\,\mathcal{E}_a^*\bigr)$.
Then using definitions in \eqref{eq:ECVaR}, the point $\bigl(\mathcal{C}_a^*,\,\mathcal{E}_a^*\bigr)$ can be expressed in
terms of $W_{a}^{p*}$ and $\mathcal{U}_{0, a}^{p\ast}$ as follows:
\EQ
\label{eq:CE}
\mathcal{C}_a^* = W_{a}^{p\ast}+ \tfrac{1 }{\alpha }
\Ebb_{\mathcal{U}_{0, a}^{p\ast}}\left[ \min\left( W_{T}-W_{a}^{p\ast} , 0 \right)\right],
\quad \text{and}\quad
\mathcal{E}_a^* = \Ebb_{\mathcal{U}_{0, a}^{p\ast}}\left[W_{T}\right].
\EN
In order to prove that $\bigl(\mathcal{B}_o^*,\,\mathcal{E}_o^*\bigr)=  \bigl(\alpha,\,\mathcal{E}_a^*\bigr)$
is in $\mathcal{S}_{o}\!\bigl(\mathcal{C}_a^*, \gamma_{o}\bigr)$,
by definition of the Mean--bPoE scalarization optimal set given in  \eqref{eq:S_o_gamma},
we need to show
\EQ
\label{eq:scar_*}
\gamma_{o} \alpha  - \mathcal{E}_a^*
=
\underset{\substack{W_o^p > \mathcal{C}_a^*,\\ \mathcal{U}_0}}{\inf}
    \left\{
    \tfrac{\gamma_{a}(W_{a}^{p*} - \mathcal{C}_a^*)}{\alpha} \Ebb_{\mathcal{U}_0}\left[
  \max\left(1 - \tfrac{W_T - \mathcal{C}_a^*}{W_o^p - \mathcal{C}_a^*},\, 0\right)
\right]
- \Ebb_{\mathcal{U}_0}[W_T]
\right\}.
\EN
We first carry out further manipulations on both sides of \eqref{eq:scar_*}.
Using \eqref{eq:gamma_o} and the formula for $\mathcal{C}_a^*$ in \eqref{eq:CE} on the lhs of \eqref{eq:scar_*}, we obtain
\EQA
\text{lhs of \eqref{eq:scar_*}} &=&   \gamma_{a}(\,W_{a}^{p*} \;-\; \mathcal{C}_a^*\,) -  \mathcal{E}_a^*
\nonumber
\\
&=& -\tfrac{\gamma_{a} }{\alpha }
\Ebb_{\mathcal{U}_{0, a}^{p\ast}}\left[ \min\left( W_{T}-W_{a}^{p\ast} , 0 \right) \right] - \mathcal{E}_a^*.
\label{eq:gamma_o_star_alpha}
\ENA
Simplifying the $\max(\cdot, 0)$ term on the rhs of  \eqref{eq:scar_*}, noting
$W_o^p - \mathcal{C}_a^*>0$,  gives
\begin{align}
\text{rhs of \eqref{eq:scar_*}}
&= \underset{\substack{W_o^p > \mathcal{C}_a^*,\\ \mathcal{U}_0}}{\inf}
    \left\{
     -\tfrac{\gamma_{a} (W_{a}^{p*} - \mathcal{C}_a^*)}{\alpha \left( W_{o}^p-\mathcal{C}_a^* \right)} \Ebb_{\mathcal{U}_0}\left[ \min\left(W_{T}-W_{o}^p , 0 \right) \right]
     - \Ebb_{\mathcal{U}_0}\left[W_{T}\right]\right\}.
\label{eq:rhs}
\\
&\overset{(i)}{=} \!-\! \underset{\substack{W_o^p > \mathcal{C}_a^*,\\ \mathcal{U}_0}}{\sup}
\!\!\left\{
     \tfrac{\gamma_{a} (W_{a}^{p*} - \mathcal{C}_a^*)}{\alpha \left( W_{o}^p-\mathcal{C}_a^* \right)}\, \Ebb_{\mathcal{U}_0}\left[ \min\left(W_{T}-W_{o}^p , 0 \right) \right]
     + \Ebb_{\mathcal{U}_0}\left[W_{T}\right]\right\}.
      \label{eq: opt Mean_bPOE v3*}
\end{align}
Here, (i) is due to the identity
$-\underset{z \in \mathcal{X}}{\sup}\{f(z)\} = \underset{z \in \mathcal{X}}{\inf}\{-f(z)\}$
for any function $f$ and any set $\mathcal{X}$.

By noting that, with $W_o^p = W_{a}^{p\ast}$ and $\mathcal{U}_0 = \mathcal{U}_{0, a}^{p\ast}$
in \eqref{eq:rhs}, we obtain \eqref{eq:gamma_o_star_alpha} as a candidate
for the infimum in \eqref{eq:rhs}. This leads to the following inequality
\EQS
\begin{aligned}
\underset{\substack{W_o^p > \mathcal{C}_a^*,\\ \mathcal{U}_0}}{\inf}
    \left\{
     \tfrac{-\gamma_{a} (W_{a}^{p*} - \mathcal{C}_a^*)}{\alpha \left( W_{o}^p-\mathcal{C}_a^* \right)} \Ebb_{\mathcal{U}_0}\left[ \min\left(W_{T}-W_{o}^p , 0 \right) \right] - \Ebb_{\mathcal{U}_0}\left[W_{T}\right]\right\}
\le
-\tfrac{\gamma_{a} }{\alpha }
\Ebb_{\mathcal{U}_{0, a}^{p\ast}}\left[ \min\left( W_{T}-W_{a}^{p\ast} , 0 \right)\right]\, -\,  \mathcal{E}_a^*.
\end{aligned}
\ENS
This shows that the rhs of \eqref{eq:scar_*} is less than or equal to its lhs.

To show the reverse inequality, note that the right-hand side of \eqref{eq:scar_*} equals the right-hand side of \eqref{eq: opt Mean_bPOE v3*}.  Therefore, noting $\gamma_a W_{a}^{p*}$ is
a constant, it suffices to show
\[
\text{rhs of \eqref{eq: opt Mean_bPOE v3*}}\;\ge\; \text{lhs of \eqref{eq:scar_*}}, ~~
\text{or equivalently}, ~~
-\,\text{(rhs of \eqref{eq: opt Mean_bPOE v3*})} + \gamma_{a} W_{a}^{p*}   \;\le\; -\,\text{(lhs of \eqref{eq:scar_*})} + \gamma_{a} W_{a}^{p*}.
\]
That is, with the lhs of \eqref{eq:scar_*} given in \eqref{eq:gamma_o_star_alpha},
we will show
\EQ
\begin{aligned}
\label{eq:equ}
&\underset{\substack{W_o^p > \mathcal{C}_a^*,\\ \mathcal{U}_0}}{\sup}
\left\{\gamma_a W_{a}^{p*} +
     \tfrac{\gamma_{a} (W_{a}^{p*} - \mathcal{C}_a^*)}{\alpha \left( W_{o}^p-\mathcal{C}_a^* \right)}\, \Ebb_{\mathcal{U}_0}\left[ \min\left(W_{T}-W_{o}^p , 0 \right) \right]
     + \Ebb_{\mathcal{U}_0}\left[W_{T}\right]\right\}
\\
&
\qquad \qquad\qquad \qquad \le
 \gamma_a W_{a}^{p*} +\tfrac{\gamma_{a} }{\alpha }
\Ebb_{\mathcal{U}_{0, a}^{p\ast}}\left[ \min\left( W_{T}-W_{a}^{p\ast} , 0 \right) \right] + \mathcal{E}_a^*.
\end{aligned}
\EN
Starting from the lhs of \eqref{eq:equ}, we have
\EQS
\begin{aligned}
 \text{lhs of \eqref{eq:equ}} &\overset{(i)}{\le}
 \underset{\substack{W_o^p,\, \mathcal{U}_0}}{\sup}
\left\{\gamma_a W_{o}^{p} +
     \tfrac{\gamma_{a} (W_{o}^{p} - \mathcal{C}_a^*)}{\alpha \left( W_{o}^p-\mathcal{C}_a^* \right)}\, \Ebb_{\mathcal{U}_0}\left[ \min\left(W_{T}-W_{o}^p , 0 \right) \right]
     + \Ebb_{\mathcal{U}_0}\left[W_{T}\right]\right\}
\\
& \overset{(ii)}{=}
 \underset{\substack{W_o^p,\, \mathcal{U}_0}}{\sup}
\left\{\gamma_a W_{o}^{p} +
     \tfrac{\gamma_{a}}{\alpha}\, \Ebb_{\mathcal{U}_0}\left[ \min\left(W_{T}-W_{o}^p , 0 \right) \right]
     + \Ebb_{\mathcal{U}_0}\left[W_{T}\right]\right\}
\\
&\overset{(iii)}{=}
 \gamma_a W_{a}^{p*} +\tfrac{\gamma_{a} }{\alpha }
\Ebb_{\mathcal{U}_{0, a}^{p\ast}}\left[ \min\left( W_{T}-W_{a}^{p\ast} , 0 \right)\right]\, +\,  \mathcal{E}_a^*,
\end{aligned}
\ENS
which is the inequality \eqref{eq:equ} we need to prove.
Here, in (i), the inequality arises because $W_a^{p\ast}$ is replaced by a free variable $W_o^p$, and  the supremum is taken over a larger set for $W_{o}^{p}$;
(ii) follows from simplification; and (iii) is due to the fact that
$W_{o}^{p}$ is a dummy variable for the optimization, noting
$\mathcal{U}_{0, a}^{p\ast}$ and $W_{a}^{p*}$
are the optimal control/threshold pair associated with
$\bigl(\mathcal{C}_a^*,\mathcal{E}_a^*\bigr)$.
This concludes the proof.

\section{Proof of Lemma~\ref{lem: Equiv_bPoE_to_CVaR}}
\label{app:Equiv_bPoE_to_CVaR}
By Lemma~\ref{lem:non_emptiness},  the existence of a point $\bigl(\mathcal{B}_o^*,\,\mathcal{E}_o^*\bigr) \in \mathcal{S}_{o}\!\bigl(D, \gamma_{o}\bigr)$  is guaranteed. Furthermore, the associated optimal threshold $W_{o}^{p*}>D$ exists by Lemma~\ref{lem:infF}. Therefore, $\gamma_{a}$ given by \eqref{eq:gamma_a} is well-defined.

Let $\mathcal{U}_{0, o}^{p\ast}$ denote the optimal control associated with $\bigl(\mathcal{B}_o^*,\,\mathcal{E}_o^*\bigr)$.
Then using definitions in \eqref{eq:EbPoE}, the point $\bigl(\mathcal{B}_o^*,\,\mathcal{E}_o^*\bigr)$ can be expressed in
terms of $W_{o}^{p*}$ and $\mathcal{U}_{0, o}^{p\ast}$ as follows:
\EQ
\label{eq:CE*}
\mathcal{B}_o^* =
\Ebb_{\mathcal{U}_{0, o}^{p\ast}}\left[ \max\left(1 - \tfrac{W_T - D}{W_o^{p\ast} - D},\, 0\right)\right] = -\tfrac{\Ebb_{\mathcal{U}_{0, o}^{p\ast}}\left[ \min\left(W_T - W_o^{p\ast}, \, 0\right)\right]}{W_o^{p\ast} - D},
\quad \text{and}\quad
\mathcal{E}_o^* = \Ebb_{\mathcal{U}_{0, o}^{p\ast}}\left[W_{T}\right].
\EN

From \eqref{eq:CE*}, we can express the disaster level $D$ in terms of
$W_o^{p\ast}$, $\mathcal{B}_o^*$, and the expectation as
\EQ
\label{eq:DE}
D = W_o^{p\ast} + \tfrac{1}{\mathcal{B}_o^*}\, \Ebb_{\mathcal{U}_{0, o}^{p\ast}}\left[ \min\left(W_T - W_o^{p\ast}, \, 0\right)\right].
\EN
In order to prove that
$\bigl(\mathcal{C}_a^*,\mathcal{E}_a^*\bigr) =  \bigl(D,\,\mathcal{E}_o^*\bigr)$
belongs to $\mathcal{S}_{a}\!\bigl(\mathcal{B}_o^*, \gamma_{a}\bigr)$,
by definition \eqref{eq:S_a_gamma}, we need to verify
\begin{eqnarray}
    \gamma_{a} D + \mathcal{E}_{o}^{\ast} & = & \underset{W_{a}^{p}, \ \mathcal{U}_{0}}{\sup} \left\{ \gamma_{a}W_{a}^{p} + \frac{\gamma_{a}}{\mathcal{B}_{o}^{\ast}} \mathbb{E}_{\mathcal{U}_{0}}\left[ \min{\left(
W_{T}-W_{a}^{p}, 0 \right)} \right] + \mathbb{E}_{\mathcal{U}_{0}} \left[ W_{T} \right] \right\}.\label{eq: 5.28}
\end{eqnarray}
The lhs of~\eqref{eq: 5.28} can be written as
\begin{eqnarray}
    \gamma_{a} D + \mathcal{E}_{o}^{\ast} & = & \gamma_{a}W_{o}^{p \ast} + \frac{\gamma_{a}}{\mathcal{B}_{o}^{\ast}} \mathbb{E}_{\mathcal{U}_{0,o}^{p \ast}}\left[ \min{\left(
W_{T}-W_{o}^{p \ast}, 0 \right)} \right] + \mathcal{E}_{o}^{\ast}.
\end{eqnarray}
By noting that with $W_{a}^{p} = W_{o}^{p \ast}$ and $\mathcal{U}_{0} = \mathcal{U}_{0,o}^{p \ast}$, we obtain a candidate for the supremum of rhs of~\eqref{eq: 5.28}. This shows that
\begin{eqnarray}
    \gamma_{a} D + \mathcal{E}_{o}^{\ast} & \le & \underset{W_{a}^{p}, \ \mathcal{U}_{0}}{\sup} \left\{ \gamma_{a}W_{a}^{p} + \frac{\gamma_{a}}{\mathcal{B}_{o}^{\ast}} \mathbb{E}_{\mathcal{U}_{0}}\left[ \min{\left(
W_{T}-W_{a}^{p}, 0 \right)} \right] + \mathbb{E}_{\mathcal{U}_{0}} \left[ W_{T} \right] \right\}.\label{eq: 5.28 le}
\end{eqnarray}
To establish the reverse inequality, we note that the rhs of~\eqref{eq: 5.28} can be written as
\begin{eqnarray}
    & & \underset{W_{a}^{p}, \ \mathcal{U}_{0}}{\sup} \left\{ \gamma_{a}W_{a}^{p} + \frac{\gamma_{a}}{\mathcal{B}_{o}^{\ast}} \mathbb{E}_{\mathcal{U}_{0}}\left[ \min{\left(
W_{T}-W_{a}^{p}, 0 \right)} \right] + \mathbb{E}_{\mathcal{U}_{0}} \left[ W_{T} \right] \right\}
\nonumber
\\
& = & \underset{W_{a}^{p}, \ \mathcal{U}_{0}}{\sup} \left\{  \frac{\gamma_{o} \mathcal{B}_{o}^{\ast} W_{a}^{p} }{W_{o}^{p \ast} - D}  + \frac{\gamma_{o}  }{W_{o}^{p \ast} - D} \mathbb{E}_{\mathcal{U}_{0}}\left[ \min{\left(
W_{T}-W_{a}^{p}, 0 \right)} \right] + \mathbb{E}_{\mathcal{U}_{0}} \left[ W_{T} \right] \right\} \quad \left( \gamma_{a} = \frac{\gamma_{o} \mathcal{B}_{o}^{\ast} }{W_{o}^{p \ast} - D} \right)
\nonumber
\\
& = & \underset{W_{a}^{p}, \ \mathcal{U}_{0}}{\sup} \left\{  \frac{\gamma_{o} \mathcal{B}_{o}^{\ast} W_{a}^{p} }{W_{o}^{p \ast} - D}  - \gamma_{o} \mathbb{E}_{\mathcal{U}_{0}}\left[ \max{\left(
\frac{W_{a}^{p} - D  }{W_{o}^{p \ast} - D} - \frac{W_{T} - D  }{W_{o}^{p \ast} - D}, 0 \right)} \right] + \mathbb{E}_{\mathcal{U}_{0}} \left[ W_{T} \right] \right\}
\nonumber
\\
& = & - \underset{W_{a}^{p}, \ \mathcal{U}_{0}}{\inf} \left\{  -\frac{\gamma_{o} \mathcal{B}_{o}^{\ast} W_{a}^{p} }{W_{o}^{p \ast} - D}  + \gamma_{o} \mathbb{E}_{\mathcal{U}_{0}}\left[ \max{\left(
\frac{W_{a}^{p} - D  }{W_{o}^{p \ast} - D} - \frac{W_{T} - D  }{W_{o}^{p \ast} - D}, 0 \right)} \right] - \mathbb{E}_{\mathcal{U}_{0}} \left[ W_{T} \right] \right\}. \label{eq: 5.28 rhs v2}
\end{eqnarray}
And lhs of~\eqref{eq: 5.28} can be written as
\begin{eqnarray}
    \gamma_{a} D + \mathcal{E}_{o}^{\ast} & = & \frac{\gamma_{o} \mathcal{B}_{o}^{\ast} D }{W_{o}^{p \ast} - D} + \mathcal{E}_{o}^{\ast}
\end{eqnarray}
So our goal is to show
{\small
\begin{eqnarray}
    \gamma_{a} D + \mathcal{E}_{o}^{\ast} & \ge & \underset{W_{a}^{p}, \ \mathcal{U}_{0}}{\sup} \left\{ \gamma_{a}W_{a}^{p} + \frac{\gamma_{a}}{\mathcal{B}_{o}^{\ast}} \mathbb{E}_{\mathcal{U}_{0}}\left[ \min{\left(
W_{T}-W_{a}^{p}, 0 \right)} \right] + \mathbb{E}_{\mathcal{U}_{0}} \left[ W_{T} \right] \right\}
\nonumber
\\
\frac{\gamma_{o} \mathcal{B}_{o}^{\ast} D }{W_{o}^{p \ast} - D} + \mathcal{E}_{o}^{\ast} & \ge & - \underset{W_{a}^{p}, \ \mathcal{U}_{0}}{\inf} \left\{  -\frac{\gamma_{o} \mathcal{B}_{o}^{\ast} W_{a}^{p} }{W_{o}^{p \ast} - D}  + \gamma_{o} \mathbb{E}_{\mathcal{U}_{0}}\left[ \max{\left(
\frac{W_{a}^{p} - D  }{W_{o}^{p \ast} - D} - \frac{W_{T} - D  }{W_{o}^{p \ast} - D}, 0 \right)} \right] - \mathbb{E}_{\mathcal{U}_{0}} \left[ W_{T} \right] \right\}
\nonumber
\\
-\frac{\gamma_{o} \mathcal{B}_{o}^{\ast} D }{W_{o}^{p \ast} - D} - \mathcal{E}_{o}^{\ast} & \le &  \underset{W_{a}^{p}, \ \mathcal{U}_{0}}{\inf} \left\{  -\frac{\gamma_{o} \mathcal{B}_{o}^{\ast} W_{a}^{p} }{W_{o}^{p \ast} - D}  + \gamma_{o} \mathbb{E}_{\mathcal{U}_{0}}\left[ \max{\left(
\frac{W_{a}^{p} - D  }{W_{o}^{p \ast} - D} - \frac{W_{T} - D  }{W_{o}^{p \ast} - D}, 0 \right)} \right] - \mathbb{E}_{\mathcal{U}_{0}} \left[ W_{T} \right] \right\}
\nonumber
\\
\frac{\gamma_{o} \mathcal{B}_{o}^{\ast} \left(W_{o}^{p \ast} - D\right) }{W_{o}^{p \ast} - D} - \mathcal{E}_{o}^{\ast} & \le &  \underset{W_{a}^{p}, \ \mathcal{U}_{0}}{\inf} \left\{  \frac{\gamma_{o} \mathcal{B}_{o}^{\ast} \left(W_{o}^{p \ast} - W_{a}^{p}\right) }{W_{o}^{p \ast} - D}  + \gamma_{o} \mathbb{E}_{\mathcal{U}_{0}}\left[ \max{\left(
\frac{W_{a}^{p} - D  }{W_{o}^{p \ast} - D} - \frac{W_{T} - D  }{W_{o}^{p \ast} - D}, 0 \right)} \right] - \mathbb{E}_{\mathcal{U}_{0}} \left[ W_{T} \right] \right\}
\nonumber
\\
\gamma_{o} \mathcal{B}_{o}^{\ast} - \mathcal{E}_{o}^{\ast} & \le &  \underset{W_{a}^{p}, \ \mathcal{U}_{0}}{\inf} \left\{  \frac{\gamma_{o} \mathcal{B}_{o}^{\ast} \left(W_{o}^{p \ast} - W_{a}^{p}\right) }{W_{o}^{p \ast} - D}  + \gamma_{o} \mathbb{E}_{\mathcal{U}_{0}}\left[ \max{\left(
\frac{W_{a}^{p} - D  }{W_{o}^{p \ast} - D} - \frac{W_{T} - D  }{W_{o}^{p \ast} - D}, 0 \right)} \right] - \mathbb{E}_{\mathcal{U}_{0}} \left[ W_{T} \right] \right\} \label{eq: 5.28 ge target}
\end{eqnarray}
}
Now we focus on rhs of~\eqref{eq: 5.28 ge target},
\begin{eqnarray*}
    & & \underset{W_{a}^{p}, \ \mathcal{U}_{0}}{\inf} \left\{  \frac{\gamma_{o} \mathcal{B}_{o}^{\ast} \left( {\color{black}W_{o}^{p \ast}} - W_{a}^{p} \right) }{W_{o}^{p \ast} - D}  + \gamma_{o} \mathbb{E}_{\mathcal{U}_{0}}\left[ \max{\left(
\frac{W_{a}^{p} - D  }{ {\color{black}W_{o}^{p \ast}} - D} - \frac{W_{T} - D  }{ {\color{black}W_{o}^{p \ast}} - D}, 0 \right)} \right] - \mathbb{E}_{\mathcal{U}_{0}} \left[ W_{T} \right] \right\}
\nonumber
\\
& \overset{(i)}{\ge} & \underset{W_{a}^{p}, \ \mathcal{U}_{0}}{\inf} \left\{  \frac{\gamma_{o} \mathcal{B}_{o}^{\ast} \left( {\color{black}W_{a}^{p }} - W_{a}^{p} \right) }{W_{o}^{p \ast} - D}  + \gamma_{o} \mathbb{E}_{\mathcal{U}_{0}}\left[ \max{\left(
\frac{W_{a}^{p} - D  }{ {\color{black}W_{a}^{p }} - D} - \frac{W_{T} - D  }{ {\color{black}W_{o}^{p \ast}} - D}, 0 \right)} \right] - \mathbb{E}_{\mathcal{U}_{0}} \left[ W_{T} \right] \right\}
\nonumber
\\
& = & \underset{ \mathcal{U}_{0}}{\inf} \left\{   \gamma_{o} \mathbb{E}_{\mathcal{U}_{0}}\left[ \max{\left(
1 - \frac{W_{T} - D  }{ {\color{black}W_{o}^{p \ast}} - D}, 0 \right)} \right] - \mathbb{E}_{\mathcal{U}_{0}} \left[ W_{T} \right] \right\}
\nonumber
\\
& \overset{(ii)}{\ge} & \underset{ W_{o}^{p} > D, \ \mathcal{U}_{0}}{\inf} \left\{   \gamma_{o} \mathbb{E}_{\mathcal{U}_{0}}\left[ \max{\left(
1 - \frac{W_{T} - D  }{ {\color{black}W_{o}^{p }} - D}, 0 \right)} \right] - \mathbb{E}_{\mathcal{U}_{0}} \left[ W_{T} \right] \right\}
\nonumber
\\
& \overset{(iii)}{=} & \gamma_{o} \mathcal{B}_{o}^{\ast} - \mathcal{E}_{o}^{\ast}.
\nonumber
\end{eqnarray*}
Here, (i) arises because $W_{o}^{p \ast}$ is replaced by a free variable $W_{a}^{p}$;
(ii) results from replacing $W_{o}^{p \ast}$ by a free variable $W_{o}^{p} > D$. Recall $W_{o}^{p \ast}>D$, so a free variable $W_{o}^{p} > D$ indeed provides a larger range;
(iii) is due to the fact that $W_{o}^{p \ast}$ and $\mathcal{U}_{0, o}^{\ast}$ are the optimal solution associated with $\left( \mathcal{B}_{o}^{\ast}, \mathcal{E}_{o}^{\ast} \right)$. This concludes the proof.

%
\def\r{\right}
\def\l{\left}
\def\f{\frac}
\def\b{\textbf}
\def\s{\sqrt}
\def\i{\infty}
\def\t{\text}
\section{\texorpdfstring{An infinite series representation of $\boldsymbol{g(y,\Delta t)}$}{Transition density}}
\label{app:num_scheme}
Under the log-transformation, the jump-diffusion dynamics in~\eqref{eq: dSt}—where $\log(\xi)$ follows an asymmetric double-exponential distribution with density given in~\eqref{eq: pdf for Kou model}—admit an infinite series representation for the conditional transition density $g(y, \Delta t)$~\cite{zhang2024monotone}[Corollary 3.1], as presented below. Define
\[
\alpha = \frac{\sigma^2}{2} \, \Delta t, \quad
\beta = \left(\mu - \lambda \kappa - \frac{\sigma^2}{2}\right) \Delta t, \quad
\theta = -\lambda \Delta t.
\]
Then
\EQ
\label{eq:series}
g(y, \Delta t) = g_0(y, \Delta t) + \sum_{\ell=1}^{\infty} \Delta g_\ell(y, \Delta t),
\EN
where $g_0(y, \Delta t) = \frac{\exp\left(\theta - \frac{(\beta + y)^2}{4\alpha}\right)}{\sqrt{4\pi \alpha}}$, and the remaining terms $\Delta g_\ell(y, \Delta t)$ are given by:
\begin{align}
\label{eq:g_EJ_double_exp}
g_{\ell}(y, \Delta t) =
&\f{e^{\theta}}{\s{4 \pi \alpha}}
\f{\left(\lambda \Delta t\right)^{\ell}}{\ell!}
\l[
\sum_{k=1}^\ell \, Q_1^{\ell,k} \,
\left(\eta_1 \, \sqrt{2\alpha}\right)^{k} \,
\e^{\eta_1 \, \left({{\beta+y}}\right) + \eta_1^2 \alpha} \,
\mathrm{Hh}_{k-1}\left(\eta_1 \sqrt{2\alpha} + \f{{{\beta+y}}}{\sqrt{2\alpha}}\right)
\r. \nonumber\\
& \qquad \qquad \qquad
\l.
+ \sum_{k=1}^\ell \, Q_2^{\ell, k} \,
\left(\eta_2 \, \sqrt{2\alpha}\right)^{k} \,
\e^{-\eta_2 \, \left({{\beta+y}} \right) + \eta_2^2 \alpha} \,
\mathrm{Hh}_{k-1}\left(\eta_2 \sqrt{2\alpha} - \f{{{\beta+y}}}{\sqrt{2\alpha}} \right)
\r],
\end{align}
where $Q_1^{\ell, k}$, $Q_2^{\ell, k}$ and $\mathrm{Hh}_{k}$ are defined as follows
\begin{align}
\label{eq:PQ}
Q_1^{\ell, k} &= \sum_{i=k}^{\ell-1} \binom{\ell-k-1}{i-k} \binom{\ell}{i}
\left(\f{\eta_1}{\eta_1 + \eta_2}\right)^{i-k}\left(\f{\eta_2}{\eta_1 + \eta_2}\right)^{\ell-i} p_{up}^{i}(1-p_{up})^{\ell-i},
\quad
1 \le k \le \ell-1,
\nonumber\\
Q_2^{\ell, k} &= \sum_{i=k}^{\ell-1} \binom{\ell-k-1}{i-k} \binom{\ell}{i}
\left(\f{\eta_1}{\eta_1 + \eta_2}\right)^{\ell-i}\left(\f{\eta_2}{\eta_1 + \eta_2}\right)^{i-k} p_{up}^{\ell-i}(1-p_{up})^{i},
\quad
1 \le k \le \ell-1,
\end{align}
where $Q_1^{\ell,\ell}=p_{up}^\ell$ and $Q_2^{\ell,\ell}=(1-p_{up})^\ell$, and
\begin{align}
\label{eq:Hhk}
Hh_{k} (x) = \f{1}{k !} \int_{x}^{\i} \left(z-x\right)^{k} e^{-\f{1}{2} z^2} d z,
\t{ with }
Hh_{-1} (x) = e^{-x^2/2}, \text{ and }
Hh_{0} (x) = \s{2\pi}\text{NorCDF}(-x).
\end{align}
Here, NorCDF denotes CDF of standard normal distribution $\mathcal{N}(0,1)$.
For this case, we note that function $Hh_{\ell}(\cdot)$ can be evaluated very efficiently using the standard normal density function
and standard normal distribution function via the three-term recursion \cite{AbramowitzStegun1972}
\EQA
k \, Hh_{k}(x) = Hh_{k-2}(x) - x Hh_{k-1}(x),
\quad k \ge 1.
\nonumber
\ENA
For computational purposes, the infinite series in~\eqref{eq:series} is truncated after $N_g$ terms, yielding the approximation $g(y, \Delta t; N_g)$. As $N_g \to \infty$, the approximation becomes exact. For finite $N_g$, however, truncation introduces an error. As shown in~\cite{zhang2024monotone}, this truncation error can be bounded by
\EQ
\label{eq:Kbound}
\left| g(y, \Delta t) - g(y, \Delta t; N_g) \right|
\le
\f{\l(\lambda\Delta t\r)^{N_g+1}}{(N_g+1)!}\,
\f{1}{\s{2 \pi \sigma^2 \Delta t}}.
\EN
Therefore, from \eqref{eq:Kbound}, as $N_g \to \infty$, we have $\f{\l(\lambda\Delta t\r)^{N_g+1}}{(N_g+1)!} \rightarrow 0$,
and the truncated approximation $g(y, \Delta t; N_g)$ converges to the exact density $g(y, \Delta t)$. For a given tolerance $\epsilon_g > 0$, one can choose $N_g$ such that the truncation error satisfies $\l| g(y, \Delta t) - g(y, \Delta t; N_g) \r| < \epsilon_g$.
This can be ensured by requiring $N_g$ to satisfy
$\ds \f{\l(\lambda\Delta t\r)^{N_g+1}}{(N_g+1)!} \le \epsilon_g \s{2 \pi \sigma^2 \Delta t}$.
It then follows that if $\epsilon_g = \mathcal{O}(h)$, we can choose $N_g = \mathcal{O}(\ln(h^{-1}))$ as $h \to 0$.

\section{Proof of Theorem~\ref{thm:convergence}}
\label{app:convergence}
We now present the convergence proof for the pre-commitment Mean–bPoE/CVaR formulations.
\subsection{Value function convergence}
We first establish the convergence of the value function:
$\ds \lim_{h \to 0} \bigl| V_h^p(s_0, b_0, t_0^-) - V^p(s_0, b_0, t_0^-) \bigr| = 0$.
The proof consists of two parts: we first establish a convergence bound for the inner optimization problem at a fixed discrete threshold, and then analyse the convergence of the outer optimization at time $t_0$ as  $h\to 0$.
\subsubsection{Inner optimization} For now, we fix $W^p \in \Gamma$ and consider the inner control problem, which corresponds to a portfolio optimization over the bounded domain $\Omega \times \mathcal{T}\cup\{T\}$. We show that the scheme for this inner problem satisfies three key properties: $\ell_\infty$-stability, monotonicity, and local consistency. To facilitate the analysis, we reformulate both the pre-commitment Mean–bPoE/CVaR localized problem in Definition~\ref{def:glwb} and the numerical scheme~\eqref{eq:terminal_sym}–\eqref{eq:incept_pre}, including both interior and boundary equations, each in a unified operator form.
We first start with the localized problem formulation.

For notational convenience,  let $\hat{x} = (y, b, W^p) \in \Omega \times \Gamma$ and $\hat{x}^m = (\hat{x}, t_m)$ with $t_m \in \mathcal{T} \cup \{t_M = T\}$. For fixed $W^p \in \Gamma$, we denote by $\widehat{V}^p(W^p, t_{m+1})$ the function $\widehat{V}^p(y, b, W^p, t_{m+1})$. When clear from context, we write $\widehat{V}^p(y, b, \cdot, t_{m+1})$ or $\widehat{V}^p(\cdot, t_{m+1})$ in place of the full expression.

Given a state $(y, b)$ and rebalancing control $u_m \in \mathcal{Z}$, we recall the definitions $y^+ = y^+(y, b, u_m)$ and $b^+ = b^+(y, b, u_m)$ from~\eqref{eq:sbplus}. Including both interior and boundary equations, we now express the pre-commitment Mean–bPoE/CVaR problem at the reference point $\hat{x}^m$ (with fixed $W^p$) using the operator $\mathcal{D}(\cdot)$ as follows:
\begin{align}
\label{eq:con_mean}
 \widehat{V}^p\big(\hat{x}^m\big)   = \mathcal{D}\left(\hat{x}^m,\,  \widehat{V}^p(\cdot, t_{m+1}^-)\right)
 &=
 \begin{cases}
 \underset{u_m \in \mathcal{Z}}{\inf}\, \widehat{V}\left(y^+(y, b, u_m),b^+(y, b, u_m), \cdot,  t_m^+\right),
 & \text{Mean--bPoE},
 \\
 \underset{u_m \in \mathcal{Z}}{\sup}\, \widehat{V}\left(y^+(y, b, u_m),b^+(y, b, u_m), \cdot,  t_m^+\right),
 & \text{Mean--CVaR}.
  \end{cases}
\end{align}
In both cases, the quantity $\widehat{V} \in \{\widehat{V}_o^p,\, \widehat{V}_a^p\}$ denotes the appropriate value function, and is given as follows:
$\widehat{V}\left(y^+(y, b, u_m),b^+(y, b, u_m), \cdot,  t_m^+\right) = \ldots$
\begin{linenomath}
\begin{subequations}\label{eq:scheme_v_all*}
\begin{empheq}[left={
\ldots
= \empheqlbrace}]{alignat=6}
&
\widehat{V}\left(y^+(y, be^{r\Delta t}, u_m), b^+(y, be^{r\Delta t}, u_m), \cdot, t_{m+1}^-\right)
&&
&& \quad\quad (y, b) \in \Omega_{y_{\min}} ,
\label{eq:scheme_v_all_1*}
\\
&
\int_{y_{\min}^{\dagger}}^{y_{\max}^{\dagger}}
\widehat{V}\left(y^+(y', be^{r\Delta t}, u_m), b^+(y, be^{r\Delta t}, u_m), \cdot, t_{m+1}^-\right)\, g(y - y', \Delta  t )~dy'
&&
&&\quad\quad (y, b) \in  \Omega_{\myin},
\label{eq:scheme_v*}
\\
&
e^{\mu\Delta t} ~ \widehat{V}\left(y^+(y, be^{r\Delta t}, u_m), b^+(y, be^{r\Delta t}, u_m),
\cdot, t_{m+1}^-\right)
&&
&&\quad\quad (y, b) \in \Omega_{y_{\max}}.
\label{eq:scheme_v_all_3*}
\end{empheq}
\end{subequations}
\end{linenomath}
Let $\Omega^h \times \Gamma^h$ denote the computational grid parameterized by $h$.
Let $\Omega_{\myin}^h$ denote the interior sub-grid, and $\Omega_{y_{\min}}^h$, $\Omega_{y_{\max}}^h$ the boundary sub-grids in $y$.
Including both interior and boundary equations, we now write the numerical scheme at the reference  node $\hat{x}_{n,j}^{k,m} = (y_n, b_j, W_k, t_m) \in \Omega^h \times \Gamma^h \times \{t_m\}$ in operator form via $\mathcal{D}_h(\cdot)$ as follows:
\begin{align}
\label{eq:con_mean_disc}
 \widehat{V}_h^p\big(\hat{x}_{n, j}^{k, m}\big)   &= \mathcal{D}_h\left(\hat{x}_{n, j}^{k, m},\,  \left\{\widehat{V}_h^p\big(\hat{x}_{l, j}^{k, (m+1)-}\big)\right\}_{l=-N^{\dagger}/2}^{N^{\dagger}/2}\right)
 \nonumber
 \\
 &=
 \begin{cases}
 \underset{\{u_i\}}{\min}\, \widehat{V}_{h}\left(y^+(y_n, b_j, u_i),b^+(y_n, b_j, u_i), \cdot,  t_m^+\right),
 & \text{Mean--bPoE},
 \\
 \underset{\{u_i\}}{\max}\, \widehat{V}_{h}\left(y^+(y_n, b_j, u_i),b^+(y_n, b_j, u_i), \cdot,  t_m^+\right),
 & \text{Mean--CVaR}.
  \end{cases}
\end{align}
In both cases, the quantity $\widehat{V}_h \in \{\widehat{V}_{o, h}^p,\, \widehat{V}_{a, h}^p\}$ denotes the appropriate discrete approximation, and is defined as
$\widehat{V}_h\left(y^+(y_n, b_j, u_i),b^+(y_n, b_j, u_i), \cdot,  t_m^+\right) = \ldots$
\begin{linenomath}
\begin{subequations}\label{eq:scheme_v_all*_disc}
\begin{empheq}[left={
\ldots
= \empheqlbrace}]{alignat=6}
&
\widehat{V}_h\left(y^+(y_n, b_je^{r\Delta t}, u_i), b^+(y_n, b_je^{r\Delta t}, u_i), \cdot, t_{m+1}^-\right)
&&
&& \quad (y_n, b_j) \in \Omega_{y_{\min}}^h ,
\label{eq:scheme_v_all_1*}
\\
&
 \sum_{l=-N^{\dagger}/2}^{N^{\dagger}/2} \varphi_{l}~
g(y_n - y_l, \Delta t)
\widehat{V}_h\left(y^+(y_l, b_je^{r\Delta t}, u_i), b^+(y_n, b_je^{r\Delta t}, u_i), \cdot, t_{m+1}^-\right)
&&
&&\quad (y_n, b_j) \in  \Omega_{\myin}^h,
\label{eq:scheme_v*_disc}
\\
&
e^{\mu\Delta t}  \widehat{V}_h\left(y^+(y_n, b_je^{r\Delta t}, u_i), b^+(y_n, b_je^{r\Delta t}, u_i),
\cdot, t_{m+1}^-\right)
&&
&&\quad(y_n, b_j) \in \Omega_{y_{\max}}^h,
\label{eq:scheme_v_all_3*_disc}
\end{empheq}
\end{subequations}
\end{linenomath}
where $y^+(y', b', u_m)$ and $b^+ = b^+(y', b', u_m)$ are  given by \eqref{eq:sbplus}.

Let $\epsilon_g$ denote the truncation error bound in approximating the transition density function $g(\cdot, \Delta t)$ by its $N_g$-term series expansion;
that is, $|g(\cdot, \Delta t) - g(\cdot, \Delta t; N_g)| < \epsilon_g$.
Suppose linear interpolation is used for the intervention step.
Also suppose that as $h \to 0$, $N_y^\dagger, N_b, N_w, N_u, N_g, \to \infty$.
Following the general framework of convergence proofs for numerical approximations of
solutions in the stochastic control setting (see \cite{kushner2001numerical, barles-souganidis:1991}), we show that our scheme, for each fixed $W_k \in \Gamma^h$,
is $\ell_\infty$-stable, locally consistent, and monotone. Specifically,
\begin{itemize}[noitemsep, topsep=0pt, leftmargin=*]
\item $\ell_\infty$-stability: the scheme  \eqref{eq:terminal_sym}--\eqref{eq:incept_pre} for $V_h^p(\cdot)$ satisfies the bound:
    \EQ
    \label{eq:stab}
    \sup_{h > 0} \big\| \widehat{V}_h^p(W_k, t_m) \big\|_{\infty} < \infty, \quad \forall t_m \in \mathcal{T} \cup \{T\}, \text{ as  $h \to 0$}.
    \EN
    Here, we have $\big\|  \widehat{V}_h^p(W_k, t_m) \big\|_{\infty} = \max_{n, j} |\widehat{V}_h^p\big(\hat{x}_{n,j}^{k,m}\big)|$, where $(y_n, b_j, W_k) \in \Omega^h\times \Gamma^h$ and $W_k$ is fixed.

\item Local consistency: For any smooth test function $\phi \in \mathcal{C}^{\infty}(\Omega \cup \Omega_{b_{\max}} \times [0,T])$ for a sufficiently small $h, \chi$,
\begin{align}
\label{eq:pointwise_consistency}
\mathcal{D}_h\left(\hat{x}_{n,j}^{k,m}, \left\{\widehat{V}_h^p\big(
\hat{x}_{l,j}^{k,(m+1)-}\big) +\chi  \right\}_{l=-N^{\dagger}/2}^{N^{\dagger}/2}\right)
=
\mathcal{D}\left(\hat{x}_{n,j}^{k,m}, \widehat{V}^p(\cdot, t_{m+1}^-)\right)
+
\mathcal{E}\left(\hat{x}_{n,j}^{k,m}, \epsilon_g, h\right)
+ \mathcal{O}\left(\chi+ h\right),
\end{align}
where $\mathcal{E}(\hat{x}_{n,j}^{k,m}, \epsilon_g, h) \to 0 $ as $\epsilon_g, h \to 0$.


\item Monotonicity: the numerical scheme $\mathcal{D}_h(\cdot)$ satisfies
\EQA
\label{eq:mon}
\mathcal{D}_h\left(\hat{x}_{n,j}^{k,m},
\left\{\varphi_{l,j}^{k,{m+1}}\right\}_{l=-N^{\dagger}/2}^{N^{\dagger}/2}\right)
\leq
\mathcal{D}_h\left(\hat{x}_{n,j}^{k,m},
\left\{\psi_{l,j}^{k,m+1})\right\}_{l=-N^{\dagger}/2}^{N^{\dagger}/2}\right)
\ENA
for bounded discrete data sets $\left\{\varphi_{l,j}^{k,{m+1}}\right\}$ and $\left\{\psi_{l,j}^{k,m+1}\right\}$ having $\left\{\varphi_{l,j}^{k,{m+1}}\right\} \leq \left\{\psi_{l,j}^{k,m+1}\right\}$, where the inequality is understood in the component-wise sense.
\end{itemize}
The $\ell_\infty$-stability follows from a maximum-principle argument, which applies since $\Omega$ is bounded, together with the monotonicity of linear interpolation, which is preserved under the scheme’s $\min/\max$ operations. Monotonicity of the scheme itself follows from this same interpolation structure. Local consistency is established via standard interpolation error analysis, Taylor expansion of the smooth test function, truncation error from the infinite series, and compactness of the admissible control set $\mathcal{Z}$. For further details, we refer the reader to~\cite{zhang2024monotone}, which develops similar techniques in the context of pre-commitment Mean–Variance optimization.

\textbf{Convergence bound for the inner problem:}
We now state a convergence bound for the inner problem using $h$-indexed notation to reflect the role of mesh refinement. A generic grid point in $\Omega_{\myin}^h$ is denoted by $(y_h, b_h)$, and we write $W_h \in \Gamma^h$ to emphasize the $h$-dependence of the discretized threshold. The following bound holds for the inner optimization problem at fixed $W_h \in \Gamma^h$:
\begin{center}
\begin{minipage}{0.95\textwidth}
\emph{Let $(y', b') \in \Omega$ be arbitrary. Suppose that linear interpolation is used for intervention step. Under the assumption that $\epsilon_g \to 0$ as $h \to 0$, we have: for each $m \in \{M-1, \ldots, 0\}$ and any sequence $\{(y_h, b_h)\}$ with $(y_h, b_h) \in \Omega_{\myin}^h$ and $(y_h, b_h) \to (y', b')$ as $h \to 0$,
\begin{align}
\label{eq:con_inner}
\bigl| \widehat{V}_h^p(y_h, b_h, W_h, t_m) - \widehat{V}^p(y', b', W_h, t_m) \bigr| \le \chi_h^m, \quad \chi_h^m \text{ is bounded } \forall h > 0 \text{ and }
\chi_h^m \to 0 \text{ as } h \to 0.
\end{align}
}
\end{minipage}
\end{center}
The result in~\eqref{eq:con_inner} can be proved using $\ell_\infty$-stability~\eqref{eq:stab}, local consistency~\eqref{eq:pointwise_consistency}, monotonicity~\eqref{eq:mon}, and backward induction on $m$.

\subsubsection{\texorpdfstring{Outer optimization at time $\boldsymbol{t_0}$ }{Outer optimization}}
\label{ssc:outer}
At time $t_0$, for each fixed $W^p\in \Gamma$ (resp.\ $W_h \in \Gamma^h$), we define
\EQ
\label{eq:FFh}
F(W^p) \quad \big(\text{resp.\ } F_h(W_h)\big) :=
 \begin{cases}
    ~~\widehat{V}_o^{p}(y_0,b_0,W_o^p,t_0) \quad \big(\text{resp.\ } ~~~\widehat{V}_{o, h}^{p}(y_0,b_0, W_h,t_0)\big)
      &\text{Mean--bPoE},
      \\
    -\widehat{V}_a^{p}(y_0,b_0,W^p_a,t_0) \quad \big(\text{resp.\ } -\widehat{V}_{a, h}^{p}(y_0,b_0,W_h,t_0)\big)
      &\text{Mean--CVaR}.
  \end{cases}
\EN
We include the minus sign in the Mean–CVaR case to unify both problems under a minimization framework for the outer optimization.
We note that the function $F$ is continuous in $W^p$.
Moreover, since the problem is localized, $F(W^p)$ remains finite on $\Gamma$.
For any $W^p \in \Gamma$, the triangle inequality yields:
\begin{equation}
\label{eq:tri_F}
\bigl|F_h(W_h) - F(W^p)\bigr|
\le
\bigl|F_h(W_h) - F(W_h)\bigr|
+
\bigl|F(W_h) - F(W^{p})\bigr|
\overset{\text{(i)}}{\le} \chi_h  +
\bigl|F(W_h) - F(W^{p})\bigr|
\overset{\text{(ii)}}{\xrightarrow[h \to 0]{}}
 0.
\end{equation}
Here, (i) follows from the definitions of $F$ and $F_h$, and from the bound~\eqref{eq:con_inner} applied at $t_0$, which gives $\bigl|F_h(W_h) - F(W_h)\bigr| \le \chi_h$, where $\chi_h \to 0$ as $h \to 0$.
The limit in (ii) is due to the denseness of $\Gamma^h$ in $\Gamma$ as $h \to 0$:
for any $W^p \in \Gamma$, there exists a sequence $\{W_h\}$ with $W_h \in \Gamma^h$ and $W_h \to W^p$ as $h \to 0$, and since $F$ is continuous, it follows that $\lim_{h \to 0} F(W_h) = F(W^p)$.

We now show
\EQ
\label{eq:conv_v}
\displaystyle \lim_{h \to 0}~\min_{W_h \in \Gamma_h } F_h(W_h) =  \min_{W^p \in \Gamma} F(W^p)
\EN
by establishing the corresponding upper and lower bounds:
\begin{align}
\label{eq:limsup}
\limsup_{\,h\to 0} \min_{W_h\in \Gamma^h} F_h(W_h) \le
\min_{W^p\in \Gamma} F(W^p), \quad \text{and}\quad
\liminf_{\,h\to 0} \min_{W_h\in \Gamma^h} F_h(W_h) \ge
\min_{W^p\in \Gamma} F(W^p).
\end{align}
Once~\eqref{eq:conv_v} is established, the convergence of the value function follows:
\[
\lim_{h\to 0}
     \bigl\lvert
       V_h^p(s_0,b_0,t_0^-)
       -
       V^p(s_0,b_0,t_0^-)
     \bigr\rvert =
     0.
\]
\\
\noindent We first show the $\limsup$ portion of \eqref{eq:limsup}.
Pick an arbitrary $\varepsilon> 0$ and  let $\rho =\min_{W\in\Gamma}F(W)$.
By definition of $\min$, there exists $W'\in \Gamma$ with $F(W')\le \rho+\tfrac{\varepsilon}{2}$. Next, because $\Gamma^h\to\Gamma$ as $h \to 0$,
choose $W_h\in \Gamma^h$ so $\lvert W_h -W'\rvert \to 0$. By continuity,
$F(W_h)\le F(W')+\tfrac{\varepsilon}{2}$.
Recalling the inner problem convergence bound \eqref{eq:con_inner},
we have  $\lvert F_h(W_h)-F(W_h)\rvert\le \chi_h$, where $\chi_h \to 0$
as $h \to 0$. Then
\[
\min_{W_i\in \Gamma^h}F_h(W_i)
     \le
     F_h(W_h)
     \le
     F(W_h)+\chi_h
     \le
     F(W') + \tfrac{\varepsilon}{2} + \chi_h
     \le
     \rho + \varepsilon + \chi_h.
\]
For sufficiently small $h$, $\chi_h\le \varepsilon$. Thus $\min_{W_i\in \Gamma^h}F_h(W_i)\le \rho +2\varepsilon$. Hence, we obtain the $\limsup$ portion
\[
\limsup_{h\to0}\min_{W_i\in\Gamma^h}F_h(W_i)\le \min_{W\in\Gamma}F(W).
\]
For the $\liminf$ portion of \eqref{eq:limsup}, we fix any $W_h\in\Gamma^h$.
By the inner problem convergence bound \eqref{eq:con_inner}, we have $F_h(W_h) \ge  F(W_h)- \chi_h$, and since $\min_{W\in\Gamma}F(W)\le F(W_h)$, we have
$F_h(W_h) \ge \min_{W \in \Gamma}F(W)- \chi_h$, from which, taking the minimum over $W_h$ gives $\min_{W_h\in \Gamma^h}\,F_h(W_k)  \ge  \min_{W\in \Gamma}F(W) -\chi_h$.
Taking limit both sides of the above $h\to0$, noting $\chi_h\to 0$, gives
the $\liminf$ portion:
\[
 \liminf_{\,h\to0}
     \min_{W_h\in\Gamma^h} F_h(W_h)
     \ge
     \min_{W\in \Gamma}F(W).
 \]

\subsection{\texorpdfstring{Convergence to optimal threshold $\boldsymbol{W^{p\ast}}$}{Optimal threshold}}
We denote by $W_{h}^{\ast}$ the computed optimal threshold
which minimizes the discrete objective $F_h(\cdot)$ defined in \eqref{eq:FFh}, i.e.\
$F_h\bigl(W_{h}^{\ast}\bigr) =   \min_{\,W_h\in \Gamma^h}\,F_h(W_h)$.
For a fixed $h > 0$, as noted earlier, the value $W_h^\ast$ is obtained by exhaustive search over the discrete set $\Gamma^h$, ensuring that the discrete global minimiser is found.
Since $\Gamma$ is compact, and the discrete thresholds $W_{h}^{\ast}\in \Gamma^h\subset \Gamma$ all lie in this compact set, any infinite sequence $\{W_{h}^{\ast}\}$ has a subsequence $\{W_{h_k}^{\ast}\}$ convergent to some $W^{\ast} \in \Gamma$, i.e.\
$W_{h_k}^{\ast} \;\xrightarrow[k\to\infty]{} W^{\ast} \in \Gamma$.

We now show that $W^{\ast}$ is indeed the optimal threshold of the pre-commitment Mean-bPoE/CVaR problem:
\[
W^{\ast} \in \argmin_{W \in \Gamma} F(W), \quad
\text{$F(\cdot)$ defined in \eqref{eq:FFh}}.
\]
Because $\Gamma^h\subset\Gamma$ and by the convergence bound for the inner problem in
\eqref{eq:con_inner}, for any fixed $W \in \Gamma^h$, we have  $\bigl|F_h(W)-F(W)\bigr|\to 0$ as $h \to 0$. This implies
$\bigl|F_h(W_{h_k}^{\ast}) - F(W_{h_k}^{\ast})\bigr| \;\xrightarrow[k\to\infty]{}0$.
Since $W_{h_k}^{\ast}\to W^{\ast}$ and $F$ is continuous,
   we get $F(W_{h_k}^{\ast})\to F(W^{\ast})$. Thus, we also have
$F_h(W_{h_k}^{\ast})\to F(W^{\ast})$ as $k \to \infty$.
Putting together, we arrive at
\EQ
\label{eq:FhtoF}
F_h\bigl(W_{h_k}^{\ast}\bigr) \;\xrightarrow[k\to \infty]{} F\bigl(W^\ast\bigr).
\EN
By convergence of the outer optimization problem at time $t_0$
given in \eqref{eq:conv_v}:
\EQ
\label{eq:minFhtoF}
\min_{W_h\in\Gamma^h}F_h(W_h)
\;\xrightarrow[h\to 0]{}
\min_{W^p \in \Gamma}F(W^p).
\EN
From \eqref{eq:minFhtoF},
we conclude that
\[
F_h\bigl(W_{h_k}^{\ast}\bigr) =
  \min_{\,W\in\Gamma^{h_k}}F_{h_k}(W)
  \;\;\xrightarrow[k\to\infty]{}
  \min_{\,W^p\in \Gamma}F(W^p),
\]
which, together with \eqref{eq:FhtoF}, gives $F(W^\ast) = \min_{\,W^p\in \Gamma}F(W^p)$,
meaning
$W^\ast\in \argmin_{\,W\in \Gamma}F(W)$, as wanted.
This concludes the proof.

\section{Proof of Theorem~\ref{thm:convergence_tc}}
\label{app:convergence_tc}
For the purpose of analysis, the time-consistent Mean–bPoE/CVaR formulations and their numerical schemes can be expressed using unified operators similar to $\mathcal{D}$ and $\mathcal{D}_h$ defined in~\eqref{eq:con_mean}–\eqref{eq:con_mean_disc} for the pre-commitment case. 
{\sifin{Specifically, including both interior and boundary equations,
we write the time–consistent Mean–bPoE/CVaR problem at the reference point
$\hat{x}^{m} = (y,b,W^c,t_m)$ using the operator $\mathcal{G}(\cdot)$ as follows
\begin{align}
\label{eq:con_mean_c}
\widehat{V}^c\bigl(\hat{x}^{m}\bigr)
   &= \mathcal{G}\bigl(\hat{x}^{m},\,\widehat{V}^c(\cdot,t_{m+1}^-)\bigr),
\nonumber
\\
&=
\begin{cases}
\displaystyle
\underset{u_m \in \mathcal{Z}}{\inf}\;
\widehat{V}_o^c\Bigl(
   y^+(y,b,u_m),\,
   b^+(y,b,u_m),\,
   W^c,\,
   t_m^+
\Bigr),
  & \text{Mean--bPoE},\\[6pt]
\displaystyle
\underset{u_m \in \mathcal{Z}}{\sup}\;
\widehat{V}_a^c\Bigl(
   y^+(y,b,u_m),\,
   b^+(y,b,u_m),\,
   W^c,\,
   t_m^+
\Bigr),
  & \text{Mean--CVaR},
\end{cases}
\end{align}
and $y^+(\cdot)$, $b^+(\cdot)$ are given in~\eqref{eq:sbplus}.}}

{\sifin{
The numerical scheme at the reference node
$\hat{x}_{n,j}^{k,m} = (y_n,b_j,W_k,t_m)
 \in \Omega^h \times \Gamma^h \times \{t_m\}$,
including both interior and boundary conditions,
is written as
\begin{align}
\label{eq:con_mean_disc_c}
\widehat{V}_h^c\bigl(\hat{x}_{n,j}^{k,m}\bigr)
&=
\mathcal{G}_h\Bigl(
   \hat{x}_{n,j}^{k,m},\,
   \bigl\{\widehat{V}_h^c\bigl(\hat{x}_{l,j}^{k,(m+1)-}\bigr)\bigr\}_{l=-N^{\dagger}/2}^{N^{\dagger}/2}
\Bigr)
\nonumber
\\
&=
\begin{cases}
\displaystyle
\underset{\{u_i\}}{\min}\;
\widehat{V}_{o,h}^c\Bigl(
   y^+(y_n, b_j, u_i),\,
   b^+(y_n, b_j, u_i),\,
   W_k,\,
   t_m^+
\Bigr),
  & \text{Mean--bPoE},\\[6pt]
\displaystyle
\underset{\{u_i\}}{\max}\;
\widehat{V}_{a,h}^c\Bigl(
   y^+(y_n, b_j, u_i),\,
   b^+(y_n, b_j, u_i),\,
   W_k,\,
   t_m^+
\Bigr),
  & \text{Mean--CVaR},
\end{cases}
\end{align}
where, again, $y^+(\cdot)$ and $b^+(\cdot)$ are defined via~\eqref{eq:sbplus}.}}


We then establish $\ell_\infty$-stability, consistency, and monotonicity of our scheme. Specifically,
\begin{itemize}[noitemsep, topsep=0pt, leftmargin=*]
\item $\ell_\infty$-stability: the scheme  \eqref{eq:terminal_c_sym}–\eqref{eq:incept_pre_c_sym} for $V_h^c(\cdot)$ satisfies the bound:
    \EQ
    \label{eq:stab_c}
    \sup_{h > 0} \big\| \widehat{V}_h^c(W_k, t_m) \big\|_{\infty} < \infty, \quad \forall t_m \in \mathcal{T} \cup \{T\}, \text{ as  $h \to 0$}.
    \EN

\item Local consistency: For any smooth test function $\phi \in \mathcal{C}^{\infty}(\Omega \cup \Omega_{b_{\max}} \times [0,T])$ for a sufficiently small $h, \chi$,
\begin{align}
\label{eq:pointwise_consistency_c}
\mathcal{G}_h\left(\hat{x}_{n,j}^{k,m}, \left\{\widehat{V}_h^c\big(
\hat{x}_{l,j}^{k,(m+1)-}\big) +\chi  \right\}_{l=-N^{\dagger}/2}^{N^{\dagger}/2}\right)
=
\mathcal{G}\left(\hat{x}_{n,j}^{k,m}, \widehat{V}^c(\cdot, t_{m+1}^-)\right)
+
\mathcal{E}'\left(\hat{x}_{n,j}^{k,m}, \epsilon_g, h\right)
+ \mathcal{O}\left(\chi+ h\right),
\end{align}
where $\mathcal{E}'(\hat{x}_{n,j}^{k,m}, \epsilon_g, h) \to 0 $ as $\epsilon_g, h \to 0$.

\item Monotonicity: the numerical scheme $\mathcal{G}_h(\cdot)$ satisfies
\EQA
\label{eq:mon}
\mathcal{G}_h\left(\hat{x}_{n,j}^{k,m},
\left\{\varphi_{l,j}^{k,{m+1}}\right\}_{l=-N^{\dagger}/2}^{N^{\dagger}/2}\right)
\leq
\mathcal{G}_h\left(\hat{x}_{n,j}^{k,m},
\left\{\psi_{l,j}^{k,m+1})\right\}_{l=-N^{\dagger}/2}^{N^{\dagger}/2}\right)
\ENA
for bounded discrete data sets $\left\{\varphi_{l,j}^{k,{m+1}}\right\}$ and $\left\{\psi_{l,j}^{k,m+1}\right\}$ having $\left\{\varphi_{l,j}^{k,{m+1}}\right\} \leq \left\{\psi_{l,j}^{k,m+1}\right\}$, where the inequality is understood in the component-wise sense.
\end{itemize}
\noindent
Proofs of these properties follow the same steps as in the pre-commitment case (see Appendix~\ref{app:convergence} for details), as the re-optimization of thresholds at each state–time node introduces no additional difficulty. This is because the $\min/\max$ operation over thresholds is simply a monotone transformation of real values.

By $\ell_{\infty}$-stability, local consistency, and monotonicity, we obtain a convergence bound analogous to the pre-commitment case:
\begin{center}
\begin{minipage}{0.95\textwidth}
\emph{Let $(y', b') \in \Omega$ be arbitrary. Suppose that linear interpolation is used for intervention step. Under the assumption that $\epsilon_g \to 0$ as $h \to 0$, we have:
for a fixed discretized threshold value $W_h \in \Gamma^h$, and each $m \in \{M-1, \ldots, 0\}$ and any sequence $\{(y_h, b_h)\}$ with $(y_h, b_h) \in \Omega_{\myin}^h$ and $(y_h, b_h) \to (y', b')$ as $h \to 0$,
\EQ
\label{eq:bound_c}
\bigl|\widehat{V}_h^c(y_h,b_h,W_h,t_m) - \widehat{V}^c(y',b',W_h,t_m)\bigr| \le \chi_h^m, \quad \chi_h^m \text{ is bounded } \forall h > 0 \text{ and }
\chi_h^m \to 0 \text{ as } h \to 0.
\EN
}
\end{minipage}
\end{center}
\paragraph{Value function convergence.}
To handle the threshold re-optimization (unifying both Mean-bPoE and Mean-CVaR problems), we follow the same approach as in the pre-commitment case (Section~\ref{ssc:outer}) but do it locally at each node: fix $(y', b', t_m) \in \Omega_{\myin} \times \{t_m\}$ (resp.\ $(y_h, b_h, t_m) \in \Omega_{\myin}^h \times \{t_m\}$) and let the threshold $W^c\in \Gamma$ (resp.\ $W_h\in \Gamma^h$) we define $F(y', b', W^c, t_m)$ (resp.\ $F_h(y_h, b_h, W_h, t_m)$) as follows
{\sifin{
\begin{equation}
\label{eq:FFh_c}
F(\cdot)~~
\;\bigl(\text{resp. }F_h(\cdot)\bigr)
:=
\begin{cases}
\widehat{V}_o^{c}(y',b',W^c,t_m) &
\bigl(\text{resp. }\widehat{V}_{o,h}^{c}(y_h,b_h,W_h,t_m)\bigr),
  \qquad ~~~~\text{Mean--bPoE},\\
-\widehat{V}_a^{c}(y',b',W^c,t_m) &
\bigl(\text{resp. }-\widehat{V}_{a,h}^{c}(y_h,b_h,W_h,t_m)\bigr),
  \qquad \text{Mean--CVaR}.
\end{cases}
\end{equation}
Here, $\widehat{V}^c$ and $\widehat{V}_h^c$ are defined by in
Definition~\ref{def:glwb_c} and \eqref{eq:terminal_c_sym}–\eqref{eq:whVhc}, respectively. In the Mean--CVaR case, the minus sign allows us to treat both problems
as minimizations in $W^c$.
}}

We note that
\EQ
\label{eq:VVh}
\widehat{V}^c(y',b',t_m)
  \;:=\;
  \min_{\,W^c \in \Gamma}\,F\bigl(y',b',W^c,t_m\bigr),
  \quad\text{and}\quad
  \widehat{V}_h^c(y_h,b_h,t_m)
  \;:=\;
  \min_{\,W_h \in \Gamma^h}\,F_h\bigl(y_h,b_h,W_h,t_m\bigr).
\EN
We now establish
\EQ
\label{eq:goal_time_con}
  \lim_{h\to 0}
  \Bigl\lvert
    \min_{W_h\in \Gamma^h}F_h(y_h,b_h,W_h,t_m)
    \;-\;
    \min_{W^c\in \Gamma}F(y',b',W^c,t_m)
  \Bigr\rvert
  =0
  \quad
  \text{as }(y_h,b_h)\to(y',b')
\EN
using the $\limsup$ and $\liminf$ arguments: $(y_h,b_h)\to(y',b')$ as $h \to 0$, we have
\EQ
\label{eq:limsupinf_c}
\begin{aligned}
\limsup_{h\to 0}
    \min_{W_h\in \Gamma^h}F_h(y_h,b_h,W_h,t_m)
  &\le
  \min_{W^c\in \Gamma}F(y',b',W^c,t_m),
\\
\liminf_{h\to0}
    \min_{W_h\in \Gamma^h}F_h(y_h,b_h,W_h,t_m)
  &\ge
  \min_{W^c\in \Gamma}F(y',b',W^c,t_m).
\end{aligned}
\EN
Once \eqref{eq:goal_time_con} is established,
recalling \eqref{eq:VVh}, we conclude that
\[
  \lim_{h \to 0}\,
  \Bigl\lvert
    \widehat{V}_h^c(y_h,b_h,t_m)
    \;-\;
    \widehat{V}^c(y',b',t_m)
  \Bigr\rvert
  \;=\;
  0
  \quad\text{as }(y_h,b_h)\to(y',b').
\]
By exponentiating $s = e^{y}$ and letting $(s_h,b_h)\to(s',b')$,
we recover the value function convergence in the original coordinates
stated in \eqref{eq:time_consistent_v_conv}.

We first show the $\limsup$ portion of \eqref{eq:limsupinf_c}. This is similar to the
proof for the pre-commitment case.
By the definition of the continuous minimum,
\EQ
\label{eq:We}
\exists W^\ast\in\Gamma \text{ such that }
     F\bigl(y',b',W^\ast,t_m\bigr) \le \min_{W^c\in \Gamma}F(y',b',W^c,t_m) +\tfrac{\varepsilon}{2}.
\EN
Because $\Gamma^h\subset \Gamma$ is dense as $h\to0$, choose $W_h\in \Gamma^h$ such that $\lvert W_h -W^\ast\rvert\to 0$.
By a triangle inequality and for sufficiently small $h$, we have
\begin{align*}
\bigl|
  F_h(y_h,b_h,W_h,t_m)
  -
  F(y',b',W^\ast,t_m)
\bigr|
&\le
  \bigl|
    F_h(y_h,b_h,W_h,t_m)
    -F(y',b',W_h,t_m)
  \bigr|
\\
& \quad +
  \bigl|
    F(y',b',W_h,t_m)
    -F(y',b',W^\ast,t_m)
  \bigr|
 \overset{(i)}{\le} \tfrac{\varepsilon}{4} + \tfrac{\varepsilon}{4} = \tfrac{\varepsilon}{2}
\end{align*}
Here, in (i),  by \eqref{eq:bound_c}, for sufficiently small $h$, we can bound the first term by $\tfrac{\varepsilon}{4}$, while due to the continuity of $F(\cdot)$ in the threshold argument, the second term can be bounded above by $\tfrac{\varepsilon}{4}$.
Thus, we have
\[
F_h(y_h,b_h,W_h,t_m)
     \le
     F\bigl(y',b',W^\ast,t_m\bigr) + \tfrac{\varepsilon}{2}.
\]
Therefore,
\EQ
\label{eq:minF_c}
     \min_{\,W_h\in \Gamma^h}F_h(y_h,b_h,W_h,t_m)
     \;\le\;
     F_h(y_h,b_h,W_h,t_m)
     \;\le\;
     F\bigl(y',b',W^\ast,t_m\bigr) +\tfrac{\varepsilon}{2}.
\EN
Using~\eqref{eq:We}–\eqref{eq:minF_c}, we obtain
\[
\min_{W_h \in \Gamma^h} F_h(y_h, b_h, W_h, t_m) \le \min_{W^c \in \Gamma} F(y', b', W^c, t_m) + \varepsilon,
\]
from which taking $\limsup_{h \to 0}$ yields the upper bound in~\eqref{eq:limsupinf_c}.

For the $\liminf$ portion of \eqref{eq:limsupinf_c}, we fix any $\widetilde{W}_h\in\Gamma^h$. By the bound \eqref{eq:bound_c}, we have
\[
F_h(y_h,b_h,\widetilde{W}_h,t_m) \ge F(y',b',\widetilde{W}_h,t_m) -\chi_h,
\]
for some small $\chi_h\to0$. Since $\Gamma^h\subseteq \Gamma$, we have
\[
\min_{W_h\in \Gamma^h}F_h(y_h,b_h,W_h,t_m)
\ge
\min_{W^c\in \Gamma} F(y',b',W^c,t_m) - \chi_h.
\]
Taking $\liminf$ as $h\to0$ and $(y_h,b_h)\to(y',b')$ yields the
 $\liminf$ portion of \eqref{eq:limsupinf_c}, as wanted.

\paragraph{Optimal threshold convergence.}
Let $W_h^{c\ast}$ be an optimal threshold for the discrete problem
$F_h(\cdot)$ at time $t_m$, i.e.\
\[
F_h\bigl(y_h,b_h,W_h^{c\ast},t_m\bigr) = \min_{W_h \in \Gamma^h}\,
  F_h\bigl(y_h,b_h,W_h,t_m\bigr).
\]
Since $\Gamma$ is a compact subset of $\mathbb{R}$
and $W_h^{c\ast}\in \Gamma^h\subset \Gamma$,
any infinite sequence $\{W_h^{c\ast}\}$
has a subsequence $\{W_{h_k}^{c\ast}\}$ that converges to
some $W^{\ast} \in \Gamma$ as $h \to 0$.
We claim that $W^{\ast}$ is indeed an optimal threshold for the exact time-consistent problem:
\EQ
\label{eq:claimW}
  W^{\ast} \in  \argmin_{W^c \,\in\, \Gamma}\,F\bigl(y',b',W^c,t_m\bigr).
\EN
To prove \eqref{eq:claimW}, we first note that, due to local consistency in $(y,b)$,
similar to pre-commitment case,  for any $W\in\Gamma$, consider $\{(y_h, b_h)\}$ with $(y_h, b_h) \in \Omega_{\myin}^h$ and $(y_h, b_h) \to (y', b')$ as $h \to 0$, we have $\bigl|\,F_h(y_h,b_h,W,t_m)\, - \, F(y',b',W,t_m)\bigr|\to 0$ as $h\to0$.
Combining this with $W_{h_k}^{c\ast} \;\xrightarrow[k\to \infty]{} W^{\ast}$, we have
\EQ
\label{eq:iq}
  F_h\bigl(y_{h_k},b_{h_k},W_{h_k}^{c\ast},t_m\bigr)
  \;\longrightarrow\;
  F\bigl(y',b',W^{\ast},t_m\bigr) \text{ as $k \to \infty$}.
\EN
In addition, by \eqref{eq:goal_time_con}, we have
\EQ
\label{eq:iqq}
  \min_{W_h\in \Gamma^h}F_h(y_h,b_h,W_h,t_m)
    \;\xrightarrow[h\to 0]{}
    \min_{W^c\in \Gamma}F(y',b',W^c,t_m).
\EN
Consequently, by \eqref{eq:iq}-\eqref{eq:iqq}, we conclude
\[
  F\bigl(y',b',W^{\ast},t_m\bigr)
  \;=\;
  \lim_{k\to\infty}\,F_h\bigl(y_{h_k},b_{h_k},W_{h_k}^{c\ast},t_m\bigr)
  \;=\;
  \min_{W^c\in \Gamma}F\bigl(y',b',W^c,t_m\bigr).
\]
Thus $W^{\ast}\in \argmin_{W^c\in\Gamma}\,\{\,F(y',b',W^c,t_m)\}$,
proving that $W^{\ast}$ is an exact optimal threshold of the continuous
problem. This concludes the proof.

{
\bibliographystyle{plain}
\setlength{\bibsep}{0pt plus 0.3ex}
\small
\bibliography{paperbib_cc}
}

\end{document}